\documentclass[preprint,
3p,
times loads txfonts
]{elsarticle}

\usepackage{amsmath}
\usepackage{amsthm}
\usepackage{amssymb}
\usepackage{color}
\usepackage{xcolor}
\usepackage{cases}
\usepackage{enumitem}              
\usepackage{mathrsfs}
\usepackage{mathtools}
\usepackage[OMLmathsfit]{isomath}  
\usepackage{graphicx}
\usepackage{grffile}               
\usepackage[T1]{fontenc}
\usepackage{lmodern}

\usepackage[cmintegrals]{newpxmath}  
\usepackage{bm}                      

\newtheorem{proposition}{Proposition}

\journal{Elsevier}

\begin{document}

\begin{frontmatter}

%
%
\title{A conditional formulation of the Vlasov-Amp{\`e}re Equations: A conservative, positivity, asymptotic, and Gauss law preserving scheme}

\author[label1]{William T. Taitano\corref{cor1}}\ead{taitano@lanl.gov}
\author[label2,label1]{Joshua W. Burby}\ead{joshua.burby@austin.utexas.edu}
\author[label3]{Alex Alekseenko}\ead{alexander.alekseenko@csun.edu}
\address[label1]{Applied Mathematics and Plasma Physics Group, Los Alamos National Laboratory, Los Alamos, New Mexico 87545.}
\address[label2]{Physics Department, University of Texas at Austin, Austin, Texas 78712.}
\address[label3]{Mathematics Department, California State University Northridge, Northridge, California, 91330.}

%
%

\begin{abstract}
We propose a novel reformulation of the Vlasov-Amp{\`e}re equations for plasmas. This reformulation exposes discrete symmetries to achieve simultaneous conservation of mass, momentum, and energy; preservation of Gauss's law involution; positivity of the distribution function; and quasi-neutral asymptotics. Our approach relies on transforming variables and coordinates, leading to a coupled system of a modified Vlasov equation and the associated moment-field equations. The modified Vlasov equation evolves a conditional distribution function that excludes information on mass, momentum, and energy densities. The mass, momentum, and energy density, in turn, are evolved using moment equations, in which discrete symmetries, conservation laws, and involution constraints are enforced. The reformulation is compatible with a recent slow-manifold reduction technique, which separates the fast electron time scales and simplifies handling the asymptotic quasi-neutrality limit within the easier to solve moment-field subsystem. Using this new formulation, we develop a numerical approach for the reduced 1D1V subsystem that, for the first time, simultaneously satisfies the important physical constraints while preserving the quasi-neutrality asymptotic limit. The advantages of our approach are demonstrated through canonical electrostatic test problems, including the multiscale ion acoustic shock wave problem.
\end{abstract}

\begin{keyword}
Conditional Vlasov Amp{\`e}re \sep fully conservative \sep asymptotic-preserving \sep positivity-preserving \sep Gauss law preserving \sep fast slow split \sep slow manifold reduction

\vspace{.5\baselineskip}
\end{keyword}

\end{frontmatter}

%
%
\section{Introduction\label{sec:introduction}}

The coupled Vlasov-Amp{\`e}re (VA) equations describe the dynamic evolution of a plasma particle distribution function (PDF), $f$, in the position-velocity phase space $\left\{ \vec{x},\vec{v}\right\} \in\mathbb{R}^{3}\times\mathbb{R}^{3}$ coupled to the electric field, $\vec{E}$. The system has wide applications in the modeling of laboratory astrophysics experiments, electric propulsion systems, semiconductor manufacturing processes, and laboratory fusion systems, among others. The VA system satisfies important physical properties such as: 
\begin{enumerate}
\item conservation of mass, momentum, and energy,
\item Gauss law involution constraint,
\item quasi-neutral asymptotics, and
\item positivity of PDF. 
\end{enumerate}
The simultaneous preservation of these continuum properties in the discrete has been of intense research focus by both the physics and applied mathematics community. Notable approaches include the asymptotic preserving techniques \cite{degond_jcp_2010_qn_ap_va, jin_ap_riv_mat_univ_parma_2012_ap_review,ye_jcp_2024_ap_vfp,coughlin_jcp_2022_low_rank_ap_vafp}, structure preserving discretization scheme \cite{cheng_jsnm_2014_vlasov_ampere, shiroto_jcp_2019_sp_vm}, micro-macro-decomposition (MMD) approaches \cite{liu_cmp_mmd_2004,gamba_jcp_mmd_2019}, spectral element approaches \cite{filbet_camc_2022_dg_hermite,manzini_jcp_2016_legendre_fourier_vp}, and discrete nonlinear constraint schemes \cite{taitano_jcp_2015_conservative_vfpa_1, anderson_jcp_2020_adaptive_grid}. To our knowledge, no approach has been developed that satisfies all four requirements, and the fundamental challenge appears to be in the seemingly incompatible constraints that arise across the different requirements. For example, when numerically solving the Vlasov equation, a hyperbolic PDE, numerical dissipation is required to stabilize the solution and ensure positivity of the PDF through a flux limiter or artificial viscosity schemes. However, due to the Hamiltonian nature of the the Vlasov-Amp{\`e}re system, such artificial dissipation usually leads to momentum and energy conservation errors. Similarly, when dynamically evolving the Amp{\`e}re equation to ensure the discrete energy conservation theorem, the Gauss law must be satisfied as an involution constraint. To satisfy the Gauss law, a finite differencing Yee scheme \cite{yee_ieee_1966_yee_scheme} is often employed, where the electric field and charge density are staggered. However, to recover the asymptotic quasi-neutral limit, it is paramount that the ambipolarity condition, $\nabla_{x}\cdot\vec{j}=0$ is satisfied; here, $\vec{j} = \sum_{\alpha}^{N_{s}}q_{\alpha} \int_{\mathbb{R}^{3}} d^{3}v\vec{v} f_{\alpha}$ is the current (collocated with $\vec{E}$). This limit can be trivially recovered in the limit in which the normalized permittivity constant vanishes, $\epsilon_{0}\rightarrow0$, and the Amp{\`e}re equation becomes an algebraic equation for the current, $\vec{j}=\vec{0}.$ However, once $f$ is defined as staggered with respect to $\vec{E}$ (to ensure that the charge density, $\rho = \sum_{\alpha}^{N_{s}} q_{\alpha} \int_{\mathbb{R}^{3}} d^{3}xf$ is staggered with $\vec{E}$), the current must be reconstructed at the collocation point of $\vec{E}$ through some interpolation operation. This very interpolation operator is non-invertible and introduces a non-trivial null space that destroys the quasi-neutrality asymptotic preserving property. Conversely, if one defines $f$ as collocated with $\vec{E}$, quasi-neutrality can be recovered, but at the cost of violating the Gauss law.

We propose a novel, yet consistent, reformulation of the Vlasov-Amp{\`e}re equation that evolves the conditional distribution function, ${\cal F} \left(t,\vec{w} |\vec{x}\right) = \frac{v_{th}^{3}f}{\int_{\mathbb{R}^{3}} d^{3}vf}$, defined in the scaled peculiar velocity coordinate $\vec{w}=\frac{\vec{v}-\vec{u}}{v_{th}}$; we note that this choice of coordinate was similarly considered in the classical works of Grad and Chapman-Enskog, as well as in Ref. \cite{taitano_cpc_2021_ifp_code} which used a similar coordinate transformation to model the Vlasov-Fokker-Planck equation in inertial confinement fusion capsules that undergo large temperature and flow variations, and in Ref. \cite{filbet_jcp_2013_rescaling_velocity} for modeling the Boltzmann like equation for granular media. In the new coordinate system, the first three $w$ moments of ${\cal F}\left(t,\vec{w}|\vec{x}\right)$ are pointwise invariants, $\int_{\mathbb{R}^{3}}d^{3}w\vec{\phi}_w{\cal F} = \vec{C} \; \forall \left(t,\vec{x}\right)$, where $\vec{\phi}_w=\left\{ 1,\vec{w},w^{2}\right\} $ and $\vec{C} \in \mathbb{R}^5$ is a vector of constants. As such, the evolution of the original hydrodynamic moments of the PDF --the mass, momentum, and energy densities-- is divorced from the evolution of ${\cal F}$. Instead, these moments are evolved in the hydrodynamic moment subsystem. Together, the coupled modified Vlasov equation for ${\cal F}$ and the moment field equations recover the original Vlasov-Amp{\`e}re dynamics. In the new formulation, the viscous stress and heat flux computed from ${\cal F}$ act as a closure for the moment system, while the solution of the moment system parameterizes the dynamics of ${\cal F}$. The advantages of the new reformulation are that the conservation of mass, momentum, and energy; the involution constraint; and quasi-neutrality asymptotic limits are isolated in the moment-field subsystem, and the choice of discretization for evolving ${\cal F}$ can be dictated primarily to satisfy positivity, in which many choices are readily available (e.g. flux limiters, TVD schemes, and flux correction techniques) while not affecting the other properties. In the moment-field subsystem, we extend a conservative and shock-capturing staggered formulation, originally developed by the fluid dynamics community \cite{ansanay_ijnmf_2011_l2_stab_ns_staggered,herbin_2012_esaim_proceedings_staggered,goudon_fvca8_2017_staggered_euler_scheme_staggered,berthelin_2015_hal_barotropic_euler_model_staggered,brunel_jsc_2024_3d_staggered_general_mesh_staggered,grapas_smai_2016_unconditionally_stable_staggered_p_corr_ns_staggered}, to electrostatic plasmas. For the first time in the literature, we demonstrate that our new formulation allows us to simultaneously meet properties 1 to 4 in discrete form for the electrostatic Vlasov-Amp{\`e}re system. In addition, it accurately resolves shocks.

The remainder of the manuscript is organized as follows. In Section~\ref{sec:va_sys_and_qn_asymptotics} we introduce the Vlasov-Amp{\`e}re equation and the specific non-dimensionalization procedure that will aid in our downstream discussion. In Section~\ref{sec:conditional_formulation}, we derive the reformulation of the Vlasov equation in the new coordinate system and in terms of ${\cal F}$. In Section \ref{sec:fast_slow_formulation}, we discuss the modification of the coupled conditional distribution function and moment-field system that naturally exposes the relevant continuum symmetries to recover the quasi-neutral limit. In Section \ref{sec:comparible_discretization}, we discuss the specific spatio-temporal discretizations employed to solve the reformulated system that satisfy requirements 1-4. In Section~\ref{sec:solver} we discuss the details of our solver. In Section~\ref{sec:numerical_results}, we demonstrate the advertised properties of our new formulation on a canonical set of test problems including the challenging multiscale ion-acoustic shock wave problem, and conclude in Section~\ref{sec:conclusions}. 

%
%
\section{Vlasov-Amp{\`e}re System and the quasi-neutral Limit \label{sec:va_sys_and_qn_asymptotics}}

The Vlasov equation describes the dynamical evolution of the plasma particle distribution function (PDF) in the position-velocity phase-space over time and is given as: 
\begin{equation}
    \frac{\partial f_{\alpha}}{\partial t}+\vec{v}\cdot\nabla_{x}f_{\alpha}+\frac{q_{\alpha}}{m_{\alpha}}\left(\vec{E}+\vec{v}\times\vec{B}\right)\cdot\nabla_{v}f_{\alpha}=0.
    \label{eq:vlasov}
\end{equation}
Here, $f_{\alpha}\left(\vec{x},\vec{v},t\right)$ is the particle PDF for species $\alpha$, $\vec{x}\in\mathbb{R}^{3}$ is the particle position, $\vec{v}\in\mathbb{R}^{3}$ is the particle velocity, $t \in \left[0,t_{max} \right] \subset \mathbb{R}_{+}$ is time, $t_{max}$ is the upper temporal bound, $\nabla_{x}=\left\{ \partial_{x} , \partial_{y} , \partial_{z}\right\} $ is the spatial gradient operator, $\nabla_{v}=\left\{ \partial_{v_{x}} , \partial_{v_{y}} , \partial_{v_{z}} \right\} $ is the velocity space gradient operator, $\vec{E}\left(\vec{x},t\right)$ is the electric field, $\vec{B}\left(\vec{x},t\right)$ is the magnetic
field, and $q_{\alpha}$ ($m_{\alpha}$) is the charge (mass) of the species. In this study, we focus on the 1D1V, electrostatic limit, 
\begin{equation}
    \frac{\partial f_{\alpha}}{\partial t} + v \frac{\partial f_{\alpha}}{\partial x} + \frac{q_{\alpha}}{m_{\alpha}} E \frac{\partial f_{\alpha}}{\partial v}=0,
    \label{eq:es_vlasov}
\end{equation}
where the evolution of the electric field is governed by the Amp{\`e}re's equation as: 
\begin{equation}
   \epsilon_{0}\frac{\partial E}{\partial t} + j =0.
   \label{eq:ampere}
\end{equation}
Here,  $\epsilon_{0}$ is the permittivity constant of the vacuum, $j = \sum_{\alpha}^{N_{s}} q_{\alpha} n_{\alpha} u_{\alpha} = \sum_{\alpha}^{N_{s}} q_{\alpha} \left<v,f_{\alpha}\right>_{v}$ is the total current,  $\left<A,B\right>_{v}=\int_{\mathbb{R}} dv AB$ is a shorthand notation for the velocity-space inner-product operation between functions $A$ and $B$, and $N_{s}$ is the total number of plasma species in the system. The two equations are nonlinearly coupled through the electrostatic acceleration term and the current in the Vlasov and Amp{\`e}re equations, respectively. 

We normalize the equations by introducing a set of reference values in terms of reference number density, $n^{*}$, electron mass, $m^{*}=m_{e}$, proton charge, $q^{*}=q_{p}$, macroscopic time, $\tau^{*}$, macroscopic length, $L^{*}$, temperature, $k_{b}T^{*}$, speed, $u^{*}=\sqrt{\frac{k_{b}T^{*}}{m^{*}}}$, distribution function, $f^{*}=n^{*}/u^{*^{3}}$, and electric field, $E^{*}=\frac{m^{*}L^{*}}{\tau^{*^{2}}q^{*}}$. Defining the normalized quantities as $\widehat{f}_{\alpha}=f_{\alpha}/f^{*}$, $\widehat{x}=x/L^{*}$, $\widehat{v}=v/u^{*}$, $\widehat{m}_{\alpha}=m_{\alpha}/m^{*}$, $\widehat{q}_{\alpha} = q_{\alpha}/q_{p}$, and $\widehat{E}=E/E^{*}$ and substituting these values into Eqs. \eqref{eq:es_vlasov} and \eqref{eq:ampere}, one obtains: 
\begin{equation}
    \frac{\partial\widehat{f}_{\alpha}}{\partial\widehat{t}}+\widehat{v}\frac{\partial \widehat{f}_{\alpha}}{\partial\widehat{x}}+\frac{\widehat{q}_{\alpha}}{\widehat{m}_{\alpha}} \widehat{E} \frac{\partial \widehat{f}_{\alpha}}{\partial \widehat{v}} = 0.
    \label{eq:es_vlasov_normalized}
\end{equation}
\begin{equation}
    \epsilon^{2}\frac{\partial \widehat{E}}{\partial\widehat{t}}+\widehat{j} = 0.
    \label{eq:ampere_normalized}
\end{equation}
Here, $\epsilon = \omega_{p,e}^{*^{-1}} \tau^{*}$ is the small parameter that defines the relative importance of charge separation and $\omega_{p,e}^{*} = \sqrt{\frac{n^{*}q_{p}^{2}}{\epsilon_{0}m_{e}}}$ is the reference electron plasma frequency. From this point forward, we assume all quantities are normalized and drop the hats for brevity. We further introduce the concept of a non-Amp{\`e}rean current, defined as
\begin{equation}
    \label{eq:nonamperean_current}
    \widetilde{j} = j/\epsilon = \sum_{\alpha}^{N_{s}} q_{\alpha} n_{\alpha} u_{\alpha}/ \epsilon
\end{equation}
and redefine the system of equations as: 
\begin{equation}
    \left({\cal V}_{\alpha}\right):\frac{\partial f_{\alpha}}{\partial t}
    +
    v \frac{\partial f_{\alpha}}{\partial x}
    +
    \frac{q_{\alpha}}{m_{\alpha}}E \frac{\partial f_{\alpha}}{\partial v}=0,
    \label{eq:vlasov_normalized}
\end{equation}
and 
\begin{equation}
    \left({\cal A}\right):\epsilon\frac{\partial E}{\partial t}+\widetilde{j} = 0.
    \label{eq:ampere_nonamperean_current}
\end{equation}
The motivation for this transformation will be described shortly. Furthermore, we show the classic relationship of the Gauss law serving as an involution constraint of the normalized and transformed Amp{\`e}re equation by taking the divergence of \eqref{eq:ampere_nonamperean_current},
\begin{equation}
    \epsilon\frac{\partial}{\partial t} \frac{\partial E}{\partial x} 
    +
    \frac{\partial \widetilde{j}}{\partial x}=0.
    \label{eq:div_of_ampere}
\end{equation}
By using the charge continuity equation, $\sum_{\alpha}^{N_{s}}q_{\alpha}\left<1,{\cal V}_{\alpha}\right>_{v} = \partial_{t} \rho + \epsilon \partial_x \widetilde{j} = 0$, with $\rho = \sum_{\alpha}^{N_{s}} q_{\alpha} \left<1,f_{\alpha}\right>_{v}$ the total charge density, we obtain: 
\begin{equation}
    \frac{\partial}{\partial t}\left(\epsilon^{2} \frac{\partial E}{\partial x} - \rho\right)=0.
    \label{eq:gauss_law}
\end{equation}
Finally, using the Gauss law, we can express the electron number density as,
\begin{equation}
    \label{eq:electron_number_density_in_terms_of_gauss_law}
    n_e = q^{-1}_{e} \left(\partial_x E - \sum^{N_i}_{i} q_i n_i \right),
\end{equation}
and recover the quasi-neutrality limit as  $\epsilon \rightarrow 0$,
\begin{equation}
    \label{eq:quasi_neutrality}
    n_e = - q^{-1}_e  \sum^{N_i}_{i} q_i n_i.
\end{equation}
Here, $n_i$ is the number density of the $i^{th}$ ion species and $N_i$ is the total number of ions in the system.

%
%
\section{Conditional Formulation\label{sec:conditional_formulation}}

We propose a novel reformulation of the underlying kinetic equation based on a conditional distribution function, which we define as 
\begin{equation}
    {\cal F}_{\alpha}\left(w,t|x\right) \equiv {\cal F}_{\alpha} \left(x,w,t\right)=\frac{f_{\alpha}\left(x,w v_{th,\alpha}+u_{\alpha},t\right)}{n_{\alpha}\left(x,t\right)}.
    \label{eq:conditional_distribution}
\end{equation}
Here, ${\cal F}_{\alpha}$ is a local probability distribution for species $\alpha$, conditional to the given $x$ in the scaled peculiar velocity space, $w=\frac{v-u_{\alpha}}{v_{th,\alpha}}$, $u_{\alpha} \left( x,t \right) = \gamma_{\alpha} / n_{\alpha}$ is the mean velocity, $\gamma_{\alpha} \left( x , t \right) = \left< v , f_{\alpha} \right>_{v}$ is the particle flux, $v_{th}\left(x,t\right)=\sqrt{\frac{2T}{m}}$ is the thermal speed, and $T_{\alpha}\left(x,t\right)=m_{\alpha}\frac{\left<\left|v-u_{\alpha}\right|^{2},f_{\alpha}\right>_{v}}{n_{\alpha}}$ is the temperature. Hereon, unless otherwise specified, we drop the explicit species indices for brevity.

We transform the Vlasov equation \eqref{eq:vlasov_normalized} from the original Cartesian phase-space coordinate system, $\vec{Z}=\left\{ t,x,v\right\} $, with the associated phase-space advection velocity, ${\vec{\dot{Z}}}=\left\{ 1,v,\frac{q}{m}E \right\} $:
\begin{equation}
    \nabla_{Z}\cdot\left(\vec{\dot{Z}}f\right)=0,\label{eq:vlasov_in_abstract_coordinate}
\end{equation}
into the new coordinate, $\vec{{\cal Z}}=\left\{ t,x,w\right\} $, using the standard machineries of curvilinear coordinates as: 
\begin{equation}
J^{-1}\nabla_{{\cal Z}}\cdot\left(J\vec{\dot{{\cal Z}}}f\right)=0.\label{eq:vlasov_in_new_abstract_coordinate}
\end{equation}
Here, $\nabla_{Z}=\left\{ \partial_{t},\partial_{x},\partial_{v} \right\} $,
$\nabla_{{\cal Z}} = \left\{ \partial_{t},\partial_{x},\partial_{w}\right\} $,
$J=\det\left|{\mathbb{J}}^{\left(ij\right)}\right|^{-1}=v_{th}$
is the Jacobian of transformation, ${\mathbb{J}}^{\left(ij\right)}=\frac{\partial\vec{{\cal Z}}}{\partial\vec{Z}}$ is the contravariant Jacobian matrix, and 
\begin{eqnarray}
    \label{eq:new_{a}dvection_{v}ector}
    \vec{\dot{{\cal Z}}}={\mathbb{J}}^{\left(ij\right)}\cdot\vec{\dot{Z}}=\nonumber \\
    \left\{ 1,v_{th}w + u,-\frac{1}{v_{th}}\left[w\frac{\partial v_{th}}{\partial t}+\frac{\partial u}{\partial t}+\left(v_{th}w+u\right) \frac{\partial}{\partial x} (v_{th}w)+\left(v_{th} w + u \right) \frac{\partial u}{\partial x}-\frac{q}{m} E\right]
    \right\}
\end{eqnarray}
is the transformed advection velocity. By redefining $f:=Jf$, multiplying and dividing $f$ with $n$, using the chain rule and the Leibniz rule, using the continuity equation, 
\begin{equation}
    \label{eq:continuity_0}
    \left<1,\left({\cal V}\right)\right>_{v} = 
    \frac{\partial n}{\partial t} 
    + 
    \frac{\partial \left( n u \right)}{\partial x}=0,
\end{equation}
and by using Eq. \eqref{eq:conditional_distribution}, we obtain the Vlasov equation in the new coordinate system in terms of ${\cal F}$ as: 
\begin{eqnarray}
    \label{eq:conditional_vlasov}
    \frac{\partial{\cal F}}{\partial t} + \frac{\partial }{\partial x} \left\{ \left[v_{th} w + u \right]{\cal F}\right\} -\nonumber \\
    \frac{1}{v_{th}} \frac{\partial}{\partial w}\left\{ \left[ w \frac{\partial v_{th}}{\partial t} + \frac{\partial u}{\partial t} + \left(v_{th} w+u \right) \frac{\partial}{\partial x}(v_{th} w)+\left(v_{th} w + u \right) \frac{\partial u}{\partial x} - \frac{q}{m}E\right]{\cal F}\right\} =\nonumber \\
    v_{th} w \left\{ \frac{\partial \ln{n}}{\partial x} - \frac{\partial u}{\partial x}\right\} {\cal F}.
\end{eqnarray}
Rearranging the inertial terms in the $w$ space, we rewrite Eq. \eqref{eq:conditional_vlasov} as: 
\begin{eqnarray}
    \label{eq:conditional_vlasov_rearranged}
    \frac{\partial{\cal F}}{\partial t}+\frac{\partial}{\partial x}\left\{ \left[v_{th} w + u \right]{\cal F}\right\} 
    -
    \frac{1}{v_{th}}\frac{\partial}{\partial w} \left\{ \underbrace{\left[\frac{\partial u }{\partial t} + u \frac{\partial u}{\partial x} \right]}_{a}{\cal F}\right\} - \nonumber \\
    \frac{\partial}{\partial w} \left\{ w \underbrace{\left[\frac{1}{v_{th}}\frac{\partial v_{th}}{\partial t}+\frac{u \partial_x v_{th}}{v_{th}}\right]}_{b}{\cal F}\right\} 
    -
    \frac{\partial}{\partial w} \left\{ \underbrace{\left(w \partial_x u + w^2 \partial_x v_{th} \right)}_{c}{\cal F}\right\} -\nonumber \\
    +
    \frac{\partial}{\partial w}\left\{ \frac{q E }{mv_{th}}{\cal F}\right\} 
    =
    \left\{ v_{th}\vec{w}\cdot\nabla_{x}\ln{n}-\nabla_{x}\cdot\vec{u}\right\} {\cal F}.
\end{eqnarray}
We further manipulate the Vlasov equation by realizing the relationship of the different inertial terms in the $\vec{w}$ space with the hydrodynamic equations. We project Eq. \eqref{eq:vlasov_normalized} onto $\vec{\phi}=\left\{ v, \frac{m\widetilde{w}^{2}}{2}\right\} ^{T}$ with $\widetilde{w} = w v_{th}$ to derive the familiar momentum and internal energy equations: 
\begin{equation}
    \label{eq:momentum_equation}
    \left\langle v ,\left({\cal V}\right)\right\rangle_{v}
    =
    n\left(\frac{\partial u}{\partial t} + u \frac{\partial u}{\partial x} \right) + \frac{1}{m} \frac{\partial P}{\partial x}
    -
    \frac{q}{m}nE=0,
\end{equation}
\begin{equation}
    \label{eq:internal_energy_equation}
    m \left<\frac{\widetilde{w}^{2}}{2},\left({\cal V}\right)\right>_{v}
    =
    \frac{1}{2}n\left(\frac{\partial T}{\partial t}
    +
    u\frac{\partial T}{\partial x} \right)
    +
    P \frac{\partial u}{\partial x}
    +
    \frac{\partial Q}{\partial x}
    =
    0.
\end{equation}
Here, $P=nT$ is the scalar isotropic pressure, and $Q=\frac{m}{2}\left<\widetilde{w}^{3},f\right>_{v}$ is the heat flux. From a trivial manipulation of Eqs. \eqref{eq:momentum_equation} and \eqref{eq:internal_energy_equation}, we obtain: 
\begin{equation}
    \label{eq:momentum_equation_wh_rhs}
    \left(\frac{\partial u}{\partial t} + u \frac{\partial u}{\partial x} \right)
    =
    -\frac{1}{nm} \frac{\partial P}{\partial x} +\frac{q}{m}\vec{E}
\end{equation}
and 
\begin{equation}
    \label{eq:internal_energy_equation_wh_rhs}
    \frac{1}{2}\left(\frac{\partial T}{\partial t} + u \frac{\partial T}{\partial x} \right)
    =
    -T \frac{\partial u}{\partial x}
    -
    \frac{1}{n}\frac{\partial Q}{\partial x}
    =
    0.
\end{equation}
We recognize that from the term $b$ in Eq. \eqref{eq:conditional_vlasov_rearranged}, that $v_{th}^{-1}\partial_{t}v_{th}=\frac{\partial\ln v_{th}}{\partial t} = \frac{1}{2T}\frac{\partial T}{\partial t}$ and $v_{th}^{-1} u\partial_x v_{th} = \frac{1}{2T}u \partial_x T$. We thus relate these terms to Eq. \eqref{eq:internal_energy_equation_wh_rhs} as:
\begin{equation}
    \label{eq:vth_evolution_equation}
    \frac{1}{v_{th}}\left(\frac{\partial v_{th}}{\partial t}+u \partial_x v_{th} \right)
    =
    \frac{1}{2T}\left(\frac{\partial T}{\partial t}
    +
    u \frac{\partial T}{\partial x}\right)
    =
    -\frac{1}{P} \frac{\partial Q}{\partial x} 
    -
    \frac{\partial u}{\partial x}.
\end{equation}
We substitute Eqs. \eqref{eq:momentum_equation_wh_rhs} and \eqref{eq:vth_evolution_equation} into Eq. \eqref{eq:conditional_vlasov_rearranged} to obtain: 
\begin{equation}
    \label{eq:full_3d3v_derivation_conditional_vlasov}
    \left({\cal F}\right):\frac{\partial{\cal F}}{\partial t}
    +
    \frac{\partial}{\partial x}\left(\dot{x}{\cal F}\right)
    +
    \frac{\partial}{\partial w}\left(\dot{w}{\cal F}\right)
    =
    \lambda{\cal F}
\end{equation}
where $\dot{x}=v_{th} w + u$, $\dot{w} = \left[\frac{ \partial_x P}{nmv_{th}}+\frac{w \partial_x Q}{P} - w^2 \partial_x v_{th}\right]$, and $\lambda = v_{th}w \partial_x\ln{n} - \partial_x u$. The conditional distribution function supports pointwise linear invariants for the first three $\vec{w}$ moments of ${\cal F}$, 
\begin{equation}
    \label{eq:cdf_invariance}
    \vec{{\cal M}}_{w}\equiv\left\langle \vec{\phi}_{w},{\cal F}\right\rangle _{w}=\vec{C}\;\forall\;\left(t,x\right),
\end{equation}
where $\vec{\phi}_{w}=\left\{ 1,w,w^{2}\right\} ^{T}$ and $\vec{C} = \left\{ 1, 0, \frac{1}{2} \right\}$ is a vector of constants --this is trivially shown by projecting Eq. \eqref{eq:full_3d3v_derivation_conditional_vlasov} onto $\vec{\phi}_{w}$ and showing, 
\begin{equation}
    \frac{\partial}{\partial t}\left\langle \vec{\phi}_{w},{\cal F}\right\rangle _{w}=\vec{0}.\label{eq:cdf_time_dependent_invariance}
\end{equation}
As such, the hydrodynamic moments are parameterized out of the evolution of ${\cal F}$ and are evolved through the associated fluid moment equations,
\begin{eqnarray}
    \label{eq:moment_equation_1d}
    & \frac{\partial n}{\partial t}+\frac{\partial nu}{\partial x} = 0 \nonumber \\
    \left({\cal M}\right): 
    & m\frac{\partial nu}{\partial t} + \frac{\partial}{\partial x}\left[mnu^2+P\right]-qnE=0 & .\label{eq:moment_equations}\\
    & \frac{1}{2}\left[\frac{\partial nT}{\partial t}
    +
    \frac{\partial}{\partial x}\left(u n T\right)\right]
    +
    \frac{\partial Q}{\partial x} 
    +
    P\frac{\partial u}{\partial x}
    =
    0\nonumber 
\end{eqnarray}
Here, $Q=\frac{mnv_{th}^{3}}{2}\left\langle w^{3},{\cal F}\right\rangle _{w}$ is the heat flux in terms of ${\cal F}$, respectively, which is self-consistently evaluated from the conditional distribution function to close the subsystem. The original Vlasov-Amp{\`e}re dynamics is recovered by simultaneously evolving Eqs. \eqref{eq:full_3d3v_derivation_conditional_vlasov}
and \eqref{eq:moment_equations} for each species and \eqref{eq:ampere_nonamperean_current} for the electric field; and constitutes our new dynamical system for state variables, $\left\{ {\cal F}_{1},\cdots,{\cal F}_{N_{s}},\vec{{\cal M}}_{1},\cdots,\vec{{\cal M}}_{N_{s}},E\right\}$, where $\vec{{\cal M}}_{\alpha}=\left\{ n_{\alpha},n_{\alpha}u_{\alpha},n_{\alpha}T_{\alpha}\right\}$. In this new formulation, the conservation of mass, momentum, and energy, as well as the Gauss law involution constraint, are isolated in the moment-field equations for $\vec{{\cal M}}$ for each species and $E$ and \emph{relevant symmetries are independent of the conditional distribution function}. 

Observe that in transforming from $f$ to $\mathcal{F}$ we centered, scaled, and normalized the distribution function. Previous authors have recognized the utility of velocity centering. For example, Ramos used it in various studies of kinetic magnetohydrodynamics \cite{Ramos_2008,Ramos_2015, Ramos_2016}, while Tronci and Holm-Tronci used it to formulate variational and Hamiltonian formulations of hybrid fluid-kinetic models \cite{Holm_Tronci_2012, Tronci_2013, Tronci_2020}. 

%
%
\section{Fast-Slow Formulation of the Conditional System\label{sec:fast_slow_formulation}}
The Vlasov-Amp{\`e}re system supports the quasi-neutral reduced model for vanishing $\epsilon$. In this limit fast electron time scales and short length scales are analytically removed from the system and provide algebraic relations for the electric field and electron fluid quantities through the classic Ohm's law and quasi-neutrality conditions. We show that this asymptotic limiting model is recovered in an especially transparent manner using the conditional formulation. For this purpose, we leverage the slow manifold reduction as discussed in Ref. \cite{burby_cnsnm_slow_manifold_2020}. 

A fast-slow formulation in dynamical systems is given as: 
\begin{eqnarray}
    \epsilon\frac{dy}{dt}=f_{\epsilon}\left(x,y\right),\\
    \frac{dx}{dt}=g_{\epsilon}\left(x,y\right),
\end{eqnarray}
where $x\in X$ ($y\in Y$) is the slow (fast) variable, $\left(x,y\right) \in X \times Y$, $f_{\epsilon} : X \times Y \rightarrow Y$ is the fast forcing function, $g_{\epsilon}:X\times Y\rightarrow X$ is the slow forcing function, and $D_yf_0(x,y) : Y\rightarrow Y$ is assumed invertible for each $(x,y)$ such that $f_0(x,y)=0$. Here, $D_yf_0(x,y)$ denotes the Gateaux derivative of $f_0$ with respect to the fast variable. In particular $D_yf_0(x,y)[\delta y] = (d/d\lambda)\mid_0 f(x,y+\lambda\,\delta y)$. For small $\epsilon$, every fast-slow system admits a \emph{formal slow-manifold}, defined as a function $y^*_\epsilon:X\rightarrow Y$, given as a formal power series in $\epsilon$, such that the ``graph" $S_\epsilon = \{(x,y) y = y^*_\epsilon(x)\}$ is an invariant manifold to all orders in $\epsilon$. For solutions $(x(t),y(t))$ with initial conditions on the formal slow manifold, the fast variables can be written as functionals of the slow variables, $y(t) = y^*_\epsilon(x(t))$. 

In our conditional Vlasov-Amp{\`e}re system, we exchange the electron momentum equation in favor of the non-Amp{\`e}rean current equation to cast into the fast-slow formulation. We rewrite our dynamical system as follows: 
\begin{equation}
    \label{eq:fs_cdf_equation}
    \frac{\partial{\cal F}_{\alpha}}{\partial t}
    +
    \frac{\partial}{\partial x}\left(\dot{x}_{\alpha}{\cal F}_{\alpha}\right)
    +
    \frac{\partial}{\partial w}\left(\dot{w}_{\alpha}{\cal F}_{\alpha}\right)
    =
    \lambda_{\alpha}{\cal F}_{\alpha},
\end{equation}
\begin{equation}
    \label{eq:fs_continuity_equation}
    \frac{\partial n_{\alpha}}{\partial t} 
    + 
    \frac{\partial}{\partial x} \left(n_{\alpha}u_{\alpha}\right)=0,
\end{equation}
\begin{equation}
    \label{eq:fs_momentum_equation}
    m_{i}\frac{\partial n_{i}u_{i}}{\partial t}
    +
    \frac{\partial}{\partial x}\left[m_{i}n_{i}u^2_{i} + P_{i}\right]
    -
    q_{i}n_{i}E=0,
\end{equation}
\begin{equation}
    \label{eq:fs_internal_energy_equation}
    \frac{1}{2}\left(\frac{\partial n_{\alpha} T_{\alpha}}{\partial t} 
    + 
    \frac{\partial}{\partial x}\left[u_{\alpha} n_{\alpha} T_{\alpha} \right]\right)
    +
    \frac{\partial Q_{\alpha}}{\partial x}
    +
    P_{\alpha}\frac{\partial u_{\alpha}}{\partial x} =0,
\end{equation}
\begin{equation}
    \label{eq:fs_current_equation}
    \epsilon\frac{\partial\widetilde{j}}{\partial t}
    +
    \frac{\partial}{\partial x} \left[\sum_{i}^{N_{i}}q_{i}n_{i}u^2_{i} + q_{e}n_{e}u^2_{e} +\sum_{\alpha}^{N_{s}}\frac{q_{\alpha}}{m_{\alpha}}P_{\alpha}\right]
    -
    \left( \sum^{N_s}_{\alpha} \frac{q^2_{\alpha} n_{\alpha} }{m_{\alpha}} \right) E 
    = 0,
\end{equation}
\begin{equation}
    \label{eq:fs_ampere_equation}
    \epsilon\frac{\partial E}{\partial t}+\widetilde{j} = 0.
\end{equation}
Here, $\alpha \in \left\{1,\cdots, N_s \right\}$ denotes all species including electrons, $i \in \left\{1,\cdots,N_i \right\}$ denotes only ion species, $\dot{x}_{\alpha}=v_{th,\alpha}w+u_{\alpha}$, $\dot{w}_{\alpha}=\frac{\partial_{x}P_{\alpha}}{n_{\alpha}m_{\alpha}v_{th,\alpha}}+w\frac{\partial_{x}Q_{\alpha}}{P_{\alpha}}-w^{2}\partial_{x}v_{th,\alpha}$,
$\lambda_{\alpha}=v_{th,\alpha}w\partial_{x}\ln n_{\alpha}-\partial_{x}u_{\alpha}$, $Q_{\alpha}=\frac{m_{\alpha}n_{\alpha}v_{th,\alpha}^{3}}{2} \left< w^{3},{\cal F}_{\alpha} \right>_w$, and $u_{e} = \frac{\epsilon \widetilde{j} - \sum_{i}^{N_{i}}q_{i}n_{i} u_{i}}{q_{e}n_{e}}$ is the parameterized electron drift velocity in terms of $\left\{ n_{i,}n_{e},u_{i},\widetilde{j}\right\} $. The slow variables are $x=\left\{ {\cal F}_{\alpha},n_{\alpha},u_{i},T_{\alpha}\right\}$, while the fast variables are $y=\left\{ \widetilde{j},E\right\} $. It is seen from the Amp{\`e}re equation that when $\epsilon \rightarrow 0$, $\widetilde{j}=0$ or equivalently $j=0$, which is precisely the ambipolarity condition. Similarly, the Ohm's law for the electric field can be recovered from the current equation by additionally invoking the ion-to-electron mass ratio, $\frac{m_{i}}{m_{e}}\gg1$, as, 
\begin{equation}
\label{eq: explicit eq E epsilon 0}
    E=\frac{1}{\omega_{p}^{2}}\frac{\partial}{\partial x}\left[\sum_{i}^{N_{i}}q_{i}n_{i}u^2_{i} 
    + 
    q_{e}n_{e}u^2_{e}
    +
    \sum_{\alpha}^{N_{s}}\frac{q_{\alpha}}{m_{\alpha}}P_{\alpha}\right]
    \approx
    \frac{\partial_x P_{e}}{q_{e}n_{e}}.
\end{equation}
%
%
%
%
%
%
\section{A compatible and asymptotic preserving discretization
\label{sec:comparible_discretization}}
A compatible finite-difference discretization scheme is proposed for the conditional fast-slow formulation in Eqs. \eqref{eq:fs_cdf_equation} to \eqref{eq:fs_ampere_equation} that is, simultaneously conserving mass, momentum, and energy, while preserving solution monotonicity, quasi-neutrality asymptotic, and the Gauss law. We discretize in $x\in \Omega_x \in\mathbb{T}$ with $\Omega_x = \left[0,L_{x}\right)$ using a staggered dual mesh representation where scalar quantities $\left\{ {\cal F}_{\alpha},n_{\alpha},T_{\alpha}\right\} $ are defined on cell centers while vector quantities $\left\{ u_{i},\widetilde{j},E\right\} $ are defined on cell faces; refer to Figure \ref{fig:staggered_grid}.
\begin{figure}[th]
    \begin{centering}
    \includegraphics[width=0.6\textwidth]{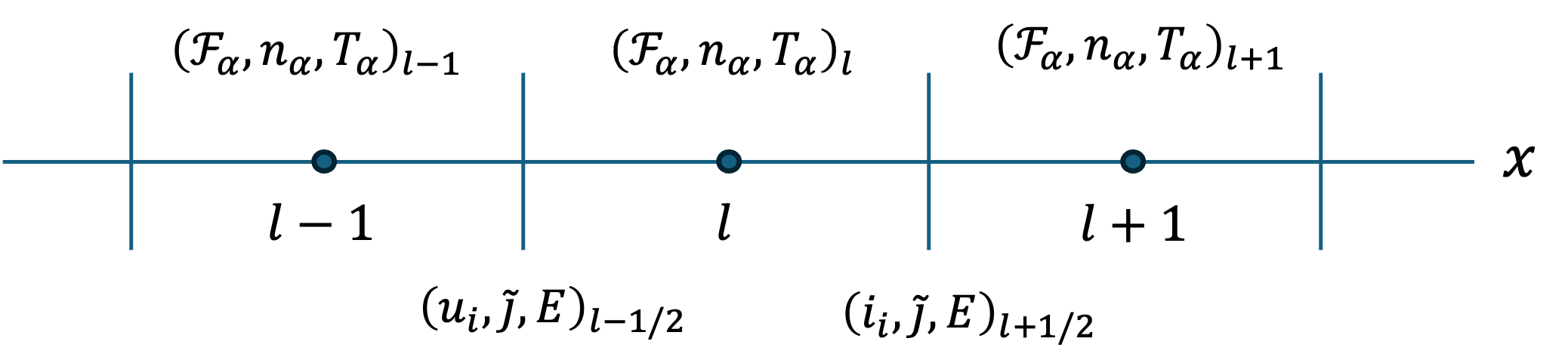}
    \par\end{centering}
    \caption{Staggered grid representation of solution variables. $\left\{ {\cal F}_{\alpha},n_{\alpha},T_{\alpha}\right\} $ are defined on cell centers, $l$, while $\left\{ u_{i},\widetilde{j},E\right\}$ are defined on cell faces, $l+\frac{1}{2}$, defined on half-integer points.
    \label{fig:staggered_grid}}
\end{figure}
In the velocity space, $w \in \Omega_w \subset \mathbb{R}$ where $\Omega_w = \left[-w_{max},w_{max}\right]$, and ${\cal F}_{\alpha}$ is discretized using a conservative finite difference with a zero flux boundary condition, $\left. \dot{w}{\cal F} \right|_{\partial \Omega_w} = 0$. The cell-centered $x$ grid is $\mathfrak{x}_{c} = \left\{ \left(l-\frac{1}{2}\right)\Delta\mathfrak{x}\right\} _{l=1}^{N_{x}} \in \mathbb{R}^{N_x}$ while the cell-face grid is $\mathfrak{x}_{f}=\left\{ \left(l-\frac{1}{2}\right)\Delta\mathfrak{x}\right\} _{l=1/2}^{N_{x}+1/2} \in \mathbb{R}^{N_x+1}$ and the $w$ grid is $\mathfrak{w}=\left\{ -w_{max}+\left(p-\frac{1}{2}\right)\Delta \mathfrak{w}\right\} _{p=1}^{N_{w}} \in \mathbb{R}^{N_w}$, where $N_{x}$ ($N_{w}$) is the number of cells in $x$ ($w$), $\Delta\mathfrak{x}=\frac{L_{x}}{N_{x}}$ ($\Delta\mathfrak{w}=\frac{2w_{max}}{N_{w}}$) is the cell intervals in $x$ ($w$), $L_{x}$ is the domain size in $x$, and $w_{max}$ is half the domain size in $w$. To compute the moments, we employ a midpoint quadrature and denote by $\left\langle A_{l}\left(w\right),B_{l}\left(w\right)\right\rangle _{\delta w} = \sum_{p=1}^{N_{w}} \Delta\mathfrak{w}A_{l} \left( \mathfrak{w}_{p} \right) B_{l} \left( \mathfrak{w}_{p} \right)$, where the subscript $l$ ($p$) denotes the real (velocity) space grid location. All vector quantities are collocated on real space cell faces to ensure that the discrete Jacobian of the fast system in the limit of $\epsilon\rightarrow0$ is invertible (to be discussed shortly), while the staggering in $n$ and $E$ is necessitated by the discrete Gauss law involution, similarly to the classical Yee scheme. The staggered formulation motivates the use of the internal energy formulation over the total energy to overcome the numerical challenges associated with the degeneracy of the drift velocities in staggered schemes for conservative representation of a fluid system. We propose an IMEX time integrator for the coupled system of equations where the conditional Vlasov equation,  \eqref{eq:fs_cdf_equation}, is evolved using a two-stage explicit Runge-Kutta (RK2) time integration with moment parameterization of $\left\{ \dot{x},\dot{w},\lambda\right\} $ evaluated using time centering, while the moment field subsystem is evolved using an implicit midpoint rule to satisfy discrete conservation symmetries (to be discussed shortly). For convenience, we define the solution vector for the coupled system of equations at time step $k$ as, 
$\vec{z}^{k} = \left\{\vec{\cal F},\vec{M} \right\}^{k}\in\mathbb{R}^{\left(N_s N_x N_w \right)+\left(3N_sN_x\right)+N_x}$, 
$\vec{\cal F} = \left\{ \vec{\cal F}_1 , \cdots, \vec{\cal F}_{N_s} \right\} \in \mathbb{R}^{N_s N_x N_w}$ 
is the conditional distribution function solution vector containing all species, $\vec{\cal F}_{\alpha} \in \mathbb{R}^{N_x N_w}$ is the conditional distribution function solution vector for species $\alpha$, $\vec{M} = \left\{\vec{\cal M} , \vec{E} \right\} \in \mathbb{R}^{3N_s N_x + N_x}$ is the solution vector containing the moment-field subsystem, $\vec{\cal M} = \left\{\vec{n}_1 , \cdots, \vec{n}_{N_s} , \vec{u}_1 , \cdots, \vec{u}_{N_i} , \cdots, \vec{T}_1 , \cdots \vec{T}_{N_s} , \vec{\widetilde{j}} \right\} \in \mathbb{R}^{3N_s N_x}$ is the moment solution vector for all species including the non-Amp{\`e}rean current, and $\vec{E} \in \mathbb{R}^{N_x}$ is the electric field solution vector. 

Consider the fully discretized conditional Vlasov equation, where the explicit species index, $\alpha$, is dropped for brevity but shall not be confused with the vector of conditional distribution function solution containing all species:
\begin{equation}
    \label{eq:discretized_cdf_vlasov}
    {\cal F}_{l,p}^{*}={\cal F}_{l,p}^{k}+\Delta tG\left(\vec{{\cal F}}^{k},\vec{{\cal M}}^{k+1},\vec{{\cal M}}^{k}\right)_{l,p},
\end{equation}
\begin{equation}
    \label{eq:discretized_cdf_vlasov_2}
    {\cal F}_{l,p}^{k+1}={\cal F}_{l,p}^{k}+\frac{\Delta t}{2}\left[G\left(\vec{{\cal F}}^{*},\vec{{\cal M}}^{k+1},\vec{{\cal M}}^{k}\right)_{l,p}+G\left(\vec{{\cal F}}^{k},\vec{{\cal M}}^{k+1},\vec{{\cal M}}^{k}\right)_{l,p}\right],
\end{equation}
where for an unspecified RK stage for ${\cal F}$, the discrete forcing function is given as
\begin{equation}
    G\left(\vec{{\cal F}},\vec{{\cal M}}^{k+1},\vec{{\cal M}}^{k}\right)_{l,p}=
    -\frac{\widehat{\dot{x}{\cal F}}_{l+\frac{1}{2},p}-\widehat{\dot{x}{\cal F}}_{l-\frac{1}{2},p}}{\Delta\mathfrak{x}}
    -\frac{\widehat{\dot{w}{\cal F}}_{l,p+\frac{1}{2}}-\widehat{\dot{w}{\cal F}}_{l,p-\frac{1}{2}}}{\Delta\mathfrak{w}}
    +\left( \lambda_{l,p} \right)^{k+\frac{1}{2}}{\cal F}_{l,p}.
\end{equation}
Here, $\widehat{\dot{x}{\cal F}}_{l+\frac{1}{2},p} = \left(\bar{\dot{x}}_{l+\frac{1}{2},p}\right)^{k+\frac{1}{2}} SMART \left( \left( \bar{\dot{x}}_{l+\frac{1}{2},p} \right)^{k+\frac{1}{2}},\vec{{\cal F}}_{p}\right)$ is the numerical flux for the real space direction at cell faces with ${\cal F}$ reconstructed using the monotone-preserving SMART discretization \cite{gaskell_ijnmf_1988_smart} and shown in \ref{app:smart_discretization} for completeness, $\left(\phi\right)^{k+\frac{1}{2}}=\frac{\phi^{k+1}+\phi^{k}}{2}$, is the temporal averaging procedure, $\widehat{\dot{w}{\cal F}}_{l,p+\frac{1}{2}} = \left( \bar{\dot{w}}_{l,p+\frac{1}{2}}\right)^{k+\frac{1}{2}} SMART\left(\left( \bar{ \dot{w}}_{l,p+\frac{1}{2}} \right)^{k+\frac{1}{2}},\vec{{\cal F}}_{l}\right)$, $\bar{\dot{x}}_{l+\frac{1}{2},p}=\frac{\dot{x}_{l+1,p} + \dot{x}_{l,p}}{2}$, $\dot{x}_{l,p}=v_{th,l}\mathfrak{w}_{p}+u_{l},$ $u_{l}=\frac{u_{l+\frac{1}{2}}+u_{l-\frac{1}{2}}}{2}$, $\dot{w}_{l,p+\frac{1}{2}}=\frac{\dot{w}_{l,p}+\dot{w}_{l,p+1}}{2}$,
\begin{equation}
    \dot{w}_{l,p}=\frac{P_{l+1,p}-P_{l-1,p}}{2n_{l}mv_{th,l}\Delta\mathfrak{x}}+\mathfrak{w}_{p}\frac{Q_{l+1}-Q_{l-1}}{2P_{l}}-\mathfrak{w}_{p}^{2}\frac{v_{th,l+1}-v_{th,l-1}}{2\Delta\mathfrak{x}},
\end{equation}
$\vec{{\cal F}_{p}}=\left\{ {\cal F}_{1,p},{\cal F}_{2,p},\cdots,{\cal F}_{N_{x},p}\right\} \in\mathbb{R}^{N_{x}}$
is a vector of the distribution function at a given velocity grid, $p$, $\vec{{\cal F}_{l}}=\left\{ {\cal F}_{l,1},{\cal F}_{l,2},\cdots,{\cal F}_{l,N_{w}}\right\} \in\mathbb{R}^{N_{w}}$ is a vector of the distribution function at a given real space grid,
$l$, and
\begin{equation}
    \lambda_{l,p}=v_{th,l}\mathfrak{w}_{p}\frac{\ln n_{l+1}-\ln n_{l-1}}{2\Delta\mathfrak{x}}-\frac{u_{l+1}-u_{l-1}}{2\Delta\mathfrak{x}}.
\end{equation}

On the other hand, for the moment-field subsystems in Eqs. \eqref{eq:fs_continuity_equation}-\eqref{eq:fs_ampere_equation}, we employ a consistent and conservative shock-capturing formulation in space while employing an implicit midpoint integration in time that satisfies the discrete energy conservation theorem (to be shown shortly). The fully discretized moment-field equations are given as:
\begin{equation}
    \label{eq:1d_discrete_continuity}
    R_{n_{\alpha},l+\frac{1}{2}} \left(\vec{M}^{k+1}; \vec{z}^{k} \right) \equiv
    \frac{n_{\alpha,l}^{k+1}-n_{\alpha,l}^{k}}{\Delta t}+\left(\frac{\widehat{F}_{n_{\alpha},l+\frac{1}{2}}-\widehat{F}_{n_{\alpha},l-\frac{1}{2}}}{\Delta\mathfrak{x}}\right)^{k+\frac{1}{2}}=
    0,
\end{equation}
\begin{eqnarray}
    \label{eq:1d_discrete_ion_momentum}
    R_{u_{i},l+\frac{1}{2}} \left(\vec{M}^{k+1} ; \vec{z}^{k} \right) \equiv
    \frac{\bar{n}_{i,l+\frac{1}{2}}^{k+1}u_{i,l+\frac{1}{2}}^{k+1}-\bar{n}_{i,l+\frac{1}{2}}^{k}u_{i,l+\frac{1}{2}}^{k}}{\Delta t}+\left(\frac{\widehat{F}_{u_{i},l+1}-\widehat{F}_{u_{i},l}}{\Delta\mathfrak{x}}\right)^{k+\frac{1}{2}}+\nonumber \\
    \left(\frac{P_{i,l+1}-P_{i,l}}{m_{i}\Delta\mathfrak{x}}\right)^{k+\frac{1}{2}} -\frac{q_{i}}{m_{i}}\left(\bar{n}_{i,l+\frac{1}{2}}E_{l+\frac{1}{2}}\right)^{k+\frac{1}{2}}=
    0,
\end{eqnarray}
\begin{eqnarray}
    \label{eq:1d_discrete_internal_energy}
    R_{T_{\alpha},l} \left(\vec{M}^{k+1}, \vec{Q}\left[\vec{\cal F}^{k+1} \right] ; \vec{z}^{k} \right) \equiv
    \frac{1}{2}\frac{n_{\alpha,l}^{k+1}T_{\alpha,l}^{k+1}-n_{\alpha,l}^{k}T_{\alpha,l}^{k}}{\Delta t}+\left(\frac{\widehat{F}_{T_{\alpha},l+\frac{1}{2}}-\widehat{F}_{T_{\alpha},l-\frac{1}{2}}}{\Delta\mathfrak{x}}\right)^{k+\frac{1}{2}}+\nonumber \\
    \left(\frac{\widehat{Q}_{\alpha,i+\frac{1}{2}}-\widehat{Q}_{\alpha,i-\frac{1}{2}}}{\Delta\mathfrak{x}}\right)^{k+\frac{1}{2}}+\left(P_{\alpha,l}\frac{u_{\alpha,l+\frac{1}{2}}-u_{\alpha,l-\frac{1}{2}}}{\Delta\mathfrak{x}}\right)^{k+\frac{1}{2}}+S_{\alpha,l}=
    0,
\end{eqnarray}
\begin{eqnarray}
    \label{eq:1d_discrete_current}
    R_{\widetilde{j},l+\frac{1}{2}}\left(\vec{M}^{k+1};\vec{z}^{k} \right)\equiv
    \epsilon\frac{\widetilde{j}_{l+\frac{1}{2}}^{k+1}-\widetilde{j}_{l+\frac{1}{2}}^{k}}{\Delta t}+
    \left(\frac{\widehat{F}_{\widetilde{j},l+1}-\widehat{F}_{\widetilde{j},l}}{\Delta\mathfrak{x}}\right)^{k+\frac{1}{2}}+\sum_{\alpha}^{N_{s}}\frac{q_{\alpha}}{m_{\alpha}}\left(\frac{P_{\alpha,l+1}-P_{\alpha,l}}{\Delta\mathfrak{x}}\right)^{k+\frac{1}{2}}-\nonumber \\
    \sum_{\alpha}^{N_{s}}\frac{q_{\alpha}^{2}}{m_{\alpha}}\left(\bar{n}_{\alpha,l+\frac{1}{2}}E_{l+\frac{1}{2}}\right)^{k+\frac{1}{2}}=
    0,
\end{eqnarray}
\begin{equation}
    \label{eq:1d_discrete_ampere}
    R_{E,l+\frac{1}{2}} \left( \vec{M}^{k+1} ; \vec{z}^{k} \right)\equiv
    \epsilon\frac{E_{l+\frac{1}{2}}^{k+1}-E_{l+\frac{1}{2}}^{k}}{\Delta t}+
    \left(\widehat{\widetilde{j}}_{l+\frac{1}{2}}\right)^{k+\frac{1}{2}}=
    0.
\end{equation}
Here, $\bar{n}_{l+\frac{1}{2}}=\frac{n_{l+1}+n_{l}}{2}$,
\begin{equation}
    S_{\alpha,l}=\frac{{\cal R}_{\alpha,l+\frac{1}{2}}+{\cal R}_{\alpha,l-\frac{1}{2}}}{2},
\end{equation}
is a numerical source term for the internal energy equation that is inspired by similar considerations as in Refs. \cite{goudon_fvca8_2017_staggered_euler_scheme_staggered, berthelin_2015_hal_barotropic_euler_model_staggered} that ensures total energy conservation theorem in a weak sense for the staggered formulation with ${\cal R}_{\alpha,l+\frac{1}{2}}$ defined as,
\begin{eqnarray}
    {\cal R}_{\alpha,l+\frac{1}{2}}=m_{\alpha}\left\{ \frac{\bar{n}_{\alpha,l+\frac{1}{2}}^{k+1}\left(u_{\alpha,l+\frac{1}{2}}^{k+1}\right)^{2}-\bar{n}_{\alpha,l+\frac{1}{2}}^{k}\left(u_{\alpha,l+\frac{1}{2}}^{k}\right)^{2}}{2\Delta t}+\left(\frac{\widehat{F}_{K_{\alpha},l+1}-\widehat{F}_{K_{\alpha},l}}{\Delta\mathfrak{x}}\right)^{k+\frac{1}{2}}+\right.\nonumber \\
    \left.\left(u_{\alpha,l+\frac{1}{2}}\frac{P_{\alpha,l+1}-P_{\alpha,l}}{m_{\alpha}\Delta\mathfrak{x}}\right)^{k+\frac{1}{2}}-\frac{q_{\alpha}}{m_{\alpha}}\left(\widehat{F}_{n_{\alpha},l+\frac{1}{2}}\right)^{k+\frac{1}{2}}\left(E_{l+\frac{1}{2}}\right)^{k+\frac{1}{2}}\right\} 
    \label{eq:1d_discrete_kinetic_energy_residual}
\end{eqnarray}
This term can be interpreted as the residual in the auxiliary evolution equation for the kinetic energy density, $\frac{nu^{2}}{2}$, that cleans the discrete total energy --the sum of kinetic and internal energy-- violation encountered in staggered schemes. We note that this term is strictly zero in the continuum and also acts as a kind of a Lagrange multiplier for the system that ensures discrete total energy conservation for the fluid in the absence of electric fields. 

We employ a monotone-preserving numerical approximation to the Euler fluxes for the mass, $\widehat{F}_{n_{\alpha},l+\frac{1}{2}}$, momentum, $\widehat{F}_{u_{i},l}$, internal energy, $\widehat{F}_{T_{\alpha},l+\frac{1}{2}}$, current, $\widehat{F}_{\widetilde{j},l}$, and auxiliary kinetic energy equations, $\widehat{F}_{K_{\alpha},l}$, from Ref. \cite{berthelin_2015_hal_barotropic_euler_model_staggered, goudon_fvca8_2017_staggered_euler_scheme_staggered, brunel_jsc_2024_3d_staggered_general_mesh_staggered} based on the kinetic-flux scheme and defined as:
\begin{equation}
    \label{eq:continuity_flux_avg}
    \widehat{F}_{n_{\alpha},l+\frac{1}{2}}=\widehat{F}_{n_{\alpha},l+\frac{1}{2}}^{-}+\widehat{F}_{n_{\alpha},l+\frac{1}{2}}^{+},
\end{equation}
\begin{equation}
    \label{eq:continuity_flux_pos_component}
    \widehat{F}_{n_{\alpha},l+\frac{1}{2}}^{-}\left(\bar{c}_{l+\frac{1}{2}},n_{\alpha,l},u_{\alpha,l+\frac{1}{2}}\right)
    =
    \left\{ \begin{array}{cc}
    0 & \text{if}\;u_{\alpha,l+\frac{1}{2}}\le-\bar{c}_{l+\frac{1}{2}}\\
    \frac{n_{\alpha,l}\left(u_{\alpha,l+\frac{1}{2}}+\bar{c}_{l+\frac{1}{2}}\right)^{2}}{4\bar{c}_{l+\frac{1}{2}}} & \text{if}\; -\bar{c}_{l+\frac{1}{2}}<u_{\alpha,l+\frac{1}{2}}+\bar{c}_{l+\frac{1}{2}}\\
    n_{\alpha,l}u_{\alpha,l+\frac{1}{2}} & \text{otherwise}
    \end{array}\right.,
\end{equation}
\begin{equation}
    \label{eq:continuity_flux_neg_component}
    \widehat{F}_{n_{\alpha},l+\frac{1}{2}}^{+}\left(\bar{c}_{l+\frac{1}{2}},n_{\alpha,l+1},u_{\alpha,l+\frac{1}{2}}\right)
    =
    \left\{ \begin{array}{cc}
    n_{\alpha,l+1}u_{\alpha,l+\frac{1}{2}} & \text{if}\;u_{\alpha,l+\frac{1}{2}}\le-\bar{c}_{l+\frac{1}{2}}\\
    -\frac{n_{\alpha,l+1}\left(u_{\alpha,l+\frac{1}{2}}-\bar{c}_{l+\frac{1}{2}}\right)^{2}}{4\bar{c}_{l+\frac{1}{2}}} & \text{if}\; - \bar{c}_{l+\frac{1}{2}}<u_{\alpha,l+\frac{1}{2}}+\bar{c}_{l+\frac{1}{2}}\\
    0 & \text{otherwise}
    \end{array}\right.,
\end{equation}
\begin{equation}
    \label{eq:speed_of_sound_averaging}
    \bar{c}_{l+\frac{1}{2}}=\frac{c_{l}+c_{l+1}}{2},
\end{equation}
\begin{equation}
    \label{eq:speed_of_sound}
    c_{l}=\sqrt{\frac{\sum_{\alpha}^{N_{s}}n_{\alpha,l}T_{\alpha,l}}{\sum_{\alpha}^{N_{\alpha}}m_{\alpha}n_{\alpha,l}},}
\end{equation}
\begin{equation}
    \label{eq:momentum_flux}
    \widehat{F}_{u_{i},l}=\frac{u_{i,l+\frac{1}{2}}\left(\widehat{F}_{n_{i},l-\frac{1}{2}}^{+}+\widehat{F}_{n_{i},l+\frac{1}{2}}^{+}\right)+u_{i,l-\frac{1}{2}}\left(\widehat{F}_{n_{i},l-\frac{1}{2}}^{-}+\widehat{F}_{n_{i},l+\frac{1}{2}}^{-}\right)}{2},
\end{equation}
\begin{equation}
    \label{eq:energy_flux}
    \widehat{F}_{T_{\alpha},l+\frac{1}{2}} =
    \frac{T_{\alpha,l}}{2} \widehat{F}_{n_{\alpha},l+\frac{1}{2}}^{-} +
    \frac{T_{\alpha,l+1}}{2}\widehat{F}_{n_{\alpha},l+\frac{1}{2}}^{+},
\end{equation}
\begin{equation}
    \label{eq:current_flux}
\widehat{F}_{\widetilde{j},l}=\sum_{i}^{N_{i}}q_{i}\widehat{F}_{u_{i},l}+q_{e}\widehat{F}_{u_{e},l},
\end{equation}
\begin{equation}
    \label{eq:electron_momentum_flux}
    \widehat{F}_{u_{e},l}=\frac{u_{e,l+\frac{1}{2}}\widehat{F}_{n_{e},l+\frac{1}{2}}+u_{e,l-\frac{1}{2}}\widehat{F}_{n_{e},l-\frac{1}{2}}}{2},
\end{equation}
\begin{equation}
    \label{eq:current_continuity_flux}
    \widehat{\widetilde{j}}_{l+\frac{1}{2}}=\frac{\sum_{i}^{N_{i}}q_{i}\widehat{F}_{n_{i},l+\frac{1}{2}} + q_{e}\widehat{F}_{n_{e},l+\frac{1}{2}}}{\epsilon},
\end{equation}
\begin{equation}
    \widehat{F}_{K_{\alpha},l}
    =
    \frac{\frac{u_{\alpha,l+\frac{1}{2}}^{2}}{2}\left(\widehat{F}_{n_{\alpha},l-\frac{1}{2}}^{+}+\widehat{F}_{n_{\alpha},l+\frac{1}{2}}^{+}\right)
    +
    \frac{u_{\alpha,l-\frac{1}{2}}^{2}}{2}\left(\widehat{F}_{n_{\alpha},l-\frac{1}{2}}^{-}+\widehat{F}_{n_{\alpha},l+\frac{1}{2}}^{-}\right)}{2},
\end{equation}
\begin{equation}
    \label{eq:discrete_electron_drift_velocity}
    u_{e,l+\frac{1}{2}}=\frac{\epsilon\widetilde{j}_{l+\frac{1}{2}}-\sum_{i}^{N_{i}}q_i \bar{n}_{i,l+\frac{1}{2}}u_{i,l+\frac{1}{2}}}{q_e \bar{n}_{e,l+\frac{1}{2}}},
\end{equation}
and we use a linear reconstruction for the heat flux, $\widehat{Q}_{\alpha,l+\frac{1}{2}}=\frac{Q_{\alpha,l+1}+Q_{\alpha,l}}{2}$ with $Q_{\alpha,l} = \frac{m_{\alpha}n_{\alpha,l} v_{th,\alpha,l}^{3}}{2}\left\langle w^{3}, \vec{\cal F}_{\alpha,l}\right\rangle _{\delta w}$. We note that an IMEX integrator for ${\cal F}$ was chosen to avoid the need to invert a large system of equations and a two-stage RK integrator was chosen over a single-stage forward Euler scheme to address the well-known numerical stability issues associated with solving hyperbolic PDEs. Although individual time integrators (RK2 for the kinetic system and implicit midpoint for the moment-field system) are second order accurate, the manner in which we combine them is not guaranteed to be so and we show in our numerical experiments that we observe first order accuracy in time.

In the following subsections, we prove that the proposed discretization satisfies the discrete conservation theorems and discretely preserves the Gauss law and the quasi-neutral limit as $\epsilon \rightarrow 0$. We stress that \emph{all of these properties are segregated away from the kinetic system and isolated in the moment-field subsystem}. 

%
%
\subsection{Gauss Law Preservation\label{subsubsec:discrete_gauss_law}}

To show that the Gauss law is satisfied in the discrete case, we start by exposing the continuum symmetries that give rise to the equivalence between the Amp{\`e}re equation and Gauss law. We begin by applying the divergence operator on Eq. \eqref{eq:fs_ampere_equation} to obtain,
\begin{equation}
    \epsilon\frac{\partial}{\partial t}\frac{\partial E}{\partial x}-\frac{\partial\widetilde{j}}{\partial x}=0.
\end{equation}
By multiplying the individual species continuity equation, Eq. \eqref{eq:fs_continuity_equation}, with their charge, and summing over all species, we obtain the charge continuity equation,
\begin{equation}
    \frac{\partial\rho}{\partial t}+\epsilon\frac{\partial\widetilde{j}}{\partial x}=0,
\end{equation}
to obtain
\begin{equation}
    \frac{\partial}{\partial t}\left(\epsilon^{2}\frac{\partial E}{\partial x}-\rho\right)=0.
\end{equation}

\begin{proposition}
    The proposed discretization satisfies the Gauss law.
\end{proposition}

\begin{proof}
In the discrete, the total charge continuity equation is derived from the Eq. \eqref{eq:1d_discrete_continuity} as,
\begin{equation}
    \sum_{\alpha}^{N_{\alpha}}\left\{ \frac{q_{\alpha}n_{\alpha,l}^{k+1}-q_{\alpha}n_{\alpha,l}^{k}}{\Delta t}+\left(\frac{q_{\alpha}\widehat{F}_{n_{\alpha},l+\frac{1}{2}}-q_{\alpha}\widehat{F}_{n_{\alpha},l-\frac{1}{2}}}{\Delta\mathfrak{x}}\right)^{k+\frac{1}{2}}\right\} \equiv\frac{\rho_{l}^{k+1}-\rho_{l}^{k}}{\Delta t}+\epsilon\left(\frac{\widehat{\widetilde{j}}_{l+\frac{1}{2}}-\widehat{\widetilde{j}}_{l-\frac{1}{2}}}{\Delta\mathfrak{x}}\right)^{k+\frac{1}{2}}=0.\label{eq:1d_discrete_charge_continuity}
\end{equation}
Taking the discrete divergence of Eq. \eqref{eq:1d_discrete_ampere},
\begin{equation}
    \epsilon\left\{ \left(\frac{E_{l+\frac{1}{2}}^{k+1}-E_{l-\frac{1}{2}}^{k+1}}{\Delta\mathfrak{x}\Delta t}\right)-\left(\frac{E_{l+\frac{1}{2}}^{k}-E_{l-\frac{1}{2}}^{k}}{\Delta\mathfrak{x}\Delta t}\right)\right\} +\left(\frac{\widehat{\widetilde{j}}_{l+\frac{1}{2}}-\widehat{\widetilde{j}}_{l-\frac{1}{2}}}{\Delta\mathfrak{x}}\right)^{k+\frac{1}{2}}=0,\label{eq:1d_discrete_ampere_divergence}
\end{equation}
and using Eq. \eqref{eq:1d_discrete_charge_continuity}, we obtain the discrete Gauss law:
\begin{equation}
    \label{eq:1d_discrete_gauss_from_ampere}
    \epsilon^{2}\frac{E_{l+\frac{1}{2}}^{k+1}-E_{l-\frac{1}{2}}^{k+1}}{\Delta\mathfrak{x}}-\rho_{l}^{k+1}=\epsilon^{2}\frac{E_{l+\frac{1}{2}}^{k}-E_{l-\frac{1}{2}}^{k}}{\Delta\mathfrak{x}}-\rho_{l}^{k} = 0.
\end{equation}
\end{proof}
This relationship states that if the initial field-density relationship satisfies the Gauss law, it remains satisfied at all times. Furthermore, from $\epsilon\rightarrow 0$, the discrete quasi-neutrality condition can be recovered:
\begin{equation}
    \label{eq:discrete_quasi_neutrality}
    n^{\epsilon \rightarrow 0}_{e,l} = -q^{-1}_e \sum^{N_i}_{i} q_i n_{i,l}.
\end{equation}

%
%
\subsection{Mass Conservation\label{subsubsec:discrete_mass_conservation}}

Mass conservation is trivially enforced because of the conservative form of the continuity equation. Due to the conservative flux discretization employed in this study, all terms involving divergence operators vanish with appropriate boundary conditions when integrated over space and thus will be dropped for brevity. Multiplying the per species continuity equation with $m_{\alpha}$ and discretely integrating in $x$, we obtain:
\[
\frac{M^{k+1}-M^{k}}{\Delta t}=\sum_{\alpha}^{N_{s}}m_{\alpha}\sum_{l}^{N_{x}}\Delta\mathfrak{x}\frac{n_{\alpha,l}^{k+1}-n_{\alpha,l}^{k}}{\Delta t}=0.
\]
%
%
%
%
%
%
\subsection{Momentum Conservation\label{subsubsec:momentum_conservation}}

To demonstrate a discrete total momentum conservation theorem, we begin the proof in the continuum to illustrate the key symmetries that must be mimicked in the discrete. The total momentum conservation theorem is defined as follows:
\begin{equation}
\frac{\partial}{\partial t}\int_{\Omega_{x}}dx\left[\sum_{i}^{N_{i}}m_{i}n_{i}u_{i}+m_{e}n_{e}u_{e}\right]=0.\label{eq:1d_total_mom_cons_theorem}
\end{equation}
Because in our modified moment-field subsystem we do not explicitly evolve the electron momentum, $m_{e}n_{e}u_{e}$, we must satisfy the involution constraint that is,
\begin{equation}
    \label{eq:electron_momentum_involution}
    \frac{\partial}{\partial t}\left(m_{e}n_{e}u_{e}\right)=\frac{m_{e}}{q_{e}}\left(\epsilon\frac{\partial\widetilde{j}}{\partial t}-\sum_{i}^{N_{i}}q_{i}\frac{\partial}{\partial t}\left(n_{i}u_{i}\right)\right)    
\end{equation} 
which yields for the total momentum conservation theorem the following relationship:
\begin{equation}
    \frac{\partial}{\partial t}\int_{\Omega_{x}}dx\left[\sum_{i}^{N_{i}}m_{i}n_{i}u_{i}+\frac{m_{e}}{q_{e}}\left(\epsilon\widetilde{j}-\sum_{i}^{N_{i}}q_{i}n_{i}u_{i}\right)\right]=0.
\end{equation}
We sum Eqs. \eqref{eq:fs_momentum_equation} over all ion species, the $m_{e}/q_{e}$ multiplied Eq. \eqref{eq:fs_current_equation} and the negative of $m_{e}q_{i}/q_{e}m_{i}$ multiplied sum of Eq. \eqref{eq:fs_momentum_equation} over all ion species to obtain:
\begin{eqnarray}
    \sum_{i}^{N_{i}}\left[\frac{\partial}{\partial t}\left(m_{i}n_{i}u_{i}\right)-q_{i}n_{i}E\right]+\frac{m_{e}}{q_{e}}\left(\epsilon\frac{\partial\widetilde{j}}{\partial t}-\sum_{\alpha}^{N_{s}}\frac{q_{\alpha}^{2}n_{\alpha}}{m_{\alpha}}E-\sum_{i}^{N_{i}}q_{i}\left[\frac{\partial}{\partial t}\left(n_{i}u_{i}\right)-\frac{q_{i}}{m_{i}}n_{i}E\right]\right) & = & 0,
\end{eqnarray}
which could be simplified as
\begin{eqnarray}
    \sum_{i}^{N_{i}}\left[\frac{\partial}{\partial t}\left(m_{i}n_{i}u_{i}\right)\right]+\frac{m_{e}}{q_{e}}\left(\epsilon\frac{\partial\widetilde{j}}{\partial t}-\sum_{i}^{N_{i}}q_{i}\frac{\partial}{\partial t}\left(n_{i}u_{i}\right)\right)-E\sum_{\alpha}^{N_{s}}q_{\alpha}n_{\alpha} & = & 0.
\end{eqnarray}
Since $\rho=\sum_{\alpha}^{N_{s}}q_{\alpha}n_{\alpha}$ and $\epsilon^{2}\frac{\partial E}{\partial x}=\rho$,
we obtain
\begin{eqnarray}
    \sum_{i}^{N_{i}}\left[\frac{\partial}{\partial t}\left(m_{i}n_{i}u_{i}\right)\right]+\frac{m_{e}}{q_{e}}\left(\epsilon\frac{\partial\widetilde{j}}{\partial t}-\sum_{i}^{N_{i}}q_{i}\frac{\partial}{\partial t}\left(n_{i}u_{i}\right)\right)-\frac{\epsilon^{2}}{2}\frac{\partial E^{2}}{\partial x} & = & 0.
\end{eqnarray}
Finally, by integrating over space the divergence of electrostatic stress tensor vanishes with appropriate boundary conditions to yield,
\begin{equation}
    \int_{\Omega_{x}}dx\left[\sum_{i}^{N_{i}}\frac{\partial}{\partial t}\left(m_{i}n_{i}u_{i}\right)+\frac{m_{e}}{q_{e}}\left(\epsilon\frac{\partial\widetilde{j}}{\partial t}-\sum_{i}^{N_{i}}q_{i}\frac{\partial}{\partial t}\left(n_{i}u_{i}\right)\right)\right]=0.
\end{equation}

\begin{proposition}
    The proposed discretization satisfies the discrete total-momentum conservation theorem.
\end{proposition}
\begin{proof}
We show that our discretization mimics key continuum symmetries, namely 
\begin{equation*}
    \frac{\partial}{\partial t} \left(m_{e}n_{e}u_{e} \right) = \frac{m_{e}}{q_{e}} \left( \epsilon \frac{\partial \widetilde{j}}{\partial t} - \sum_{i}^{N_{i}} q_{i} \frac{\partial}{\partial t} \left( n_{i} u_{i} \right) \right)
\end{equation*}
and $\rho E = \frac{ \epsilon^{2}}{2} \frac{\partial E^{2}}{\partial x}$ to prove the discrete total momentum conservation theorem. From Eqs.
\eqref{eq:1d_discrete_ion_momentum} and \eqref{eq:1d_discrete_current}, we obtain:
\begin{eqnarray}
    \sum_{i}^{N_{i}}\left(m_{i}\frac{\bar{n}_{i,l+\frac{1}{2}}^{k+1}u_{i,l+\frac{1}{2}}^{k+1}-\bar{n}_{i,l+\frac{1}{2}}^{k+1}u_{i,l+\frac{1}{2}}^{k+1}}{\Delta t}-q_{i}\left(\bar{n}_{i,l+\frac{1}{2}}E_{l+\frac{1}{2}}\right)^{k+\frac{1}{2}}\right)+\nonumber \\
    +\frac{m_{e}}{q_{e}}\left(\epsilon\frac{\widetilde{j}_{l+\frac{1}{2}}^{k+1}-\widetilde{j}_{l+\frac{1}{2}}^{k}}{\Delta t}-\sum_{\alpha}^{N_{s}}\frac{q_{\alpha}^{2}}{m_{\alpha}}\left(\bar{n}_{\alpha,l+\frac{1}{2}}E_{l+\frac{1}{2}}\right)^{k+\frac{1}{2}}\right)-\nonumber \\
    \frac{m_{e}}{q_{e}}\left(\sum_{i}^{N_{i}}q_{i}\frac{\bar{n}_{i,l+\frac{1}{2}}^{k+1}u_{i,l+\frac{1}{2}}^{k+1}-\bar{n}_{i,l+\frac{1}{2}}^{k}u_{i,l+\frac{1}{2}}^{k}}{\Delta t}-\frac{q_{i}^{2}}{m_{i}}\left(\bar{n}_{i,l+\frac{1}{2}}E_{l+\frac{1}{2}}\right)^{k+\frac{1}{2}}\right)=0.
\end{eqnarray}
It is trivial to show that this relationship can be simplified to,
\begin{eqnarray}
    \sum_{i}^{N_{i}}\left(m_{i}\frac{\bar{n}_{i,l+\frac{1}{2}}^{k+1}u_{i,l+\frac{1}{2}}^{k+1}-\bar{n}_{i,l+\frac{1}{2}}^{k+1}u_{i,l+\frac{1}{2}}^{k+1}}{\Delta t}\right)-\nonumber \\
    +\frac{m_{e}}{q_{e}}\left(\epsilon\frac{\widetilde{j}_{l+\frac{1}{2}}^{k+1}-\widetilde{j}_{l+\frac{1}{2}}^{k}}{\Delta t}-\sum_{i}^{N_{i}}q_{i}\frac{\bar{n}_{i,l+\frac{1}{2}}^{k+1}u_{i,l+\frac{1}{2}}^{k+1}-\bar{n}_{i,l+\frac{1}{2}}^{k}u_{i,l+\frac{1}{2}}^{k}}{\Delta t}\right)-\left(\bar{\rho}_{\alpha,l+\frac{1}{2}}E_{l+\frac{1}{2}}\right)^{k+\frac{1}{2}}=0.
\end{eqnarray}
Finally, using Eq. \eqref{eq:1d_discrete_gauss_from_ampere}, we write $\bar{\rho}_{l+\frac{1}{2}}=\epsilon^{2}\frac{\left(E_{l+\frac{3}{2}}\right)-\left(E_{l-\frac{1}{2}}\right)}{2\Delta \mathfrak{x}}$ to obtain:
\begin{eqnarray}
    \sum_{i}^{N_{i}}\left(m_{i}\frac{\bar{n}_{i,l+\frac{1}{2}}^{k+1}u_{i,l+\frac{1}{2}}^{k+1}-\bar{n}_{i,l+\frac{1}{2}}^{k+1}u_{i,l+\frac{1}{2}}^{k+1}}{\Delta t}\right)+\frac{m_{e}}{q_{e}}\left(\epsilon\frac{\widetilde{j}_{l+\frac{1}{2}}^{k+1}-\widetilde{j}_{l+\frac{1}{2}}^{k}}{\Delta t}-\sum_{i}^{N_{i}}q_{i}\frac{\bar{n}_{i,l+\frac{1}{2}}^{k+1}u_{i,l+\frac{1}{2}}^{k+1}-\bar{n}_{i,l+\frac{1}{2}}^{k}u_{i,l+\frac{1}{2}}^{k}}{\Delta t}\right)\nonumber \\
    -\epsilon^{2}\frac{\left(E_{l+\frac{3}{2}}E_{l+\frac{1}{2}}\right)^{k+\frac{1}{2}}-\left(E_{l-\frac{1}{2}}^{k+1}E_{l+\frac{1}{2}}\right)^{k+\frac{1}{2}}}{2\Delta \mathfrak{x}}=0.
\end{eqnarray}
Summing over space and telescoping the last term, we show the discrete total momentum conservation theorem:
\begin{equation}
    \sum_{l}^{N_{x}}\Delta\mathfrak{x}\left\{ \sum_{i}^{N_{i}}m_{i}\frac{\bar{n}_{i,l+\frac{1}{2}}^{k+1}u_{i,l+\frac{1}{2}}^{k+1}-\bar{n}_{i,l+\frac{1}{2}}^{k+1}u_{i,l+\frac{1}{2}}^{k+1}}{\Delta t}+\frac{m_{e}}{q_{e}}\left(\epsilon\frac{\widetilde{j}_{l+\frac{1}{2}}^{k+1}-\widetilde{j}_{l+\frac{1}{2}}^{k}}{\Delta t}-\sum_{i}^{N_{i}}q_{i}\frac{\bar{n}_{i,l+\frac{1}{2}}^{k+1}u_{i,l+\frac{1}{2}}^{k+1}-\bar{n}_{i,l+\frac{1}{2}}^{k}u_{i,l+\frac{1}{2}}^{k}}{\Delta t}\right)\right\} =0.
\end{equation}
\end{proof}
%
%
%
%
%
%
\subsection{Energy Conservation\label{subsubsec:energy_conservation}}
By summing the $u_i$ multiplied Eq. \eqref{eq:fs_momentum_equation} for all ion species, the $u_e$ multiplied Eq. \eqref{eq:electron_momentum_involution}, and Eq. \eqref{eq:fs_internal_energy_equation} for all species and using chain rules and continuity equations, we obtain:
\begin{eqnarray}
    \sum^{N_s}_{\alpha} \left\{
    \frac{\partial}{\partial t} \left( \frac{m_{\alpha} n_{\alpha} u^2_{\alpha}}{2}\right) + 
    \frac{\partial}{\partial x}\left(\frac{m_{\alpha} n_{\alpha} u^2_{\alpha}}{2} \right) +
    u_{\alpha} \frac{\partial P_{\alpha}}{\partial x} -
    q_{\alpha} n_{\alpha}u_{\alpha} E     
    \right. +
    \nonumber \\
    \left.
    \frac{1}{2} \frac{\partial}{\partial t}\left(n_{\alpha} T_{\alpha} \right) 
    +
    \frac{1}{2} \frac{\partial}{\partial x} \left(u_{\alpha} n_{\alpha} T_{\alpha} \right)
    +
    P_{\alpha} \frac{\partial u_{\alpha}}{\partial x}
    \right\} = 0.
    \label{eq:continuum_energy_conservation_proof}
\end{eqnarray}
Here, non-conservative terms corresponding to compressive heating and $PdV$ work combine to form $u_{\alpha} \partial_x P_{\alpha} + P_{\alpha} \partial_x u_{\alpha} = \partial_x \left( u_{\alpha} P_{\alpha} \right)$. Using Eq. \eqref{eq:fs_ampere_equation} and \eqref{eq:nonamperean_current} such that $\sum^{N_s}_{\alpha} q_{\alpha} n_{\alpha} u_{\alpha} E = -\frac{\epsilon^2}{2}\frac{\partial E^2}{\partial t}$, and integrating over space with appropriate boundary conditions, we obtain the total energy conservation theorem:
\begin{equation}
    \label{eq:continuum_energy_conservation_theorem}
    \frac{\partial}{\partial t} \int_{\Omega_x} dx
    \left[
    \sum^{N_s}_{\alpha} \left(\frac{m_{\alpha}n_{\alpha}u^2_{\alpha}}{2} + \frac{n_{\alpha} T_{\alpha}}{2}\right)
    +
    \frac{\epsilon^2 E^2}{2}
    \right] = 0.
\end{equation}

\begin{proposition}
    The proposed discretization scheme satisfies the discrete total energy conservation theorem.
\end{proposition}

\begin{proof}
By summing Eqs. \eqref{eq:1d_discrete_internal_energy} for all species, summing over all space and using the definition of ${\cal R}_{\alpha,l+\frac{1}{2}}$ from Eq. \eqref{eq:1d_discrete_kinetic_energy_residual} we obtain:
\begin{eqnarray}
    \sum^{N_x}_{l} \Delta\mathfrak{x}
    \sum_{\alpha}^{N_{s}}
    \left\{ 
        m_{\alpha}\frac{\bar{n}_{\alpha,l+\frac{1}{2}}^{k+1}\left(u_{\alpha,l+\frac{1}{2}}^{k+1}\right)^{2}
        -
        \bar{n}_{\alpha,l+\frac{1}{2}}^{k}\left(u_{\alpha,l+\frac{1}{2}}^{k}\right)^{2}}{2\Delta t}
        -
        q_{\alpha}\left( \widehat{F}_{n_{\alpha},l+\frac{1}{2}}\right)^{k+\frac{1}{2}}\left(E_{l+\frac{1}{2}}\right)^{k+\frac{1}{2}}
    \right.+\nonumber \\
    \left.
        \left(u_{\alpha,l+\frac{1}{2}}\frac{P_{\alpha,l+1}-P_{\alpha,l}}{\Delta\mathfrak{x}}\right)^{k+\frac{1}{2}}+\frac{n_{\alpha,l}^{k+1}T_{\alpha,l}^{k+1}-n_{\alpha,l}^{k}T_{\alpha,l}^{k}}{2\Delta t}+\left(P_{\alpha,l+1}\frac{u_{\alpha,l+\frac{1}{2}}-u_{\alpha,l-\frac{1}{2}}}{\Delta\mathfrak{x}}\right)^{k+\frac{1}{2}}
    \right\} =0.
\end{eqnarray}
It is trivially shown that, when integrated over space, the non-conservative terms cancel out to yield:
\begin{equation}
    \frac{{\cal E}_{p}^{k+1}-{\cal E}_{p}^{k}}{\Delta t}
    -
    \sum_{l}^{N_{x}}\Delta\mathfrak{x}\sum_{\alpha}^{N_{s}}\left\{ q_{\alpha}\left(\widehat{F}_{n_{\alpha},l+\frac{1}{2}}\right)^{k+\frac{1}{2}}\left(E_{l+\frac{1}{2}}\right)^{k+\frac{1}{2}}\right\} 
    = 0,
\end{equation}
where ${\cal E}_{p}=\sum_{l}^{N_{x}}\Delta\mathfrak{x}\sum_{\alpha}^{N_{s}}\left\{ m_{\alpha}\frac{\bar{n}_{\alpha,l+\frac{1}{2}}u_{\alpha,l+\frac{1}{2}}^{2}}{2}+\frac{n_{\alpha,l}T_{\alpha,l}}{2}\right\} $ is the total plasma energy in our staggered formulation. From Eq. \eqref{eq:1d_discrete_ampere},
\begin{equation}
    \sum_{\alpha}^{N_{s}}q_{\alpha}\left( \widehat{F}_{n_{\alpha},l+\frac{1}{2}} \right)^{k+\frac{1}{2}}=-\epsilon^{2}\frac{E_{l+\frac{1}{2}}^{k+1}-E_{l+\frac{1}{2}}^{k}}{\Delta t},
\end{equation}
and because
\begin{equation}
    \label{eq:implicit_midpoint_chain_rule}
    \frac{E^{k+1}_{l+\frac{1}{2}} + E^{k}_{l+\frac{1}{2}}}{2}
    \frac{E^{k+1}_{l+\frac{1}{2}} - E^{k}_{l+\frac{1}{2}}}{\Delta t}
    =
    \frac{\left( E^{k+1}_{l+\frac{1}{2}} \right)^2 - \left( E^{k}_{l+\frac{1}{2}} \right)^2}{2\Delta t},
\end{equation}
we obtain:
\begin{equation}
    \frac{{\cal E}^{k+1}-{\cal E}^{k}}{\Delta t}\equiv\frac{{\cal E}_{p}^{k+1}-{\cal E}_{p}^{k}}{\Delta t}+\frac{\epsilon^{2}}{2}\sum_{l}^{N_{x}}\Delta \mathfrak{x} \left\{ \frac{\left(E_{l+\frac{1}{2}}^{k+1}\right)^{2}-\left(E_{l+\frac{1}{2}}^{k}\right)^{2}}{\Delta t}\right\} =0,
\end{equation}    
the discrete total energy conservation theorem.
\end{proof}
%
%
%
%
%
%
\subsection{Existence of the Discrete Formal Slow Manifold and the quasi-neutral Asymptotic Preserving Property \label{subsubsec:discrete_qn_ap}}
%
%
%
%
In Section \ref{sec:fast_slow_formulation} we showed that Eqs.~(\ref{eq:fs_cdf_equation}) -- (\ref{eq:fs_ampere_equation}) 
form a fast-slow system with slow variables $\left\{ {\cal F}_{\alpha},n_{\alpha},u_{i},T_{\alpha}\right\}$ and fast variables $\left\{ \widetilde{j},E\right\}$. We also showed that in the limit $\epsilon\to 0$ the fast variables are functions of slow variables $\left\{ \widetilde{j},E\right\} = Y_{0}({\cal F}_{\alpha},n_{\alpha},u_{i},T_{\alpha})$ given explicitly by Eq.~(\ref{eq: explicit eq E epsilon 0}) and $\widetilde{j}=0$. In this section, we will show that Eqs.~(\ref{eq:discretized_cdf_vlasov}), (\ref{eq:discretized_cdf_vlasov_2}), (\ref{eq:1d_discrete_continuity}) -- (\ref{eq:1d_discrete_ampere}) represent a discrete-time fast slow system that implements a slow manifold integrator for (\ref{eq:fs_cdf_equation}) -- (\ref{eq:fs_ampere_equation}). We prove this for two scenarios: 1) a general multi-ion case \emph{without} the kinetic-flux discretization used in Eq. \eqref{eq:continuity_flux_avg} but instead using a linear reconstruction (to be discussed shortly), and 2) using the kinetic-flux reconstruction but for a single-ion limit.
We will begin by introducing the following technical result.

\begin{proposition} \label{proposition:equilvalence_of_numerical_current_and_solution_current_1}
    Let the linear reconstruction flux $\widehat{F}_{n_{\alpha},l+\frac{1}{2}} = u_{\alpha,l+\frac{1}{2}} \bar{n}_{\alpha,l+\frac{1}{2}}$ be used in the numerical scheme instead of 
    Eqs.~(\ref{eq:continuity_flux_pos_component}), (\ref{eq:continuity_flux_neg_component}). Then, Eqs.~(\ref{eq:discretized_cdf_vlasov}), (\ref{eq:discretized_cdf_vlasov_2}), (\ref{eq:1d_discrete_continuity}) -- (\ref{eq:1d_discrete_ampere}) take the form
   $(x^{k+1},y^{k+1}) = (x^{k},\psi_{x}(y^{k}))$  in the limit $\Delta t\to 0$, $\frac{\epsilon}{\Delta t}\to 0$ for a smooth function $\psi_{x}(y)$, where $x=\left\{ {\cal F}_{\alpha,l,p},{\cal F}^{\ast}_{\alpha,l,p}, n_{\alpha,l},\right.$  $\left. u_{i,l+\frac{1}{2}},T_{\alpha,l}\right\}$ and $y=\left\{ \widetilde{j}_{l+\frac{1}{2}},E_{l+\frac{1}{2}}\right\}$. 
\end{proposition}
\begin{proof}
    From Eq. \eqref{eq:current_continuity_flux} and \eqref{eq:discrete_electron_drift_velocity}, and assuming a linear reconstruction for $\widehat{F}_{n_{\alpha},l+\frac{1}{2}} = u_{\alpha,l+\frac{1}{2}} \bar{n}_{\alpha,l+\frac{1}{2}}$, we can write $\widehat{\widetilde{j}}_{l+\frac{1}{2}}$ as
    \begin{equation}
        \label{eq:multi_ion_jhattilde}
        \widehat{\widetilde{j}}_{l+\frac{1}{2}} = 
        \frac{\left(\epsilon \tilde{j}_{l+\frac{1}{2}} - \sum^{N_i}_{i} q_i \bar{n}_{i+\frac{1}{2}} u_{i,l+\frac{1}{2}} \right) + \sum^{N_i}_{i}q_i \widehat{F}_{n_i,l+\frac{1}{2}}}{\epsilon} + O(\epsilon) = \tilde{j}_{l+\frac{1}{2}}+O(\epsilon).
    \end{equation}
    Also, Eq.~\eqref{eq:1d_discrete_current} yeilds
    \begin{eqnarray}
    \label{eq: Solving for E with O(epsilon)}
    E_{l+\frac{1}{2}}^{k+1} &=&
    -\left(\sum_{\alpha}^{N_{s}} \frac{q_{\alpha}^{2}}{m_{\alpha}} \bar{n}^{k+1}_{\alpha,l+\frac{1}{2}}\right)^{-1}
    \left(\sum_{\alpha}^{N_{s}} \frac{q_{\alpha}^{2}}{m_{\alpha}} \bar{n}^{k}_{\alpha,l+\frac{1}{2}}\right) E_{l+\frac{1}{2}}^{k} \nonumber \\
    & & {} + 
    2\left(\sum_{\alpha}^{N_{s}} \frac{q_{\alpha}^{2}}{m_{\alpha}} \bar{n}^{k+1}_{\alpha,l+\frac{1}{2}}\right)^{-1}
    \left( \left(\frac{\widehat{F}_{\widetilde{j},l+1}-\widehat{F}_{\widetilde{j},l}}{\Delta\mathfrak{x}}\right)^{k+\frac{1}{2}}+\sum_{\alpha}^{N_{s}}\frac{q_{\alpha}}{m_{\alpha}}\left(\frac{P_{\alpha,l+1}-P_{\alpha,l}}{\Delta\mathfrak{x}}\right)^{k+\frac{1}{2}}
    \right) + O(\epsilon). 
\end{eqnarray}
    By multiplying Eqs.~(\ref{eq:1d_discrete_continuity}) -- (\ref{eq:1d_discrete_internal_energy}) by $\Delta t$ and taking the limit $\Delta t\to 0$, $\frac{\epsilon}{\Delta t}\to 0$ in Eqs.~(\ref{eq:discretized_cdf_vlasov}), (\ref{eq:discretized_cdf_vlasov_2}), (\ref{eq:1d_discrete_continuity}) -- (\ref{eq:1d_discrete_ampere}) we obtain the following asymptotic relations:
\begin{gather}
    \label{eq: Limiting case slow variables}
    {\cal F}_{l,p}^{k+1}={\cal F}_{l,p}^{*}={\cal F}_{l,p}^{k} \quad , \quad n^{k+1}_{\alpha, l}=n^{k}_{\alpha, l}\quad , \quad 
     u_{\alpha,l+\frac{1}{2}}^{k+1} = u_{\alpha,l+\frac{1}{2}}^{k} \quad , \quad T^{k+1}_{\alpha,l}=T^{k}_{\alpha,l} \quad , \\
    \label{eq: Limiting case fast variable j-hat-tilde} 
    \left(\widehat{\widetilde{j}}^{\epsilon\rightarrow 0}_{l+\frac{1}{2}}\right)^{k+\frac{1}{2}}=
    0.
\end{gather}
    In view of Eq. (\ref{eq:multi_ion_jhattilde}), the latter equation yields 
    \begin{equation}
    \label{eq: linear limiting case j-tilde}
    \tilde{j}^{k+1}_{l+\frac{1}{2}} = -\tilde{j}^{k}_{l+\frac{1}{2}}. 
    \end{equation}
    Also, Eqs.~\eqref{eq: Solving for E with O(epsilon)} and \eqref{eq: Limiting case slow variables} imply that 
    \begin{equation}
    \label{eq: limiting case E explicit}
    E_{l+\frac{1}{2}}^{k+1} = - E_{l+\frac{1}{2}}^{k} + 2\left( \sum_{\alpha}^{N_{s}}\frac{q_{\alpha}^{2}}{m_{\alpha}}\bar{n}^{k}_{\alpha,l+\frac{1}{2}}\right)^{-1} \left[ \left( \frac{\widehat{F}^{\epsilon\rightarrow 0,k}_{\widetilde{j},l+1}-\widehat{F}^{\epsilon\rightarrow 0,k}_{\widetilde{j},l}}{\Delta\mathfrak{x}}\right)+\sum_{\alpha}^{N_{s}}\frac{q_{\alpha}}{m_{\alpha}}\left(\frac{P^{k}_{\alpha,l+1}-P^{k}_{\alpha,l}}{\Delta\mathfrak{x}}\right)\right],
    \end{equation}
   Thus, in the limiting case  $\Delta t\to 0$, $\frac{\epsilon}{\Delta t}\to 0$, Eqs.~(\ref{eq:discretized_cdf_vlasov}), (\ref{eq:discretized_cdf_vlasov_2}), (\ref{eq:1d_discrete_continuity}) -- (\ref{eq:1d_discrete_ampere}) take the form $(x^{k+1},y^{k+1}) = \Phi_{0,0}(x^{k},y^{k})$, where $\Phi_{0,0} (x,y)=(x,\psi_{x}(y))$, and the variables are given by $x=\left\{ {\cal F}_{\alpha,l,p},{\cal F}^{\ast}_{\alpha,l,p}, n_{\alpha,l},u_{i,l+\frac{1}{2}},T_{\alpha,l}\right\}$ and $y=\left\{ \widetilde{j}_{l+\frac{1}{2}},E_{l+\frac{1}{2}}\right\}$, and the components of $\psi_{x}(y)$ are defined by Eqs.~(\ref{eq: linear limiting case j-tilde}), 
   (\ref{eq: limiting case E explicit}).
\end{proof}
\begin{proposition}\label{proposition:equilvalence_of_numerical_current_and_solution_current_2}
   Suppose that discrete continuity flux (\ref{eq:continuity_flux_pos_component}), (\ref{eq:continuity_flux_neg_component}) be used  in the numerical scheme. Then, Eqs.~(\ref{eq:discretized_cdf_vlasov}), (\ref{eq:discretized_cdf_vlasov_2}), (\ref{eq:1d_discrete_continuity}) -- (\ref{eq:1d_discrete_ampere}) take the form
   $(x^{k+1},y^{k+1}) = (x^{k},\psi_{x}(y^{k}))$  in the limit $\Delta t\to 0$, $\frac{\epsilon}{\Delta t}\to 0$ in the case of single-ion species. \\

\end{proposition}
\begin{proof}
   From Eqs.~\eqref{eq:current_continuity_flux}, \eqref{eq:continuity_flux_avg}, \eqref{eq:continuity_flux_neg_component}, \eqref{eq:continuity_flux_pos_component}, we can write $\widehat{\widetilde{j}}_{l+\frac{1}{2}}$ as,
     \begin{eqnarray}    
        \widehat{\widetilde{j}}_{l+\frac{1}{2}} &=& 
        \frac{1}{\epsilon} \left\{\sum_{i}^{N_{i}} q_{i} \left[ 
            \begin{array}{ccc}
                n_{i,l+1}u_{i,l+\frac{1}{2}} &  
                \textnormal{if} & 
                u_{i,l+\frac{1}{2}} \le -\bar{c}_{l+\frac{1}{2}} \\[2mm]
                \displaystyle
                (n_{i,l} - n_{i,l+1})\frac{u_{i,l+\frac{1}{2}}^2 + \bar{c}^2_{l+\frac{1}{2}}}{4\bar{c}_{l+\frac{1}{2}}} + \bar{n}_{i,l+\frac{1}{2}} u_{i,l+\frac{1}{2}} &
                \textnormal{if} & 
                -\bar{c}_{l+\frac{1}{2}}<u_{i,l+\frac{1}{2}} < \bar{c}_{l+\frac{1}{2}}\; \\[4mm]
                n_{i,l} u_{i,l+\frac{1}{2}} &  
                \textnormal{if} & 
                \bar{c}_{l+\frac{1}{2}} \le u_{i,l+\frac{1}{2}}  
            \end{array}
        \right] \right.  \nonumber \\[4mm]
        \label{eq: j hat tilde expressed using flux}
        & & \hspace*{8mm} \left. + q_e \left[ 
            \begin{array}{ccc}
                n_{e,l+1}u_{e,l+\frac{1}{2}} &  
                \textnormal{if} & 
                u_{e,l+\frac{1}{2}} \le -\bar{c}_{l+\frac{1}{2}} \\[2mm]
                \displaystyle
                (n_{e,l} - n_{e,l+1})\frac{u_{e,l+\frac{1}{2}}^2 + \bar{c}^2_{l+\frac{1}{2}}}{4\bar{c}_{l+\frac{1}{2}}} + \bar{n}_{e,l+\frac{1}{2}} u_{e,l+\frac{1}{2}} &
                \textnormal{if} & 
                -\bar{c}_{l+\frac{1}{2}} < u_{e,l+\frac{1}{2}} < \bar{c}_{l+\frac{1}{2}}\; \\[4mm]
                n_{e,l} u_{e,l+\frac{1}{2}} &  
                \textnormal{if} & 
                \bar{c}_{l+\frac{1}{2}} \le u_{e,l+\frac{1}{2}}  
            \end{array}
        \right] \right\}
    \end{eqnarray}
    In addition, the discrete Gauss law~\eqref{eq:1d_discrete_gauss_from_ampere} can be re-stated as 
    \begin{equation}
    \label{eq: discrete gauss law with epsilon -- one cell}
    \sum_{i}^{N_i} q_i n_{i,l} + q_e n_{e,l} = O(\epsilon^2). 
    \end{equation}
    By writing the latter equation on cells $l$ and $l+1$ and taking the average we obtain
    \begin{equation}
    \label{eq: discrete gauss law with epsilon -- cell average}
    \sum_{i}^{N_i} q_i \bar{n}_{i,l+\frac{1}{2}} + q_e \bar{n}_{e,l+\frac{1}{2}} = O(\epsilon^2). 
    \end{equation}
    Combining the last equation with Eq.~\eqref{eq:discrete_electron_drift_velocity} we have
    \begin{equation}
        \label{eq: electron velocity finite epsilon}
        u_{e,l+\frac{1}{2}} = -\left( q_e \bar{n}_{e,l+\frac{1}{2}}\right)^{-1} \sum_{i}^{N_{i}}  q_i \bar{n}_{i,l+\frac{1}{2}} u_{i,l+\frac{1}{2}}+\epsilon \left( q_e \bar{n}_{e,l+\frac{1}{2}}\right)^{-1} \tilde{j}_{l+\frac{1}{2}}\, .
    \end{equation} 
    In general, it is not trivial to discuss properties of \eqref{eq: j hat tilde expressed using flux} because we expect an incoherent behavior of $u_{\alpha,l+\frac{1}{2}}$ for ion and electron species resulting in different stencils used for different species. However, a simple treatment is possible in the case of $N_{i}=1$, that is, when single ion specie is present. In this case, Eqs.~\eqref{eq: discrete gauss law with epsilon -- one cell} and \eqref{eq: electron velocity finite epsilon} simplify to the following 
    after a simple manipulation:%
    \begin{gather}
    \label{eq: discrtete Gauss law single ion specie}
    q_i n_{i,l} + q_e n_{e,l} = O(\epsilon^2), \\
    \label{eq: electron velocity finite epsilon single ion specie}
        u_{e,l+\frac{1}{2}} = u_{i,l+\frac{1}{2}}+\epsilon \left( q_e \bar{n}_{e,l+\frac{1}{2}}\right)^{-1} \tilde{j}_{l+\frac{1}{2}}+O(\epsilon^2) .
    \end{gather}%
    The last equation states that for sufficiently small $\epsilon$, the velocities $u_{i,l+\frac{1}{2}}$ and $u_{i,l+\frac{1}{2}}$ start to behave coherently. In particular, \eqref{eq: j hat tilde expressed using flux} simplifies to 
    \begin{eqnarray}    
     \widehat{\widetilde{j}}_{l+\frac{1}{2}} &=& 
        \frac{1}{\epsilon} \left\{  
            \begin{array}{ccc}
                q_i n_{i,l+1}u_{i,l+\frac{1}{2}} + q_e n_{e,l+1}u_{e,l+\frac{1}{2}} &  
                \textnormal{if} & 
                u_{i,l+\frac{1}{2}} \le -\bar{c}_{l+\frac{1}{2}} \\[2mm]
                \displaystyle
                q_i \bar{n}_{i,l+\frac{1}{2}} u_{i,l+\frac{1}{2}} + 
                q_e \bar{n}_{e,l+\frac{1}{2}} u_{e,l+\frac{1}{2}} + 
                \epsilon \frac{n_{e,l} - n_{e,l+1}}{\bar{n}_{e,l+\frac{1}{2}}}\frac{u_{i,l+\frac{1}{2}}}{2\bar{c}_{l+\frac{1}{2}}} \tilde{j}_{l+\frac{1}{2}}+
                O(\epsilon^2) &
                \textnormal{if} & 
                -\bar{c}_{l+\frac{1}{2}}<u_{i,l+\frac{1}{2}} < \bar{c}_{l+\frac{1}{2}}\; \\[4mm]
                q_i n_{i,l}u_{i,l+\frac{1}{2}} + q_e n_{e,l}u_{e,l+\frac{1}{2}} &  
                \textnormal{if} & 
                \bar{c}_{l+\frac{1}{2}} \le u_{i,l+\frac{1}{2}}  
            \end{array}
        \right\}  \nonumber \\[3mm]
        \label{eq: j hat tilde expressed using flux - single ion}
        &=& 
           \left\{  
            \begin{array}{ccc}
                \frac{n_{i,l+1}}{\bar{n}_{i,l+\frac{1}{2}}} \tilde{j}_{l+\frac{1}{2}}+O(\epsilon) &  
                \textnormal{if} & 
                u_{i,l+\frac{1}{2}} \le -\bar{c}_{l+\frac{1}{2}} \\[2mm]
                \displaystyle
                \left(1 + \frac{n_{i,l} - n_{i,l+1}}{n_{i,l} + n_{i,l+1}} \frac{u_{i,l+\frac{1}{2}}}{\bar{c}_{l+\frac{1}{2}}} \right) \tilde{j}_{l+\frac{1}{2}}+
                O(\epsilon) &
                \textnormal{if} & 
                -\bar{c}_{l+\frac{1}{2}}<u_{i,l+\frac{1}{2}} < \bar{c}_{l+\frac{1}{2}}\; \\[4mm]
                \frac{n_{i,l}}{\bar{n}_{i,l+\frac{1}{2}}} \tilde{j}_{l+\frac{1}{2}}+O(\epsilon) &  
                \textnormal{if} & 
                \bar{c}_{l+\frac{1}{2}} \le u_{i,l+\frac{1}{2}}  
            \end{array}
        \right\} 
    \end{eqnarray}
    It is trivial to show that the coefficient in front of $\widetilde{j}_{l+\frac{1}{2}}$ in the second condition is strictly positive. Taking the limit $\Delta t\to 0$, $\frac{\epsilon}{\Delta t}\to 0$ in Eqs.~(\ref{eq:discretized_cdf_vlasov}), (\ref{eq:discretized_cdf_vlasov_2}), (\ref{eq:1d_discrete_continuity}) -- (\ref{eq:1d_discrete_ampere}) we obtain
    $\left(\widehat{\widetilde{j}}^{\epsilon\rightarrow 0}_{l+\frac{1}{2}}\right)^{k+\frac{1}{2}}=0$ and $n^{k+1}_{\alpha, l}=n^{k}_{\alpha, l}$ and
    $u_{\alpha,l+\frac{1}{2}}^{k+1} = u_{\alpha,l+\frac{1}{2}}^{k}$. In view of (\ref{eq: j hat tilde expressed using flux - single ion}) and using similar considerations as in Proposition~\ref{proposition:equilvalence_of_numerical_current_and_solution_current_1}, the latter equation yields $\tilde{j}^{k+1}_{l+\frac{1}{2}} = -\tilde{j}^{k}_{l+\frac{1}{2}}$ in all three cases of the flux. The rest of the proposition follows similarly to Proposition~\ref{proposition:equilvalence_of_numerical_current_and_solution_current_1}.

\end{proof}
\begin{proposition}\label{proposition:invertibility_of_jacobian}
Eqs.~\eqref{eq:1d_discrete_current} and \eqref{eq:1d_discrete_ampere} implicitly define $\widetilde{j}^{k+1}_{l+1/2}$ and $E^{k+1}_{l+1/2}$ as smooth functions of the rest of the discrete variables in the limiting case $\epsilon\to 0$ for the special cases of 1) using a linear reconstruction for $\widehat{F}_{n_{\alpha},l+\frac{1}{2}}$ or 2) using the kinetic-flux reconstruction but for a single ion species. 
\end{proposition}
\begin{proof}
Taking the limit $\epsilon \rightarrow 0$ of Eqs.~\eqref{eq:1d_discrete_current} and \eqref{eq:1d_discrete_ampere} and taking the Gateaux derivatives of the resulting equations with respect to $\widetilde{j}^{k+1}_{l+1/2}$ and $E^{k+1}_{l+1/2}$, we obtain the following for the linear reconstruction case:
\begin{equation}
    \label{eq:gateaux_E_discrete_1d_current}
    \frac{\partial R^{\epsilon\rightarrow 0}_{\tilde{j},l+1/2}}{\partial E^{k+1}_{l+1/2}}
    =
    \sum^{N_s}_{\alpha} \frac{q^2_{\alpha}}{m_{\alpha}}\frac{\bar{n}^{k+1}_{\alpha,l+1/2}}{2}
\end{equation}
\begin{equation}
    \label{eq:gateau_j_discrete_1d_current}
    \frac{\partial R^{\epsilon\rightarrow 0}_{\tilde{j}_{l+1/2}}}{\partial \widetilde{j}^{k+1}_{l+1/2}}
    =
    0,
\end{equation}
\begin{equation}
    \label{eq:gateaux_e_discrete_1d_ampere}
    \frac{\partial R^{\epsilon\rightarrow 0}_{E,l+1/2}}{\partial E^{k+1}_{l+1/2}}  
    =
    0,
\end{equation}
\begin{equation}
    \label{eq:gateaux_j_discrete_1d_ampere}
    \frac{\partial R^{\epsilon\rightarrow 0}_{E,l+1/2}}{\partial \widetilde{j}^{k+1}_{l+1/2}} =
    \frac{1}{2},
\end{equation}
and for the kinetic-flux reconstruction case, the only difference is in $\partial R^{\epsilon\rightarrow 0}_{\widetilde{j}_{l+1/2}}/\partial \widetilde{j}^{k+1}_{l+1/2}$ given as
\begin{equation}
    \label{eq:gateaux_j_discrete_1d_ampere2}
    \frac{\partial R^{\epsilon\rightarrow 0}_{E,l+1/2}}{\partial \widetilde{j}^{k+1}_{l+1/2}} =
           \frac{1}{2}\left\{  
            \begin{array}{ccc}
                \frac{n^{k+1}_{i,l+1}}{\bar{n}^{k+1}_{i,l+\frac{1}{2}}} &  
                \textnormal{if} & 
                u^{k+1}_{i,l+\frac{1}{2}} \le -\bar{c}^{k+1}_{l+\frac{1}{2}} \\[2mm]
                \displaystyle
                1 + \frac{n^{k+1}_{i,l} - n^{k+1}_{i,l+1}}{n^{k+1}_{i,l} + n^{k+1}_{i,l+1}} \frac{u^{k+1}_{i,l+\frac{1}{2}}}{\bar{c}^{k+1}_{l+\frac{1}{2}}} &
                \textnormal{if} & 
                -\bar{c}^{k+1}_{l+\frac{1}{2}}<u^{k+1}_{i,l+\frac{1}{2}} < \bar{c}^{k+1}_{l+\frac{1}{2}}\; \\[4mm]
                \frac{n^{k+1}_{i,l}}{\bar{n}^{k+1}_{i,l+\frac{1}{2}}}  &  
                \textnormal{if} & 
                \bar{c}^{k+1}_{l+\frac{1}{2}} \le u^{k+1}_{i,l+\frac{1}{2}}  
            \end{array}
        \right\} 
\end{equation}
Hence, for non-vanishing $n_i$, the Jacobian is diagonal and invertible.
\end{proof}
\begin{proposition}
Eqs.~(\ref{eq:discretized_cdf_vlasov}), (\ref{eq:discretized_cdf_vlasov_2}), (\ref{eq:1d_discrete_continuity}) -- (\ref{eq:1d_discrete_ampere}) represent a discrete-time fast-slow system for the cases of 1) using a linear reconstruction for $\widehat{F}_{n_{\alpha},l+\frac{1}{2}}$ or 2) using the kinetic-flux reconstruction but for a single-ion species. Specifically, let $X$ and $Y$ be finite-dimensional spaces such that $x=\left\{ {\cal F}_{\alpha,l,p}, n_{\alpha,l},u_{i,l+\frac{1}{2}},T_{\alpha,l}\right\}\in X$ and $y=\left\{ \widetilde{j}_{l+\frac{1}{2}},E_{l+\frac{1}{2}}\right\}\in Y$.
Then 
\begin{itemize}
\item[(1)]{} There is a neighborhood of $\{ 0, 0 \}$ in the parameter variable $\{ \Delta t, \frac{\epsilon}{\Delta t} \}$ such that a Eqs.~\eqref{eq:discretized_cdf_vlasov}, \eqref{eq:discretized_cdf_vlasov_2}, \eqref{eq:1d_discrete_continuity} -- 
\eqref{eq:1d_discrete_ampere}  define a family of smooth mappings $\Phi_{\Delta t, \frac{\epsilon}{\Delta t}}:X \times Y \to 
X\times Y$ with respect to parameters $\{\Delta t, \frac{\epsilon}{\Delta t} \}$ such that $\{ x^{k+1},y^{k+1} \} = \Phi_{\Delta t, \frac{\epsilon}{\Delta t}}(x^{k},y^{k})$.
\item[(2)]{} There is a family of smooth mappings $\psi_{x}(y):Y\to Y$ smoothly parameterized by $x\in X$ such that
$\Phi_{0,0}=\{x,\psi_{x}(y)\}$.
\item[(3)]{} For each $x\in X$, the mapping $\psi_{x}(y)$ has a unique non-degenerate fixed point $y=y^{\ast}_{0}(x)$. In particular, 
$D_{y}\psi_{x}(y^{\ast}_0(x))-\mathrm{id}_{Y}$, is an invertible linear map, where $\mathrm{id}_{Y}$ is the identity transformation on $Y$. 
\end{itemize}
\end{proposition}
\begin{proof} To establish Part (1), we observe that in the case $\epsilon\gg \Delta t>0$, 
Eqs.~\eqref{eq:discretized_cdf_vlasov}, \eqref{eq:discretized_cdf_vlasov_2},  
\eqref{eq:1d_discrete_continuity} -- \eqref{eq:1d_discrete_ampere} define a smooth function $\Phi_{\Delta t, \frac{\epsilon}{\Delta t}}(x^{k},y^{k})$ implicitly for sufficiently small $\Delta t$. Indeed, in this case, the system's Jacobian approach identity as $\Delta t\to 0$. The case $\Delta t \gg \epsilon >0$ is covered by Propositions~\ref{proposition:equilvalence_of_numerical_current_and_solution_current_1} and
\ref{proposition:equilvalence_of_numerical_current_and_solution_current_2}. Indeed, in this case, 
Eqs.~\eqref{eq:discretized_cdf_vlasov}, \eqref{eq:discretized_cdf_vlasov_2}, \eqref{eq:1d_discrete_continuity}--\eqref{eq:1d_discrete_internal_energy} implicitly define the slow components of $\Phi_{\Delta t, \frac{\epsilon}{\Delta t}}(x^{k},y^{k})$ to be $Id_{x}+O(\Delta t)$. The expressions for the fast components are provided by Eqs.~\eqref{eq:multi_ion_jhattilde}, \eqref{eq: Solving for E with O(epsilon)}, and \eqref{eq: j hat tilde expressed using flux - single ion}. 

Part (2) is covered by Propositions \ref{proposition:equilvalence_of_numerical_current_and_solution_current_1} and \ref{proposition:equilvalence_of_numerical_current_and_solution_current_2}. Finally, Part (3) follows from Proposition \ref{proposition:invertibility_of_jacobian} and Eqs.~\eqref{eq: limiting case E explicit} and \eqref{eq: linear limiting case j-tilde}. 
\end{proof}

Existence of non-degenerate fixed point $y_{0}(x)$ of $\psi_{x}(y)$ such that 
$D_{y}\psi_{x}(y^{\ast}_0(x))-\mathrm{id}_{Y}$ is invertible implies existence of a unique formal slow manifold for the discrete fast-slow system \eqref{eq:discretized_cdf_vlasov}, \eqref{eq:discretized_cdf_vlasov_2},  
\eqref{eq:1d_discrete_continuity} -- \eqref{eq:1d_discrete_ampere} (see Theorem~5 in Section~7 of \cite{burby_cnsnm_slow_manifold_2020}). Moreover, the system \eqref{eq:discretized_cdf_vlasov}, \eqref{eq:discretized_cdf_vlasov_2},  \eqref{eq:1d_discrete_continuity} -- \eqref{eq:1d_discrete_ampere} implements a slow manifold integrator in the sense that for fixed $\epsilon$ and small $\Delta t$, $(x^{K+1},y^{k+1})=\Phi_{\Delta t, \frac{\epsilon}{\Delta t}} (x^{k},y^{k})$ approximates the time $\Delta t$ flow of the fast slow system \eqref{eq:fs_cdf_equation}--\eqref{eq:fs_ampere_equation} and that $(x^{k+1},y^{k+1}) = \Phi_{\Delta t, \frac{\epsilon}{\Delta t}} (x^{k},y^{k})$ is a discrete-time fast slow system. It was shown in \cite{burby_cnsnm_slow_manifold_2020} that slow manifold integrators have favorable numerical properties, including fast convergence of Picard iterations in the case of 
small $\Delta t$ and $\epsilon$.

With the existence of the discrete slow manifold, we show that the limiting form of the discrete fast solutions of our scheme is consistent with the continuum analogues. Taking the limit of $\epsilon \rightarrow0$, from Eqs. \eqref{eq:1d_discrete_current}, we obtain:
\begin{equation}
    \left(E^{\epsilon\rightarrow 0}_{l+\frac{1}{2}}\right)^{k+\frac{1}{2}}=
    \frac{\left(\frac{\widehat{F}^{\epsilon\rightarrow 0}_{\widetilde{j},l+1}-\widehat{F}^{\epsilon\rightarrow 0}_{\widetilde{j},l}}{\Delta\mathfrak{x}}\right)^{k+\frac{1}{2}}
    +
    \sum_{\alpha}^{N_{s}}\frac{q_{\alpha}}{m_{\alpha}}\left(\frac{P_{\alpha,l+1}-P_{\alpha}}{\Delta\mathfrak{x}}\right)^{k+\frac{1}{2}}}{\sum_{\alpha}^{N_{s}}\frac{q_{\alpha}^{2}\left(\bar{n}^{\epsilon\rightarrow 0}_{\alpha,l+\frac{1}{2}}\right)^{k+\frac{1}{2}}}{m_{\alpha}}}.
    \label{eq:1d_generalized_ohms_law}
\end{equation}
We recover the well-known Ohm's law for the electric field expressed in terms of the electron pressure gradient by invoking the $m_{i}/m_{e}\gg1$ ordering,
\begin{equation}
    \label{eq:1d_discrete_ohms_law}
    \left(E_{l+\frac{1}{2}}\right)^{k+\frac{1}{2}}\approx\frac{\left(P_{e,l+1}-P_{e,l}\right)^{k+\frac{1}{2}}}{q_{e}\left(\bar{n}^{\epsilon\rightarrow 0}_{e,l+\frac{1}{2}}\right)^{k+\frac{1}{2}}\Delta\mathfrak{x}}.
\end{equation}
We close this section by noting that, for the multi-ion case with the kinetic-flux reconstruction, symmetries between electrons and ions are lost in the definition of the numerical flux reconstruction. As a consequence, terms containing $\epsilon^{-1}$ do not vanish, and an explicit invertibility of the limiting Jacobian is difficult to prove. However, we demonstrate numerically in Sec. \ref{subsec:multi_ion_case} that quasi-neutrality and ambipolarity limits are still robustly recoverable.
%
%
%
%
%
%
\subsection{Invariance Preservation of ${\cal F}$\label{subsec:discrete_invariance_preservation_cdf}}

The invariance for ${\cal F}$ in Eq. \eqref{eq:cdf_invariance} must be upheld discretely to ensure consistency with the original Vlasov-Amp{\`e}re system for each species. In the discrete, Eq. \eqref{eq:cdf_time_dependent_invariance} is not generally satisfied because the discrete product rule and the chain rule are not exactly preserved. As such, in a single step of the Runge-Kutta stage, the discrete relationship of Eq. \eqref{eq:cdf_time_dependent_invariance} is only satisfied to some discretization truncation error. Hereon, unless otherwise mentioned, for brevity, we drop the explicit species index in $\vec{\cal F}$, that is, $\vec{\cal F} \equiv \vec{\cal F}_{\alpha}$.  Consider the error accumulated in the invariance measure from the first stage of an RK2 integrator (that is, a forward Euler step):
\begin{equation}
    \left\langle \left\{ 1,\mathfrak{w},\mathfrak{w}^{2}\right\} ,\vec{{\cal F}}_{l}^{\dagger}-\vec{{\cal F}}_{l}^{k}\right\rangle _{\delta w}=\Delta t\vec{\mathbb{E}}_{l}\approx{\cal O}\left(\Delta t,\Delta\mathfrak{x}^{\beta},\Delta\mathfrak{w}^{\eta}\right).
\end{equation}
Here, $\vec{\mathbb{E}}_{l}=\left\{ \mathbb{E}_{n_{l}},\mathbb{E}_{u_{l}},\mathbb{E}_{T_{l}}\right\} =\left\langle \left\{ 1,\mathfrak{w},\mathfrak{w}^2\right\} ,\vec{G}_{l}\right\rangle _{\delta w}\in\mathbb{R}^{3}$ is the discrete error in the invariances, $\mathfrak{w}^2 \equiv \mathfrak{w} \odot \mathfrak{w} \in \mathbb{R}^{N_w}$ where $\odot$ denotes an element-wise multiplication operation, $\vec{G}_{l}=\left\{ G_{l,1},\cdots,G_{l,N_{w}}\right\} \in \mathbb{R}^{N_w}$, and $\beta$ ($\eta$) is the order of discretization error in $x$ ($w$). This discrete error is projected out through an additive correction of the intermediate solution per RK stage, $\vec{{\cal F}}_{l}^{\dagger}$, for a given spatial point,
\begin{equation}
    \vec{{\cal F}}_{l}^{*} =\vec{{\cal F}}_{l}^{\dagger}+\delta\vec{{\cal F}}_{l},
\end{equation}
such that $\left\langle \left\{ 1,\mathfrak{w},\mathfrak{w}^{2}\right\} ,\vec{{\cal F}}_{l}^{*}-\vec{{\cal F}}_{l}^{k}\right\rangle _{\delta w}=\vec{0}$. Here, 
\begin{equation}
    \delta\vec{{\cal F}}_{l}=\left(c_{0,l}+\mathfrak{w}c_{1,l}+\mathfrak{w}^{2}c_{2,l}\right)\odot\vec{{\cal F}}_{l}^{\dagger}
\end{equation}
is the additive correction and $\left\{c_{0,l}, c_{1,l}, c_{2,l} \right\}$ are the projection coefficients determined by solving the following local optimization problem:
\begin{equation}
    \min_{\vec{c}_l}\vec{c}_l^{T}\cdot \mathfrak{M}\cdot\vec{c}_{l}-\vec{\lambda}\cdot\left[\vec{\mathbb{E}}_{l}-\left\langle \left\{ 1,\mathfrak{w},\mathfrak{w}^{2}\right\} ,\delta\vec{{\cal F}}_{l}\right\rangle _{\delta w}\right],
\end{equation}
$\vec{c}_{l}=\left\{ c_{0,l},c_{1,l},c_{2,l}\right\} $, and $\mathfrak{M}\in\mathbb{R}^{3\times3}$ is a diagonal penalty matrix with elements defined as $\mathfrak{M}_{ij}=\delta_{ij}w_{max}^{i}$ and is designed to penalize contributions of the linear and quadratic terms in the correction heavier, and $\vec{\lambda}\in\mathbb{R}^{3}$ is a vector of Lagrange multipliers to ensure that the additive correction projects the invariance errors from the intermediate step. This projection operation is performed at the end of each RK step and is denoted as, $\vec{\cal F}^* = {\cal P}\left( \vec{\cal F}^{\dagger}, \vec{\cal F}^{k}, \right)$ where $\vec{\cal F}^k$ is appropriately replaced with the solution at the previous RK stage. We conclude this section by noting that $\delta \vec{\cal F}_l \ll \vec{\cal F}_l$, for a sufficiently small truncation error in $x$, $w$, \emph{or a sufficiently small} $\Delta t$. As such, for explicit time integrators used for the kinetic subsystem, the solution remains positive even for problems supporting shocks where the SMART discretization asymptotically goes to an upwind discretization (first-order accurate in $x$), as will be demonstrated later. 

%
%
\section{Solver\label{sec:solver}}

The fully discretized equations, \eqref{eq:discretized_cdf_vlasov}-\eqref{eq:1d_discrete_ampere}, form a coupled nonlinear system of equations which we solve using a fixed-point iteration scheme,
\begin{equation}
    \vec{z}^{\left(s+1\right)}={\cal G}\left(\vec{z}^{\left(s\right)}\right).\label{eq:fixed_point_solver}
\end{equation}
Here, $\vec{z}=\left\{ \vec{{\cal F}},\vec{{\cal M}},\vec{E}\right\} \in\mathbb{R}^{N_z}$ is the discrete solution vector as defined earlier, $N_z = N_{sp}N_x N_w + 3N_{sp}N_x + N_x$ is the total number of unknowns, $s$ in parentheses denotes the fixed point iteration index, and ${\cal G}:\mathbb{R}^{N_z} \rightarrow \mathbb{R}^{N_z}$ is the fixed-point map. For this study, we employ a two-stage block solver for the fixed-point map, which solves the kinetic equation and the moment-field subsystem in a Picard linearized manner. The fixed-point iteration continues until $\frac{\left|\vec{z}^{\left(s\right)}-\vec{z}^{\left(s-1\right)}\right|}{N_z}\le\epsilon_{z}$; and unless otherwise specified, $\epsilon_{z}=10^{-6}$ is the relative convergence tolerance for the fixed-point solver. Note that for the first fixed-point iteration, the solution is initialized from the previous time-step solution, $\vec{z}^{\left(0\right)}=\vec{z}^{k}$.

For the moment-field subsystem, a quasi-Newton (QN) solver is used to solve the nonlinear system, 
\begin{equation}
    \label{eq:moment_field_residual_vector}
    \vec{R} = \left\{ \vec{R}_{n},\vec{R}_{u},\vec{R}_{\widetilde{j}},\vec{R}_{T},\vec{R}_{E}\right\} =\vec{0},     
\end{equation}
where $\vec{R}_{n}=\left\{ \vec{R}_{n_{1}},\cdots,\vec{R}_{N_{s}}\right\} \in\mathbb{R}^{N_{s}N_{x}}$ is the nonlinear residual vector for continuity equations, $\vec{R}_{u}=\left\{ \vec{R}_{u_{1}},\cdots,\vec{R}_{u_{N_{i}}}\right\} \in\mathbb{R}^{N_{i}N_{x}}$ is for the ion momentum equations, $\vec{R}_{\widetilde{j}} \in \mathbb{R}^{N_{x}}$ is for the current equation, $\vec{R}_{T}=\left\{ \vec{R}_{T_{1}},\cdots,\vec{R}_{T_{N_{s}}}\right\} \in\mathbb{R}^{N_{s}N_{x}}$ is for the internal energy equations, and $\vec{R}_{E}\in\mathbb{R}^{N_{x}}$ is for Amp{\`e}re equation. With $\vec{M}_{\left(0\right)}\equiv\left\{ \vec{{\cal M}},\vec{E}\right\} ^{\left(s\right)}$ provided as the initial guess for the QN solver and the heat flux Picard linearized from the previous fixed-point iteration, $\vec{Q}^{\left(s\right)}=\vec{Q}\left[\vec{{\cal F}}^{\left(s\right)}\right]\in\mathbb{R}^{N_{s}N_{x}}$, the $r^{th}$ iteration QN solve is defined as:
\begin{equation}
    \vec{M}_{\left(r+1\right)}=\vec{M}_{\left(r\right)}+\delta\vec{M}_{(r)}.
\end{equation}
Here, $\delta\vec{M}_{(r)}$ is the QN update, obtained by solving the following quasi-linearized system at each grid point:
\begin{equation}
    \label{eq:1d_quasi_linear_continuity}
    \frac{\delta n_{\alpha,l,(r)}}{\Delta t}=-R_{n_{\alpha},l}\left(\vec{M}_{\left(r\right)} ; \vec{z}^{k}\right)\;\;\forall\;\; \alpha\in\left\{1,\cdots,N_s\right\},
\end{equation}
\begin{equation}
    \label{eq:1d_quasi_linear_momentum}
    \frac{\bar{n}_{i,l+\frac{1}{2},(r)}\delta u_{i,l+\frac{1}{2},(r)}}{\Delta t}=-R_{u_{i},l+\frac{1}{2}}\left(\vec{M}_{\left(r\right)} ; \vec{z}^{k}\right) \;\;\forall\;\; i\in\left\{1,\cdots,N_i\right\},
\end{equation}
\begin{equation}
    \label{eq:1d_quasi_linear_energy}
    \frac{n_{\alpha,l,(r)}\delta T_{\alpha,l,(r)}}{2\Delta t}=-R_{T_{\alpha},l}\left(\vec{M}_{\left(r\right)},\vec{Q}^{\left(s\right)} \left[ \vec{\cal F}^{\left(s\right)} \right] ; \vec{z}^{k}\right) \;\;\forall\;\; \alpha\in\left\{1,\cdots,N_s\right\} ,
\end{equation}
\begin{equation}
    \label{eq:qn_current}
    \epsilon\frac{\delta\widetilde{j}_{l+\frac{1}{2},(r)}}{\Delta t}-\frac{\delta E_{l+\frac{1}{2},(r)}}{4}\sum_{\alpha}^{N_{s}}\frac{q_{\alpha}^{2}\left(\bar{n}_{\alpha,l+\frac{1}{2},(r)}+\bar{n}_{\alpha,l+\frac{1}{2}}^{k}\right)}{m_{\alpha}}=
    -R_{\widetilde{j},l+\frac{1}{2}}\left(\vec{M}_{\left(r\right)} ; \vec{z}^{k}\right),
\end{equation}
\begin{equation}
    \label{eq:qn_ampere}
    \epsilon\frac{\delta E_{l+\frac{1}{2},(r)}}{\Delta t}+\frac{\delta\widetilde{j}_{l+\frac{1}{2},(r)}}{2}=-R_{E,l+\frac{1}{2}}\left(\vec{M}_{\left(r\right)} ; \vec{z}^{k} \right).
\end{equation}
The specific choice of quasi-linearization is motivated by the existence of the discrete slow manifold discussed in Section \ref{subsubsec:discrete_qn_ap}, which recovers the quasi-neutral asymptotic limit when $\epsilon\rightarrow 0$. We observe that the quasi-linearized operator for the fast subsystem in Eqs. \eqref{eq:qn_current} and \eqref{eq:qn_ampere} closely resemble the Jacobian of the limiting form in Eqs. \eqref{eq:gateaux_E_discrete_1d_current}
to \eqref{eq:gateaux_j_discrete_1d_ampere} and shares similarity to the so-called \emph{physics-based preconditioners} in Refs. \cite{taitano_jcp_2015_conservative_vfpa_1, taitano_sisc_2013, chen_jcp_2011_implicit_pic, chen_jcp_2014_pc_es_pic} which captures the stiff electron plasma waves to improve the efficiency of the underlying Newton-Krylov solvers. We demonstrate the effectiveness of this quasi-linearization in controlling the number of QN iterations in Sec. \ref{sec:numerical_results}. Specifically, we demonstrate that, despite the lack of an explicit proof of discrete slow manifold for the multi-ion case using the kinetic-flux reconstruction, the choice of quasi-linearization leads to a robust moment-field solver.

The QN iteration for the moment-field subsystem is continued until $\left| \vec{R}_{(r)}\right| \le \epsilon_{rel} \left| \vec{R}_{(0)} \right|$, where $\epsilon_{rel} = 10^{-6}$ is the relative convergence tolerance used in this study unless otherwise specified. After the moment-field system is solved for $\vec{M}^{\left(s+1\right)} \equiv \left\{ \vec{{\cal M},}\vec{E}\right\} ^{\left(s+1\right)} := \vec{M}_{\left(r\right)}$, the updated moments are used to update the kinetic solution using a second order Runge-Kutta integrator as $\vec{\cal F}^{*} = G\left(\vec{{\cal F}}^{k},\vec{{\cal M}}^{\left(s+1\right)},\vec{{\cal M}}^{k} \right)$ and $\vec{\cal F}^{\left(s+1\right)} = G\left(\vec{{\cal F}}^{*},\vec{{\cal M}}^{\left(s+1\right)},\vec{{\cal M}}^{k}\right)$ with the invariance-preserving projection discussed in Sec. \ref{subsec:discrete_invariance_preservation_cdf} performed at each step. A flow diagram illustrating the evaluation of the fixed-point map for a given iteration, $s$, is shown in Figure \ref{fig:fixed_point_solver}.
\begin{figure}[th]
    \begin{centering}
    \includegraphics[width=0.8\textwidth]{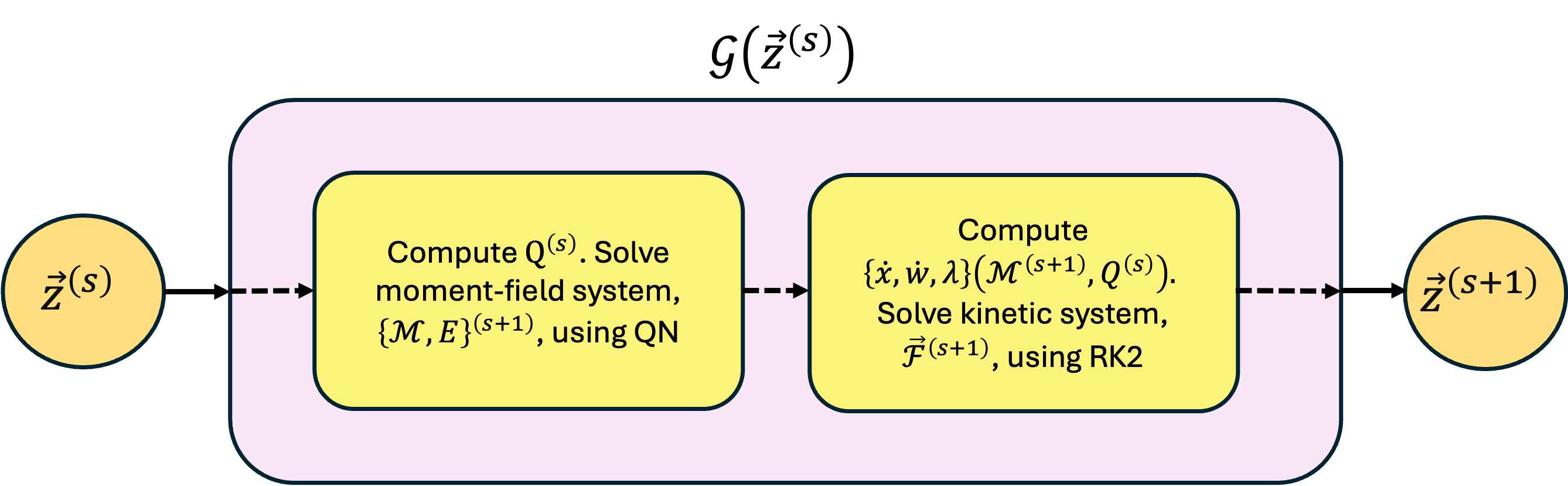}
    \par\end{centering}
    \caption{Illustration of the fixed point solver for iteration $s$. The moment-field subsystem is solved using a QN solver. With the updated moment-field system, the kinetic subsystem is updated with RK2.}
    \label{fig:fixed_point_solver}
\end{figure}
%
%
%
%
%
%
\section{Numerical Results\label{sec:numerical_results}}

The proposed discretization and solver, along with their efficacy, are demonstrated on a series of canonical test problems with increasing complexity. Unless otherwise stated, the following parameters are used for all problems: $\epsilon=1$, $N_{x}=32$, $N_{w}=128$, $L_{x}=4\pi$, $w_{max} = 7$, $n_{\alpha}^{0}\equiv n_{\alpha} \left(t=0,x\right) = 1+\delta_{n_{\alpha}}\sin\left(k_{x}x\right)$, $u_{\alpha}^{0}\equiv u_{\alpha} \left(t=0,x\right) = \widetilde{u}_{\alpha} + \delta_{u_{\alpha}} \sin\left(k_{x}x\right)$, $T_{\alpha}^{0}\equiv T_{\alpha}\left(t=0,x\right)=\widetilde{T}_{\alpha}+\delta_{T_{\alpha}}\sin\left(k_{x}x\right)$, $k_{x}=\frac{2\pi}{L_{x}}$, and $\widetilde{j}^0 \equiv \widetilde{j}\left(t=0,x\right) = \epsilon^{-1} \sum^{N_s}_{\alpha}q_{\alpha}n^0_{\alpha} u^0_{\alpha}$. Here, the tilde denotes (except for $\widetilde{j}$) the initial amplitude, $\delta_{n_{\alpha}}$, $\delta_{u_{\alpha}}$, $\delta_{T_{\alpha}}$ denotes the initial perturbation amplitude for the moment quantities for species $\alpha$. The initial electric field is evaluated using the Poisson equation,
\begin{equation}
    -\epsilon^{2}\frac{\partial^{2}\phi}{\partial x^{2}}=\sum_{\alpha}^{N_{s}}q_{\alpha}n_{\alpha}^{0}.
\end{equation}
Here, $\phi\left(x\right)$ is the electrostatic potential. For discrete purpose, $\phi$ is solved on cell centers and the electric field is computed using a backward difference on cell faces: $E_{l+\frac{1}{2}}^{0}\equiv E_{l+\frac{1}{2}}\left(t=0\right)=-\frac{\left(\phi_{l+1}-\phi_{l}\right)}{\Delta \mathfrak{x}}$. Finally, the initial distribution function is defined as a Gaussian everywhere, ${\cal F}_{\alpha}^{0}\equiv{\cal F}_{\alpha}\left(t=0,x,w\right)=\frac{\exp\left(-w^{2}\right)}{\sqrt{\pi}}$.
\subsection{Free Streaming Test \label{subsec:neutral_gas}}
A single species, neutral, free-streaming gas is considered with the discrete-moment equation identical to Eq. \eqref{eq:1d_discrete_continuity}-\eqref{eq:1d_discrete_kinetic_energy_residual}, but without electrostatic acceleration, Amp{\`e}re's equation, nor non-Amp{\`e}rean current contribution. The purpose of this problem is to highlight the critical importance of ensuring the invariance of ${\cal F}$ in the discrete, Eq. \eqref{eq:cdf_time_dependent_invariance}. For
this problem, we consider the following parameters $m=1$, $L_{x}=1$, $w_{max}=6$, $N_{x}=128$, $N_{w}=128$, $\delta_{n_{\alpha}}=\delta_{u_{\alpha}}=0.2$,
$\delta_{T_{\alpha}}=0$, a simulation time of $t_{max}=0.1$ and a time step size, $\Delta t=5\times10^{-4}$. In Figure \ref{fig:free_streaming_neutral_pdf_wh_projection},
we show the reconstructed distribution function (in $\left\{ x,v\right\} $) at $t=0.025$, $0.05,$ and $0.1$ with the projection technique discussed
in sec. \eqref{subsec:discrete_invariance_preservation_cdf}. 
\begin{figure}[th]
    \begin{centering}
    \includegraphics[width=0.33\textwidth]{Figures/neutral_gas_kinetics/wh_proj_t=0_025}\includegraphics[width=0.33\textwidth]{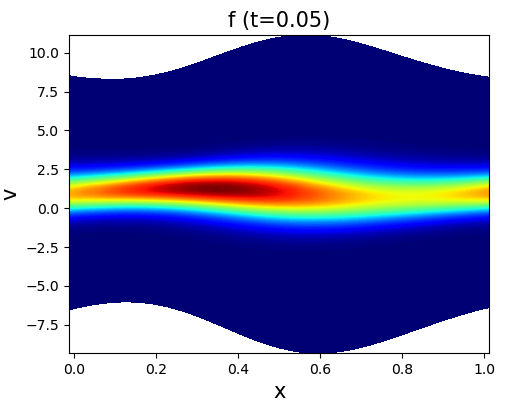}\includegraphics[width=0.33\textwidth]{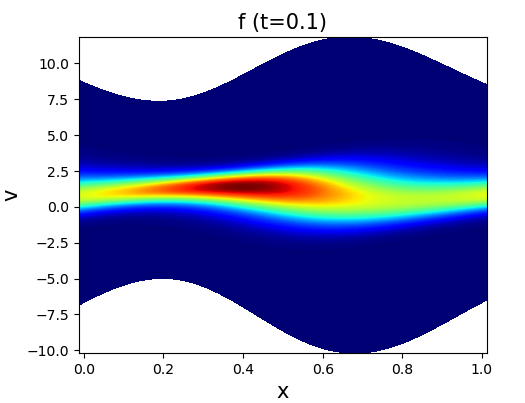}
    \par\end{centering}
    \caption{Free streaming test: The reconstructed PDF in $\left\{ x,v\right\} $ at times $t=0.025$ (left), $t=0.05$ (center), and $t=0.1$ (right) with the use of invariance enforcing projection.\label{fig:free_streaming_neutral_pdf_wh_projection}}
\end{figure}
As can be seen, the expected solution based on a free streaming flow is recovered. In Figure \ref{fig:invariance_error_vs_time}, an integrated measure of the discrete error in the invariances as a function of time,
\begin{equation}
    {\cal E}_{0}^{k}=\frac{\sum_{l=1}^{N_{x}}\Delta\mathfrak{x}\left\langle 1,\vec{{\cal F}}_{l}^{k}-\vec{{\cal F}}_{l}^{0}\right\rangle _{\delta w}}{\sum_{l=1}^{N_{x}}\Delta\mathfrak{x}\left\langle 1,\vec{{\cal F}}\right\rangle _{\delta w}},
\end{equation}
\begin{equation}
    {\cal E}_{1}^{k}=\sum_{l=1}^{N_{x}}\Delta\mathfrak{x}\left\langle \mathfrak{w},\vec{{\cal F}}_{l}^{k}-\vec{{\cal F}}_{l}^{0}\right\rangle _{\delta w},
\end{equation}
\begin{equation}
    {\cal E}_{2}^{k}=\frac{\sum_{l=1}^{N_{x}}\Delta\mathfrak{x}\left\langle \mathfrak{w}^{2},\vec{{\cal F}}_{l}^{k}-\vec{{\cal F}}_{l}^{0}\right\rangle _{\delta w}}{\sum_{l=1}^{N_{x}}\Delta\mathfrak{x}\left\langle \mathfrak{w}^{2},\vec{{\cal F}}_{l}^{0}\right\rangle _{\delta w}}.
\end{equation}
is shown and that the invariances are ensured to machine precision using the projection technique.
\begin{figure}[th]
    \begin{centering}
    \includegraphics[width=1\textwidth]{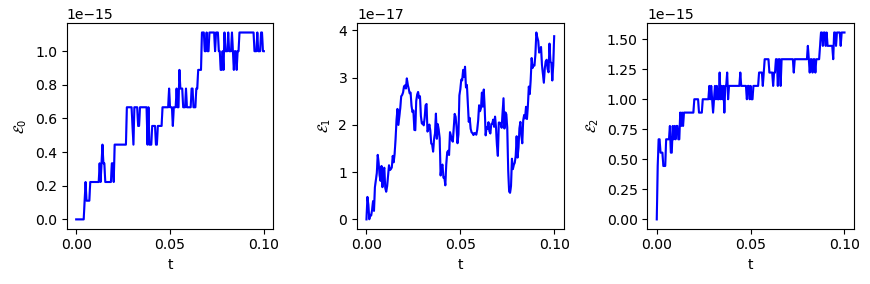}
    \par\end{centering}
    \caption{Free streaming test: ${\cal E}_{0}$, ${\cal E}_{1}$, ${\cal E}_{2}$ as a function of time. 
    \label{fig:invariance_error_vs_time}}
\end{figure}
These invariances ensure the nonlinear consistency between the dynamics of the kinetic system in the original and transformed representation and is critical to maintain long time stability of our solution. In Figure \ref{fig:free_streaming_neutral_pdf_no_projection}, we show the same quantity, but without projection operation. 
\begin{figure}[th]
    \begin{centering}
        \includegraphics[width=0.33\textwidth]{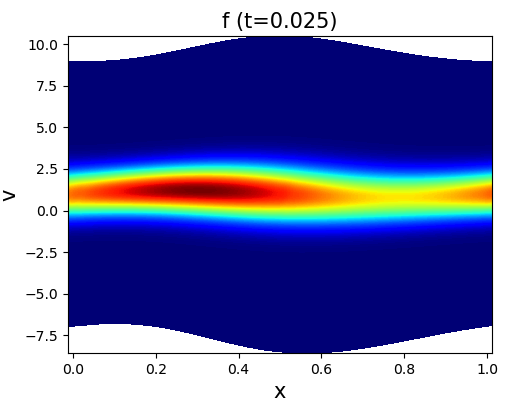}\includegraphics[width=0.33\textwidth]{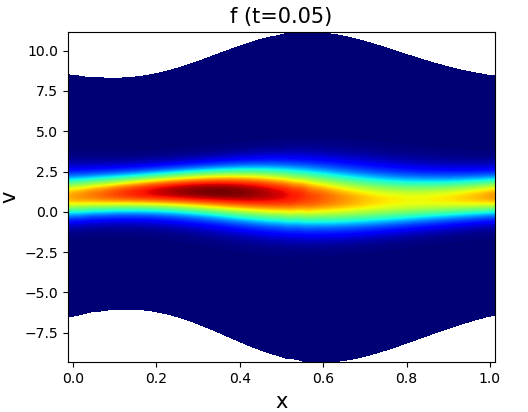}\includegraphics[width=0.33\textwidth]{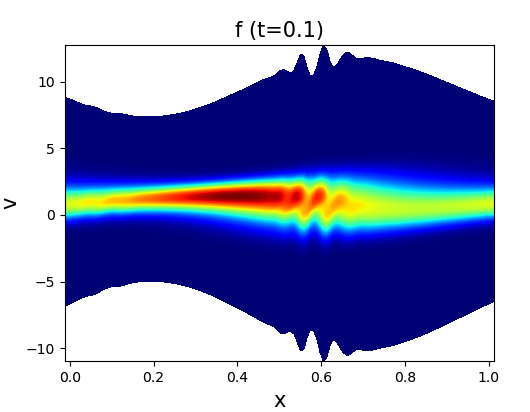}
    \par\end{centering}
    \caption{Free streaming test: The PDF reconstructed in $\left\{ x,v\right\} $ at times $t=0.025$ (left), $t=0.05$ (center), and $t=0.1$ (right)
    without the use of projection that enforces the invariance. 
    \label{fig:free_streaming_neutral_pdf_no_projection}}
\end{figure}
As can be seen, if the discrete invariance of ${\cal F}$ is not upheld, nonlinear inconsistency grows and unphysical modes are excited, which eventually leads to catastrophic numerical instabilities. Due to the highly nonlinear nature of the transformed equations --even though the original equation is perfectly linear-- the analysis is challenging and is left for future work.

%
%
\subsection{Landau Damping\label{subsec:weak_landau_damping}}

The weak Landau damping case is used to study collisionless damping of electric field energy. The conversion rate of field energy to plasma wave energy is determined through the linear dispersion relationship,
\begin{equation}
1+\frac{1}{k^{2}}\left[1+\frac{\omega}{\sqrt{2}k}Z\left(\frac{\omega}{\sqrt{2}k}\right)\right]=0.
\end{equation}
Here, $k$ is the wave vector of the perturbation, $\omega$ is the wave frequency, and $Z$ is the dispersion function of Fried and Conte. Solving for $\omega$, a complex solution is obtained in terms of the oscillatory (real) $\widetilde{\omega}$ and decaying (imaginary) $\gamma$ components of the frequency. The purpose of this test problem is to test the ability of the proposed approach to recover highly sensitive resonant behaviors and to demonstrate accuracy properties. We consider a proton-electron plasma with $m_{i}=1836$, $m_{e}=1$, $\epsilon=1$, $N_{x}=128,$ $N_{w}=256,$ $L_{x}=4\pi$, $\widetilde{u}_{i} = \widetilde{u}_{e}=0$, $\widetilde{T}_{i} = \widetilde{T}_{e}=1$, $\delta_{n_{i}} = \delta_{u_{i}} = \delta_{u_{e}}=\delta_{T_{i}}=\delta_{T_{e}}=0$, $\delta_{n_{e}}=0.01$ $\Delta t=10^{-2}$ and $t_{max}=60$. For this setup, the analytical damping rate of the electric field energy, ${\cal E}_{E} = \sum_{l=1}^{N_{x}} \Delta \mathfrak{x} \frac{\epsilon^{2}E_{l+\frac{1}{2}}^{2}}{2}$, is given as $2\gamma=-0.310$ and the simulation results are shown in Ref. \ref{fig:weak_landau_damping_rate}. 
\begin{figure}
    \begin{centering}
        \includegraphics[width=0.4\textwidth]{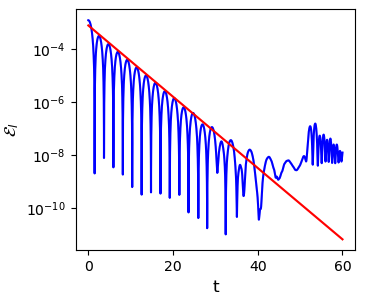}
    \par\end{centering}
    \caption{Weak Landau Damping: The energy of the electric field as a function of time.    The blue line denotes the simulation results and the red line denotes the theoretical damping rate $2\gamma=-0.31$. \label{fig:weak_landau_damping_rate}}
\end{figure}
As can be seen, the simulated damping rate agrees excellently with the analytical theory. 

We demonstrate the order accuracy of the proposed approach by performing a grid and time convergence study. For the spatial grid convergence study, we fix $\Delta t=10^{-3}$, $t_{max}=1$, and $N_{w}=256$ while varying $N_{x}=16,32,64,128,256,512$ with a reference solution obtained from $N_{x}=1024$. For the velocity grid convergence study, we fix $\Delta t=10^{-3}$, $t_{max}=1$, and $N_{x}=32$ while varying $N_{w}=16,32,64,128,256,512$ with a reference solution obtained from
$N_{w}=1024$. For the time-convergence study, we fix $N_{x}=32$, $N_{w}=32$, $t_{max}=1$, while varying $\Delta t = 10^{-2}, 4\times10^{-3}, 2\times10^{-3}, 10^{-3},5\times10^{-4}$ and a reference solution obtained from $\Delta t=10^{-4}$. The error is measured on the 2-norm of the temperatures,
\begin{equation}
    {\cal E}^{\Delta\mathfrak{x}}=\sum_{\alpha}^{N_{s}}\sqrt{\sum_{l}^{N_{x}}\Delta\mathfrak{x}\left(T_{\alpha,l}^{\Delta\mathfrak{x}}-T_{\alpha,l}^{\Delta\mathfrak{x},ref}\right)^{2}},
\end{equation}
\begin{equation}
    {\cal E}^{\Delta\mathfrak{w}}=\sum_{\alpha}^{N_{s}}\sqrt{\sum_{l}^{N_{x}}\Delta\mathfrak{x}\left(T_{\alpha,l}^{\Delta\mathfrak{w}}-T_{\alpha,l}^{\Delta\mathfrak{w},ref}\right)^{2}},
\end{equation}
\begin{equation}
    {\cal E}^{\Delta t}=\sum_{\alpha}^{N_{s}}\sqrt{\sum_{l}^{N_{x}}\Delta\mathfrak{x}\left(T_{\alpha,l}^{\Delta t}-T_{\alpha,l}^{\Delta t,ref}\right)^{2}},
\end{equation}
where superscript $ref$ denotes the reference solution obtained for each convergence study and superscripts $\Delta \mathfrak{x}$, $\Delta \mathfrak{w}$, $\Delta t$ represent the coarse grid (time-step) cases. In Figure \ref{fig:weak_landau_damping_convergence_study}, the results of the convergence test are summarized. 
\begin{figure}[th]
    \begin{centering}
        \includegraphics[width=0.33\textwidth]{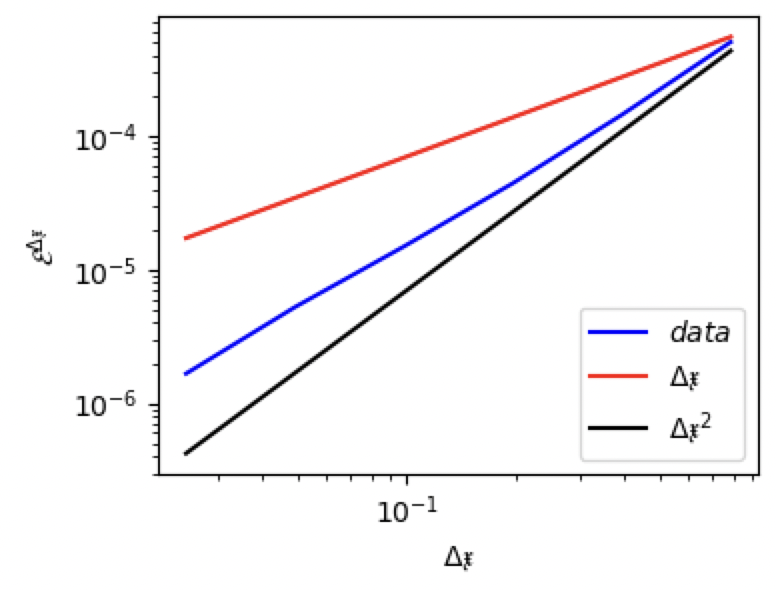}\includegraphics[width=0.33\textwidth]{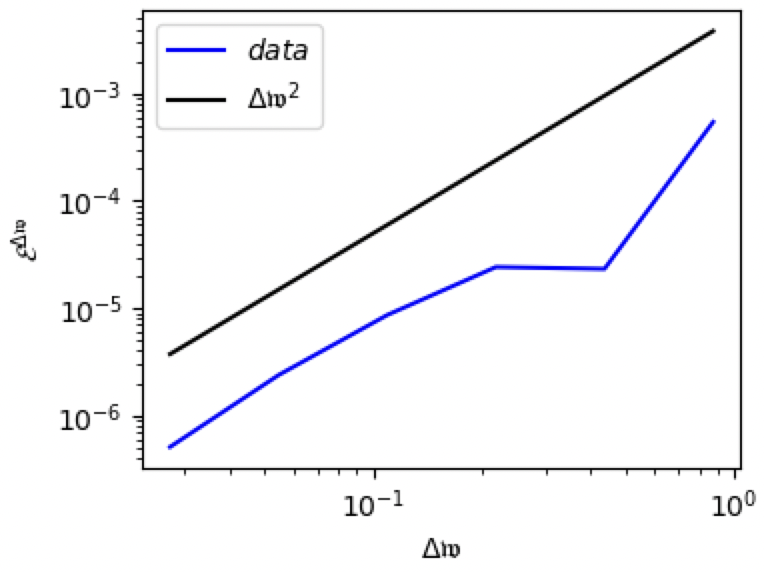}\includegraphics[width=0.33\textwidth]{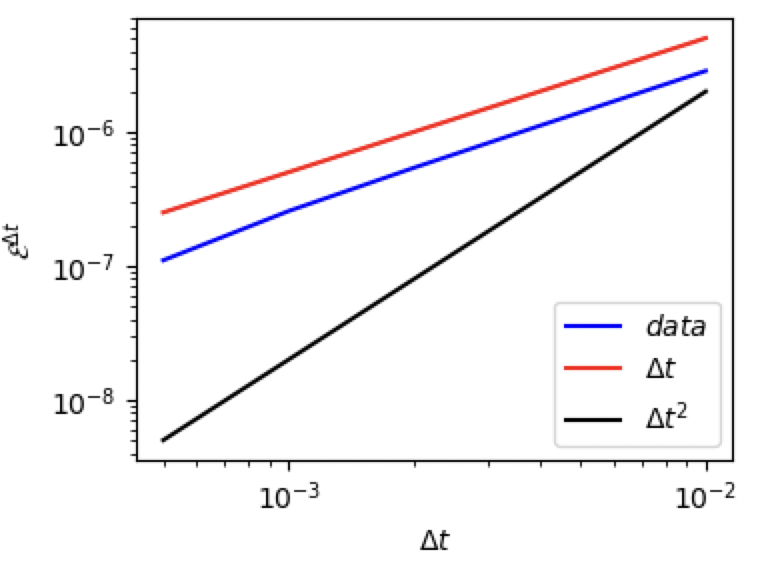}
    \par\end{centering}
    \caption{Weak Landau Damping: The results of the convergence study for $\Delta\mathfrak{x}$ (left), $\Delta\mathfrak{w}$ (center) and $\Delta t$ (right) are shown.
    \label{fig:weak_landau_damping_convergence_study}}
\end{figure}
As can be seen, second-order convergence is recovered in $w$, but only first-order convergence is recovered in $t$ and $x$. The first order accuracy in $t$ is expected as the mixed integration of RK2 for ${\cal F}$ with time-centered coefficients of $\vec{{\cal M}}$ does not guarantee a second-order accuracy. A lower order accuracy in $x$ is also expected since the flux reconstruction technique for the moment equations is adopted from Ref. \cite{goudon_fvca8_2017_staggered_euler_scheme_staggered} and relies on a first-order shock capture scheme that dominates the error. The low order time accuracy can be remedied by using a consistent and high-order accurate time integration scheme such as the implicit midpoint rule for both subsystems. However, such an integrator requires a non-linearly implicit treatment of the high-dimensional kinetic equation that can deal with stiff advection timescales. Similarly, high-order accuracy can be achieved in space by employing a high-order conservative and monotone-preserving staggered differencing scheme for the fluid subsystem, and both aspects are left for future study. 

A strong non-linear Landau damping case is tested to study the accuracy of the solver with strong perturbations. We follow the same setup as the weak Landau damping case, but with $N_{x}=128$, $N_{w}=512$, $w_{max}=7$, and $\delta_{n_{e}}=0.5$. In Figure \ref{fig:strong_landau_damping_rate}, the energy of the electric field is shown as a function of time. 
\begin{figure}[th]
    \begin{centering}
        \includegraphics[width=0.4\textwidth]{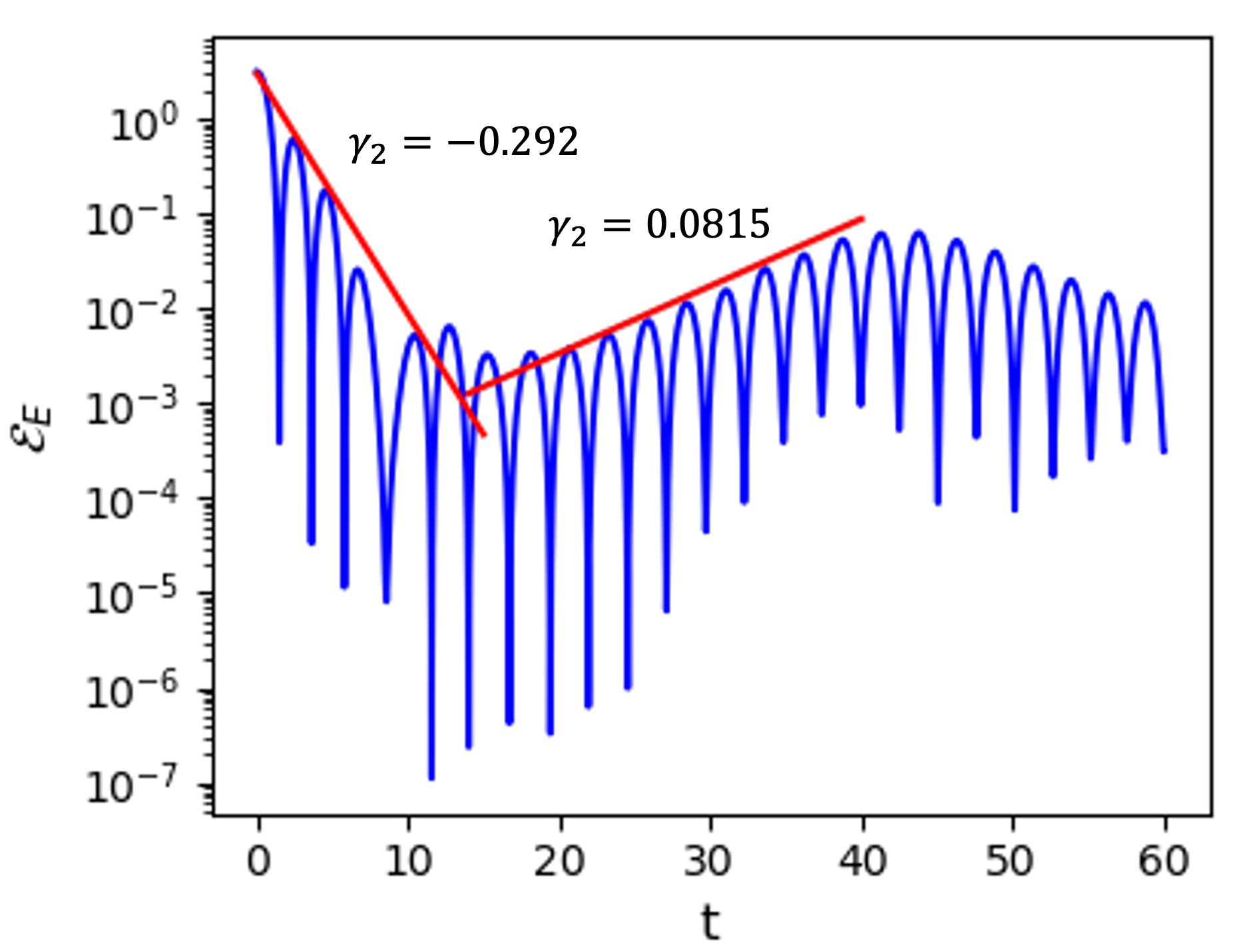}
    \par\end{centering}
    \caption{Strong Landau Damping: Field energy as a function of time.
    \label{fig:strong_landau_damping_rate} }
\end{figure}
As can be seen, the distinct damping and growth phases are captured with rates consistent with those reported in Refs. \cite{taitano_jcp_2015_conservative_vfpa_1, rossmanith_jcp_2011_positive_ho_sl_dg_vp}. 

%
%
\subsection{Ion Acoustic Shock Wave (IASW)\label{subsec:ion_acoustic_shockwave_problem}}

We model the IASW problem \cite{shay2007_jcp_ef_pi_vp, chen_jcp_2011_implicit_pic, taitano_jcp_2015_conservative_vfpa_1} to demonstrate the integrated capability of the proposed algorithm to simultaneously satisfy 1) conservation of mass, momentum, and energy; 2) Gauss law, 3) quasi-neutral asymptotic preservation, and 4) positivity of the PDF on a challenging collisionless electrostatic shock problem. For this problem, the parameters are chosen as $m_{i}=1$, $m_{e}=1/1836$, $L_{x}=144\lambda_{D}$, $\epsilon=\frac{1}{36}$, $w_{max}=8.5$, $N_{x}=1024$, $N_{w}=256$,
$\Delta t=\frac{\epsilon}{\sqrt{1836}}=6.48\times10^{-4}$, $t_{max}=5000\omega_{p,e}^{-1}$. $\delta_{n_{i}}=0.2$, $\delta_{n_{e}}=0.2\left(1-k^{2}\epsilon^{2}\right)$,
$\widetilde{u}_{e}=\widetilde{u}_{i}=-1$, $\delta_{u_{i}}=\delta_{u_{e}}=0.2$, $\widetilde{T}_{i}=0.05$, $\widetilde{T}_{e}=1$, $\delta_{i}=\delta_{e}=0$.
Here, $k_{x}=\frac{2\pi}{L_{x}}$ is the wave vector of the perturbation and the problem is initialized such that the bulk motion of the plasma is evolved in the reference of the propagation speed of the acoustic shock front of ions, $c_{s}\approx\sqrt{\widetilde{T}}\approx1$. Furthermore, the ion temperature is set much lower than that of the electrons to prevent ion Landau damping from dissipating the wave.  In Figure \ref{fig:iasw_density_field}, we show the number densities and the electric field at $t_{max}$.
\begin{figure}[th]
    \begin{centering}
        \includegraphics[width=0.7\textwidth]{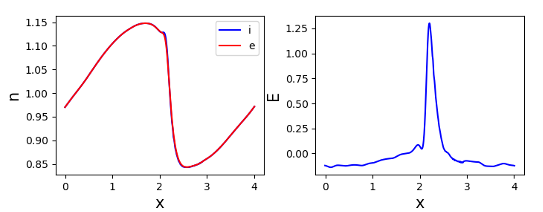}
    \par\end{centering}
    \caption{IASW: number densities for ions/electrons, $i$/$e$, (left) and the
    electric field (right) at the $t=t_{max}$.
    \label{fig:iasw_density_field}}
\end{figure}
As can be seen, the initial sinusoidal density profile leads to the formation of a shock structure with charge separation observed near the shock front ($x\approx2.25)$ on the scale of $\lambda_{D}=\epsilon$. A strong field forms near the shock front, followed by the formation of a soliton-like structure that propagates away from the front. Such strong electric fields lead to a rapid acceleration of ions, and classic wave-breaking features are seen in the ion distribution function, as reported similarly in the literature; refer to Figure \ref{fig:iasw_pdf}. 
\begin{figure}[th]
    \begin{centering}
    \includegraphics[width=0.8\textwidth]{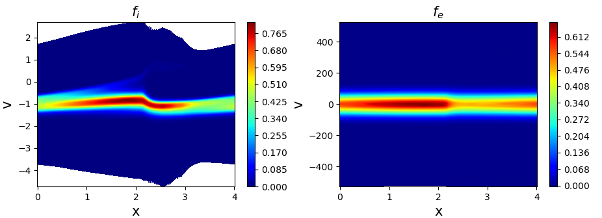}
     \par\end{centering}
    \caption{IASW: Distribution function of ions (left) and electrons (right).
    \label{fig:iasw_pdf}}
\end{figure}
In Figure \ref{fig:iasw_conservation_gauss} we show the simulation error in number density conservation,
\begin{equation}
    \mathbb{E}_{n}^{k}=\frac{\left|M^{k}-M^{0}\right|}{M^{0}},
\end{equation}
momentum conservation,
\begin{equation}
    \mathbb{E}_{\mathbb{P}}^{k}=\frac{\left|\mathbb{P}^{k}-\mathbb{P}^{0}\right|}{\left|\mathbb{P}^{0}\right|}
\end{equation}
energy conservation,
\begin{equation}
    \mathbb{E}_{{\cal E}}^{k}=\frac{\left|{\cal E}^{k}-{\cal E}^{0}\right|}{{\cal E}^{0}},
\end{equation}
and quality of Gauss law preservation,
\begin{equation}
    \mathbb{E}_{GL}^{k}=\left|\sum_{l}^{N_{x}}\left(\epsilon^{2}\frac{E_{l+\frac{1}{2}}^{k}-E_{l-\frac{1}{2}}^{k}}{\Delta \mathfrak{x}}-\sum_{\alpha}^{N_{s}}q_{\alpha}n_{\alpha,l}^{k}\right)\right|,
\end{equation}
where $M^{k}=\sum_{\alpha}^{N_{s}}\sum_{l}^{N_{x}}\Delta \mathfrak{x}m_{\alpha}n_{\alpha,l}^{k}$ is the total mass at time step $k$, $\mathbb{P}^{k}=\sum_{\alpha}^{N_{s}}m_{\alpha}\sum_{l}^{N_{x}}\Delta \mathfrak{x}\bar{n}_{\alpha,l+\frac{1}{2}}^{k}u_{\alpha,l+\frac{1}{2}}^{k}$ is the total momentum, and ${\cal E}^{k}=\sum_{\alpha}^{N_{s}}\left[\sum_{l}^{N_{x}}\Delta \mathfrak{x}\left(\frac{m_{\alpha}}{2}\bar{n}_{\alpha,l+\frac{1}{2}}^{k}\left(u_{\alpha,l+\frac{1}{2}}^{k}\right)^{2}+\frac{n_{\alpha,l}^{k}T_{\alpha,l}^{k}}{2}\right)\right]+\frac{\epsilon^{2}}{2}\sum_{l}^{N_{x}}\Delta \mathfrak{x}\left(E_{l+\frac{1}{2}}^{k}\right)^{2}$
is the total energy.
\begin{figure}[th]
    \begin{centering}
        \includegraphics[width=0.7\textwidth]{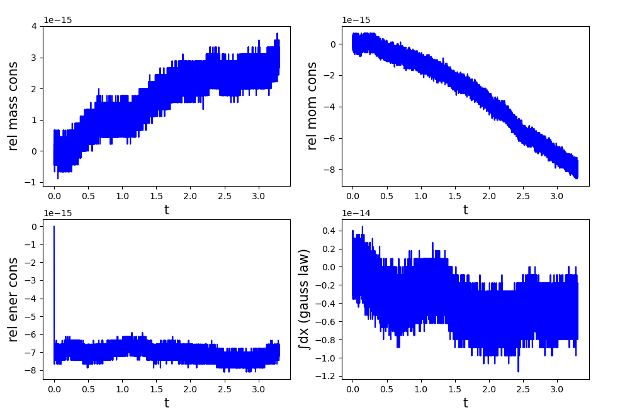}
    \par\end{centering}
    \caption{IASW: The quality of discrete conservation of mass (top-left), momentum (top-right), energy (bottom-left), and preservation of the Gauss law (bottom-right).
    \label{fig:iasw_conservation_gauss} }
\end{figure}
As can be seen, all of the critical physical invariance quantities of interest are preserved to machine precision. We also show in Figure \ref{fig:iasw_positivity} that our method ensures that the distribution function is positive at all times by plotting $\min\left(\vec{\cal F}\right)$.
\begin{figure}[th]
    \begin{centering}
        \includegraphics[width=0.4\textwidth]{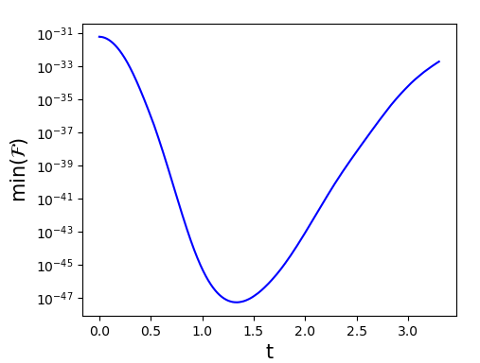}
    \par\end{centering}
    \caption{IASW: Positivity of the distribution function. 
    \label{fig:iasw_positivity}}
\end{figure}

We test the ability of the proposed formulation and discretization to recover the quasi-neutral asymptotic limit. For this purpose, we set $\epsilon=10^{-7}$, while using the same $\Delta t$ as the original setup. In this limit, from the Gauss law $\rho^{k}=\epsilon^{2}\frac{E^{k}_{l+1/2} - E^{k}_{l-1/2}}{\Delta \mathfrak{x}}\approx{\cal O}\left(\epsilon^{2}\right)$ and from the Amp{\`e}re's equation, $j^{k+1/2}_{l+1/2}=-\epsilon^{2}\frac{E^{k+1}_{l+1/2} - E^{k}_{l+1/2}}{\Delta t}\approx{\cal O}\left(\epsilon^{2}\right)$, which robustly places us in the quasi-neutral limit and the Ohm's law, Eq. \eqref{eq:1d_discrete_ohms_law}, should be valid to approximate the electric field, $E^{k+1/2}_{l+1/2}$, to ${\cal O}\left(\frac{m_{e}}{m_{i}}\right)$. In Figure \ref{fig:iasw_ap_property}, we show the quality of quasi-neutral asymptotic preservation at
$t_{max}$. 
\begin{figure}[th]
    \begin{centering}
        \includegraphics[width=0.9\textwidth]{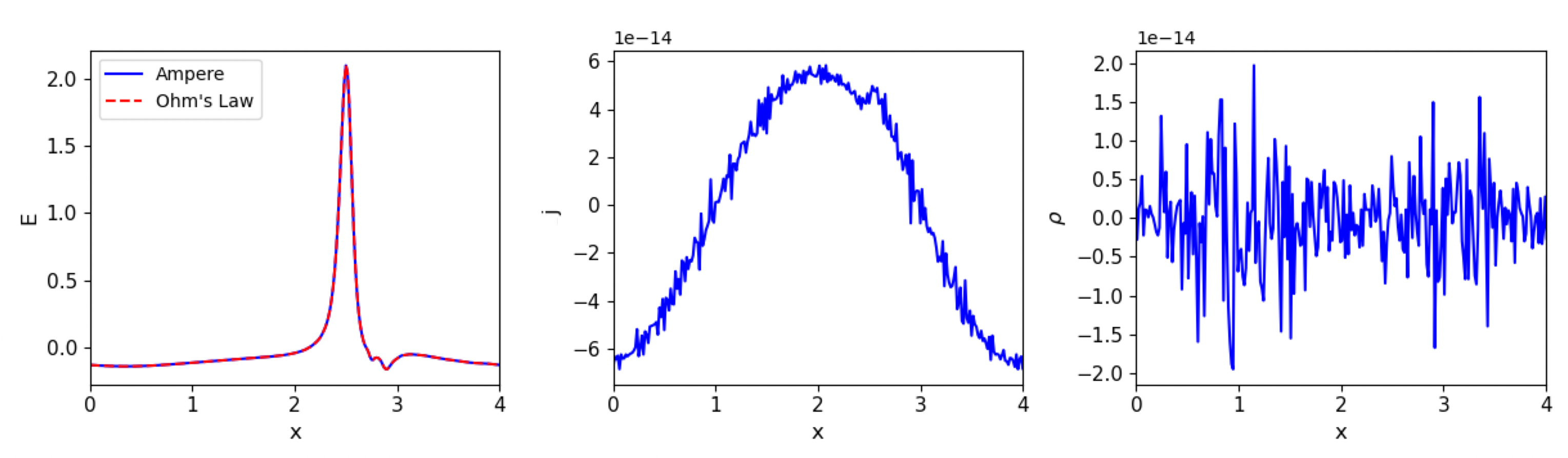}
    \par\end{centering}
    \caption{IASW: Comparison of the electric field from the evolution of the Amp{\`e}re's equation and the Ohm's law (left), the current density, $j = \epsilon\widetilde{j} = \epsilon \widehat{\widetilde{j}}$ for single-ion species (center), and the total charge density (right) for $\epsilon=10^{-7}$.}
    \label{fig:iasw_ap_property}
\end{figure}
As can be seen, the charge density and the currents are computed to be ${\cal O}\left(\epsilon^{2}\right)\sim10^{-14}$ as expected, and the Ohm's law acts as an excellent approximation to the exact electric field, demonstrating the capability of our formulation in recovering the asymptotic limiting solution guaranteed by the existence of discrete slow manifold. 

Finally, we demonstrate the performance of our nested iterative solver discussed in Section \ref{sec:solver}. In Figure \ref{fig:iasw_solver_performance}, we show the performance of the outer fixed point solver for converging the coupled kinetic-moment-field system and the QN solver for the moment-field subsystem with respect to $\Delta t \epsilon^{-1}$.
\begin{figure}[th]
    \begin{centering}
        \includegraphics[height=0.3\textwidth]{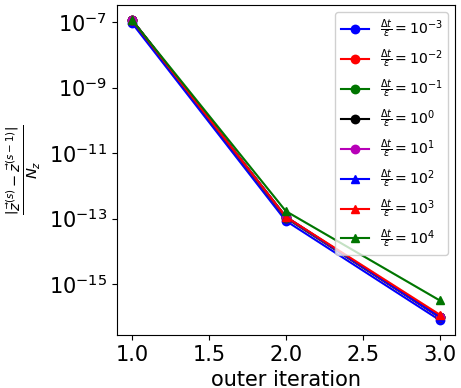}\includegraphics[height=0.3\textwidth]{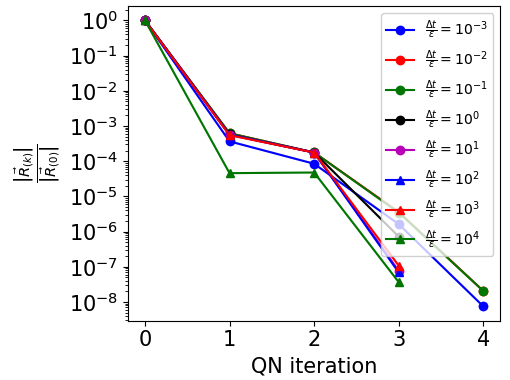}
    \par\end{centering}
    \caption{IASW: Performance of the outer fixed-point iterative solver (left) and the inner QN solver (right) with respect to $\Delta t/\epsilon$.
    \label{fig:iasw_solver_performance}}
\end{figure}
As can be seen, the performance of both the outer fixed-point iteration and the inner moment-field QN solver based on a quasi-linearization that satisfies the discrete slow manifold is largely insensitive to the ratio of $\Delta t \epsilon^{-1}$. A slight degradation in performance is seen for $\Delta t < \epsilon$ as expected; for non-vanishing $\epsilon$, $u_i \neq u_e$ and a small truncation error mismatch is introduced for $\widehat{\widetilde{j}}$ between Eq. \eqref{eq:1d_discrete_ampere} and \eqref{eq:qn_ampere}. 
%
%
%
%
%
%
\subsection{Multi-Ion Case}
\label{subsec:multi_ion_case}
Despite the lack of rigorous proof of the existence of discrete slow manifold for the multi-ion case, as discussed in Sec. \ref{subsubsec:discrete_qn_ap}, we demonstrate empirically the robustness of the proposed algorithm for a multi-ion problem in the quasi-neutral limit. For this problem, the parameters are chosen as $m_1 = 1$, $m_2 = 2$, $q_1 = 1$, $q_2 = 2$, $m_e = 1/1836$, $q_e = -1$, $L_x = 2\pi$, $\epsilon = 10^{-8}$, $w_{max} = 7$, $N_x = 128$, $N_w = 128$, $\Delta t = 10^{-3}$, $t_{max} = 5$, $\delta_{n_1} = \delta_{n_2} =  0.2$, $\widetilde{u}_1 = \widetilde{u}_2 = -1$, $\delta_{u_1} -\delta_{u_2} = -1$, $\widetilde{T}_1 = \widetilde{T}_2 = \widetilde{T}_e = 1$, $\delta_{T_1} = -\delta_{T_2} = \delta_{T_e} = 0.2$. We set the initial electron number density and flow speed so that quasi-neutrality and ambipolarity are exactly satisfied, $n^0_e = q_1 n^0_1 + q_2 n^0_2$ and $u^0_e = -\frac{q_1 n_1 + q_2 n_2}{q_e n^0_e}$. In Fig. \ref{fig:multiion_qn}, we show the quality of preservation of ambipolarity and quasi-neutrality at several times.
\begin{figure}[th]
    \begin{centering}
    \includegraphics[width=0.8\textwidth]{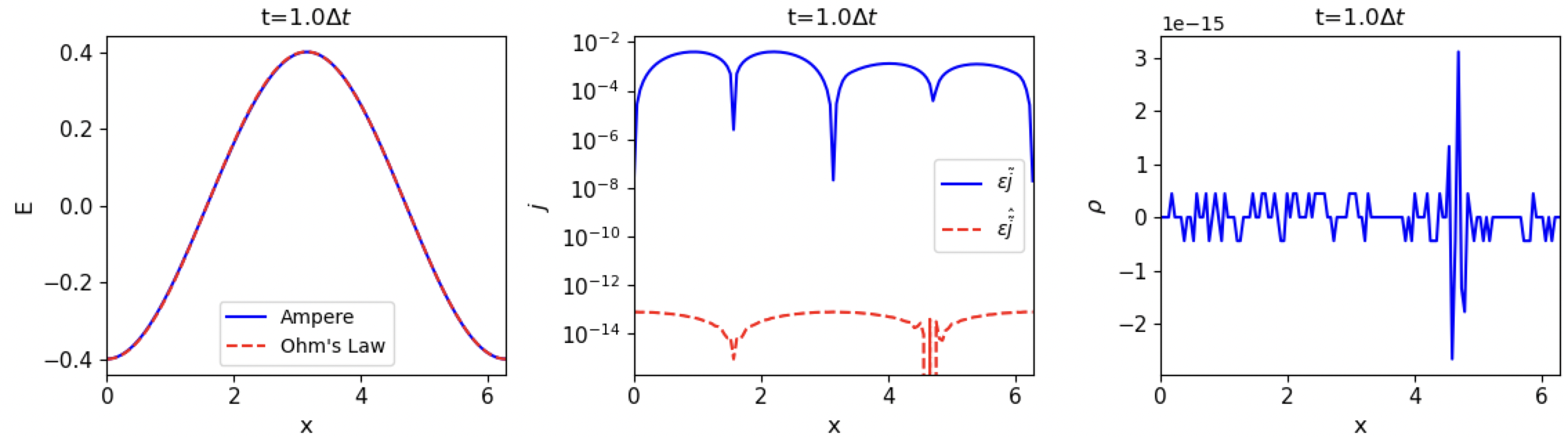}\\
    \includegraphics[width=0.8\textwidth]{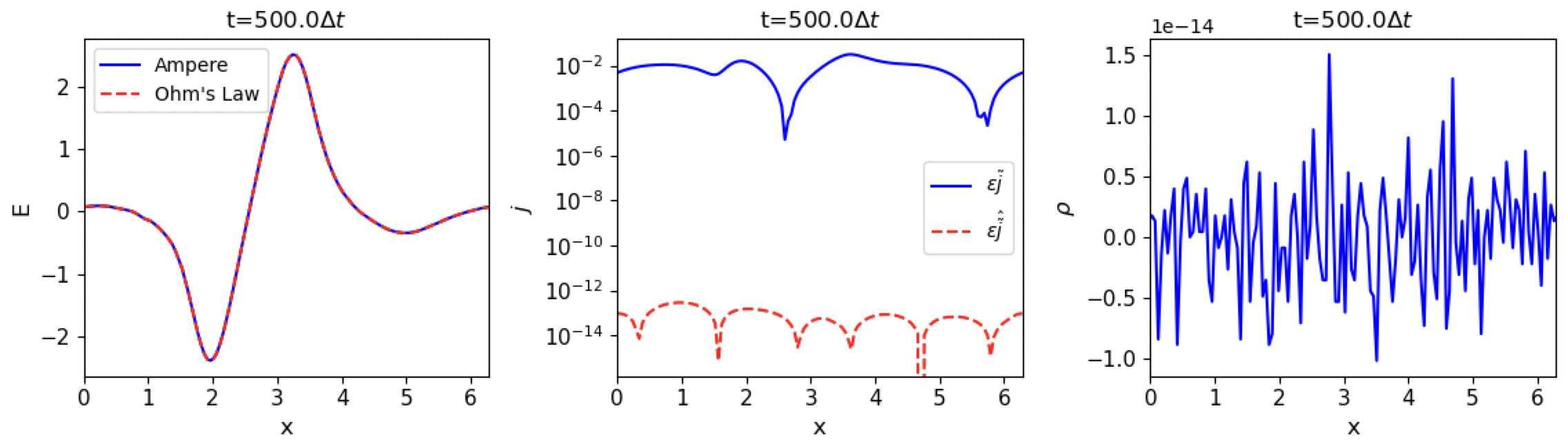}\\
    \includegraphics[width=0.8\textwidth]{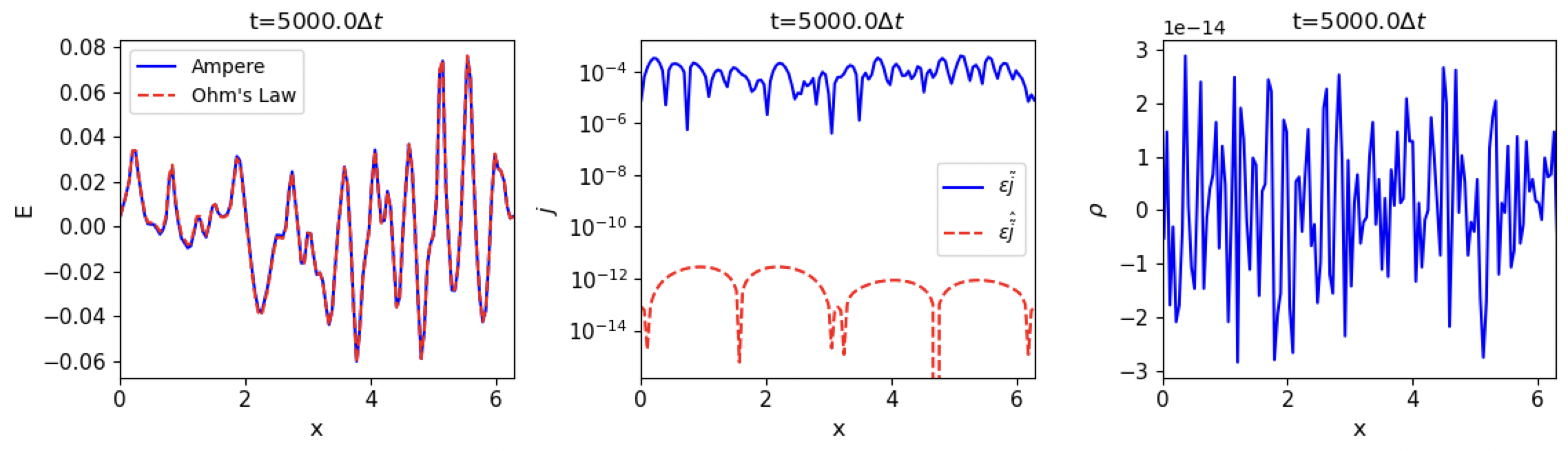}\\
    \par\end{centering}
    \caption{Multi-Ion Case: The comparison of electric field from Ampere's equation and Ohm's law (left column), the comparison of current from $\epsilon \widehat{\widetilde{j}}$ and $\epsilon \widetilde{j}$ measures (center column), and the total charge density (right column) at $t = \Delta t$ (top row), $t = 500\Delta t$ (middle row), and $t = 5000\Delta t$ (bottom row).}
    \label{fig:multiion_qn}
\end{figure}
As can be seen, quasi-neutrality is preserved to $\approx{\cal O}\left(\epsilon^2 \right)$, and Ohm's law acts as a good approximation for the electric field. In contrast, in the measure of $\epsilon \widetilde{j}_{l+\frac{1}{2}}$ ambipolarity is not rigorously enforced as expected in a multi-ion setting and discussed in Sec. \ref{subsubsec:discrete_qn_ap} while in the measure of $\epsilon \widehat{\widetilde{j}}$ it is satisfied to a much smaller value closer to ${\cal O}\left(\epsilon^2 \right)$, also as expected. In Fig. \ref{fig:multiion_pdf}, we show the evolution of the distribution function at various times.
\begin{figure}[th]
    \begin{centering}
    \includegraphics[width=0.8\textwidth]{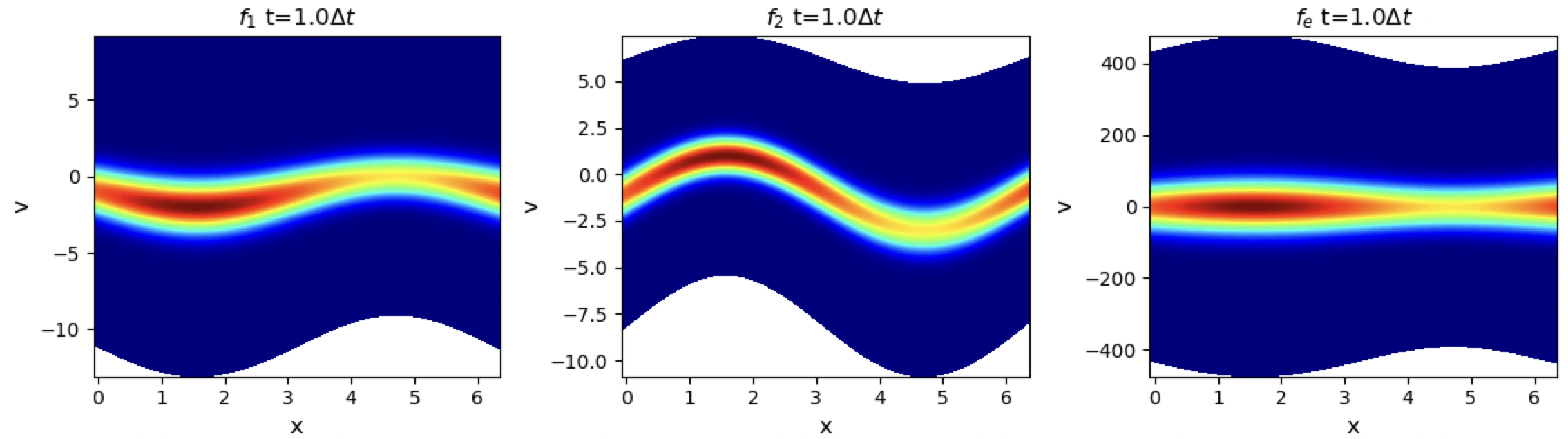}\\
    \includegraphics[width=0.8\textwidth]{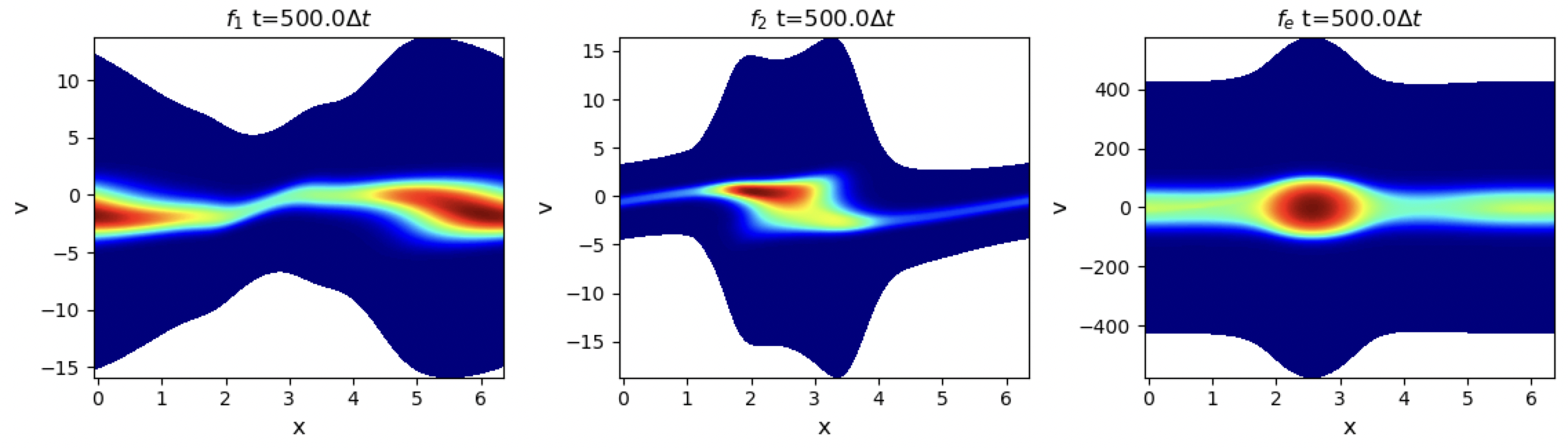}\\
    \includegraphics[width=0.8\textwidth]{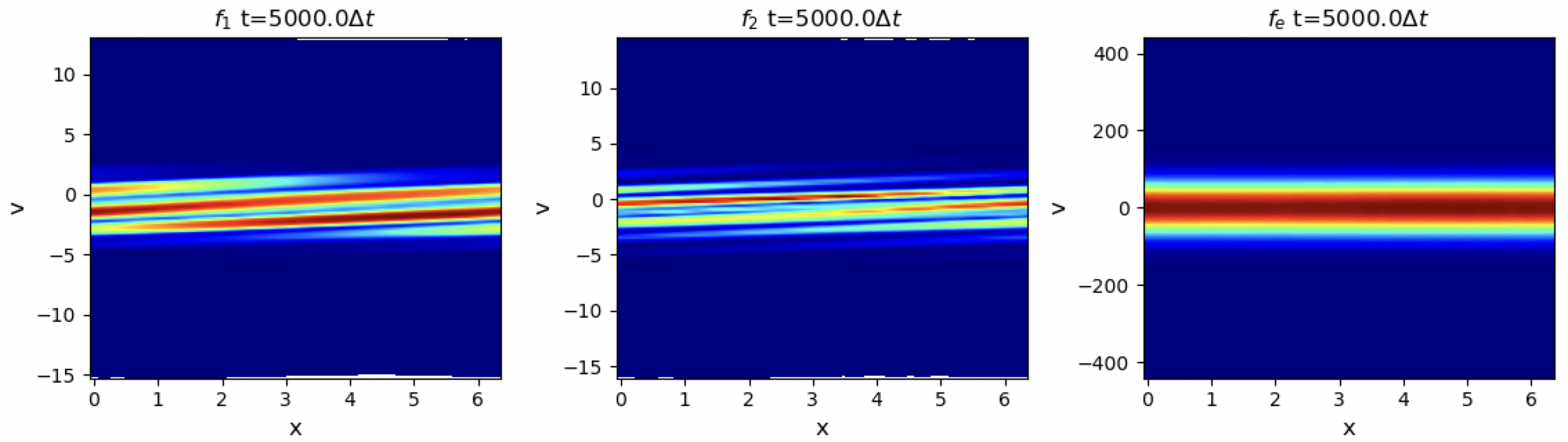}\\
    \par\end{centering}
    \caption{Multi-Ion Case: The comparison of the distribution function on the $\left\{x,v\right\}$ space for the two ions (left and center column) and electrons (right column) at $t = \Delta t$ (top row), $t = 500\Delta t$ (middle row), and $t = 5000\Delta t$ (bottom row).}
    \label{fig:multiion_pdf}
\end{figure}
Despite the fact that ambipolarity is only satisfied in a weak sense, solution remains regular, with the expected formation of the ion distribution function filamentation at late stages of time as the ion Landau damping dissipates the initial ion acoustic waves. Furthermore, we demonstrate in Fig. \ref{fig:multiion_conservation} that the discrete conservation principle and the Gauss law are upheld rigorously. 
\begin{figure}[th]
    \begin{centering}
    \includegraphics[width=0.6\textwidth]{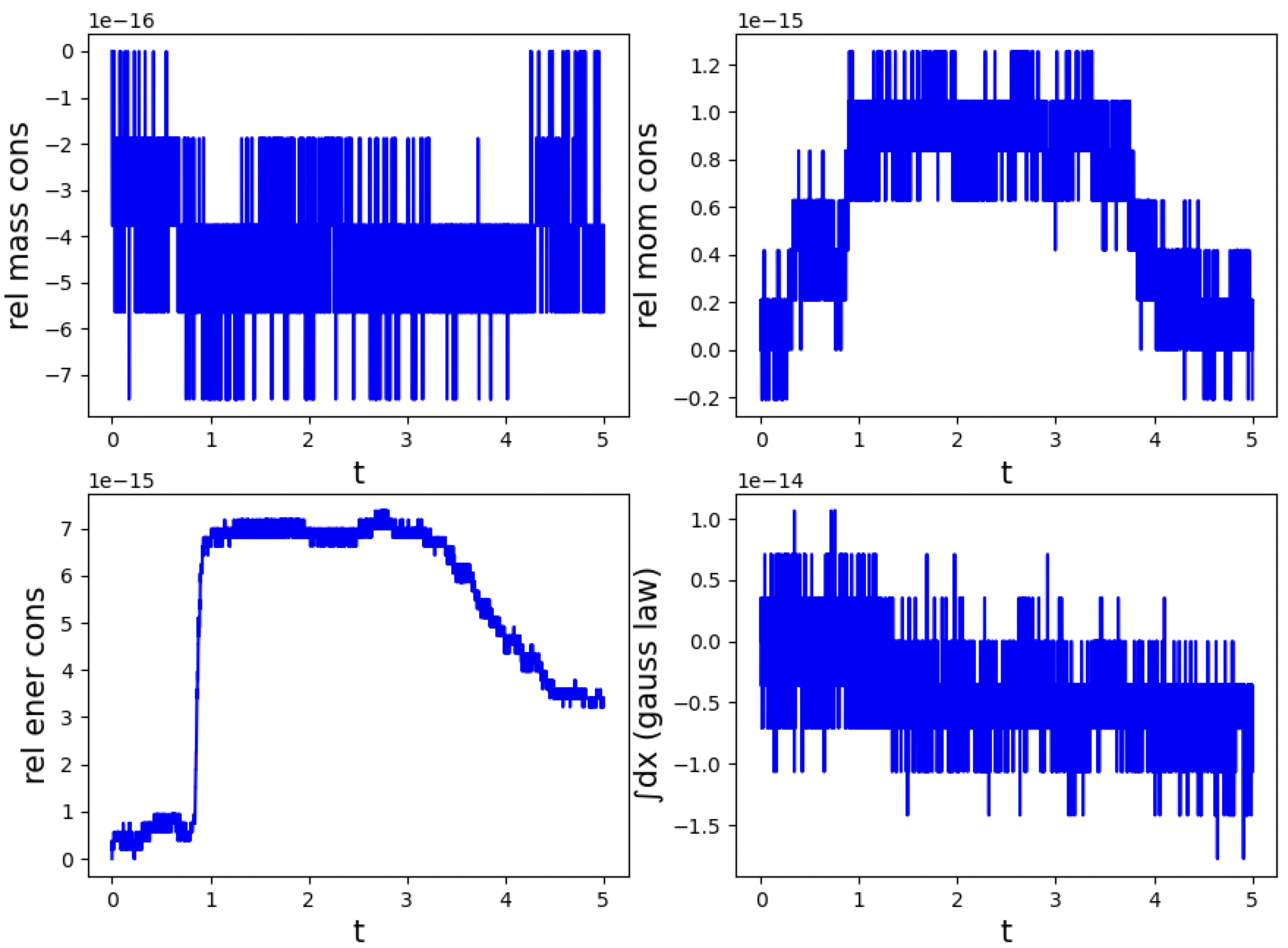}
    \par\end{centering}
    \caption{Multi-Ion Case: Relative error in the total mass (top left), momentum (top right), total energy (bottom left), and error in Gauss law (bottom right) as a function of time.}
    \label{fig:multiion_conservation}
\end{figure}
Finally, we show that the QN solver for the moment-field subsystem is robust with respect to $\Delta t \epsilon^{-1}$, despite the quasi-linearization (e.g., preconditioner) being based effectively on \emph{linear flux reconstruction} discussed in Proposition \ref{proposition:equilvalence_of_numerical_current_and_solution_current_1}, where it is not strictly equivalent with the multi-ion case that uses the kinetic flux discretization of $\widehat{\widetilde{j}}$; refer to Figure \ref{fig:multiion_solver_performance}. 
\begin{figure}[th]
    \begin{centering}
        \includegraphics[width=0.35\textwidth]{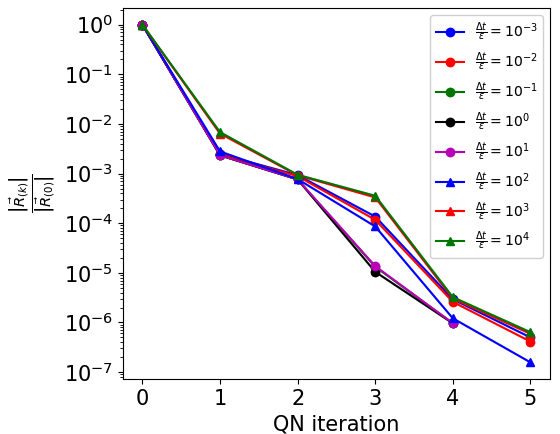}
    \par\end{centering}
    \caption{Mutli-Ion Case: Performance of the inner QN solver with respect to $\Delta t/\epsilon$.
    \label{fig:multiion_solver_performance}}
\end{figure}
The robust solver performance is attributed to the fact that the discrepencies between the residual function evaluations in Eqs. \eqref{eq:1d_discrete_continuity}-\eqref{eq:1d_discrete_ampere} and the quasi-linearization in Eqs. \eqref{eq:1d_quasi_linear_continuity}-\eqref{eq:qn_ampere} is on the order of the discretization truncation error, which is often small and the quasi- linearization is effective in approximating the full Jacobian for the purpose of the solver.

%
%
\section{Conclusions\label{sec:conclusions}}

In this study, we have developed a novel reformulation of the Vlasov-Amp{\`e}re equations that evolves the conditional distribution function of plasma particles. The reformulation relies on the transformation of the Vlasov equation using a sequence of velocity-space coordinate transformation and transformation of variables to separate the mass, momentum, and energy from the evolution of the distribution function. In this new formulation, the conserved quantities are evolved separately in the associated moment equations. Consequently, the symmetries for conservation, Gauss law involution, and quasi-neutrality asymptotics are separated from the kinetic equation and isolated into the moment-field subsystem. We have developed a fast-slow formulation for the new formulation based on a recent slow-manifold reduction theory to segregate fast electron time-scales out of the kinetic equations and into the easier to solve moment-field subsystem, which helps to expose the quasi-neutral limit in the discrete. The resulting system of equations was discretized with a compatible staggered finite differencing scheme, and we have demonstrated for the first time in the literature the simultaneous conservation of mass, momentum, and energy; preservation of the Gauss law; positivity of the distribution function; and the quasi-neutral limit for the vanishing parameter $\epsilon$. 

We note that, in addition to the conditional formulation, the staggered primitive formulation in the moment-field subsystem was key to achieve the simultaneous achievement of the above properties, as the staggering between the number density and the field is required to satisfy the Gauss law, while the quasi-neutral asymptotic preservation required collocation in the flow and the field, which is challenging to achieve simultaneously with collocated discretization schemes used in popular shock capturing Riemann solvers. The challenge of energy conservation in the staggered formulation was addressed by extending to plasmas, a strategy inspired by the fluid community that corrects the discrete internal energy equation through a Lagrange multiplier-like term. The resulting coupled nonlinear system of equations was solved efficiently using a block solver that iterates between the conditional Vlasov equation and the moment-field subsystem. The Quasi-Newton solver for the moment-field subsystem employs a linearization that is inspired by the slow manifold of the system and efficiently solves the system when $\epsilon\rightarrow0$ while recovering the quasi-neutral asymptotics. The advertised capabilities of the new formulation and the numerical methods were demonstrated on a series of canonical electrostatic benchmark problems including the multiscale ion acoustic shock wave problem.

The capability developed in this study is currently limited to first-order accuracy in $x$ and $t$ due to the particular choice of monotone-preserving staggered discretization for the moment-field subsystem and the choice of time integration of the coupled system. As such, to enhance the practical viability of the approach, a higher order scheme will be developed in the future. Furthermore, despite being able to take a large $\Delta t\epsilon^{-1}$, we are limited by the explicit CFL condition in the kinetic subsystem. To deal with the stiff advection time scales of the modified Vlasov equation, we have begun the development of a semi-Lagrangian capability of the new formulation, and the work will be reported in a follow-up manuscript.

%
%
\section*{Acknowledgments} 
W.T.T was supported by Triad National Security, LLC under contract 89233218CNA000001 and the DOE Office of Applied Scientific Computing Research (ASCR) through the Mathematical Multifaceted Integrated Capability Centers program. A.A  was supported by the National Science Foundation grant DMS-2111612.


%
%
\bibliographystyle{elsarticle-num}
\bibliography{references.bib}

%
%
\appendix
\section{SMART flux reconstruction}
\label{app:smart_discretization}
Consider a 1D, scalar advection equation of form,
\begin{equation}
    \label{eq:linear_advection}
    \frac{\partial \phi}{\partial t} + \frac{\partial}{\partial x} \left( a \phi \right) = 0.
\end{equation}
Here, $\phi\left(x,t\right)$ is the conserved scalar quantity and $a\left(x\right)$ is the advection constant, $x \in \left[0,L_x\right]$ and $t \in \left[0,t_{max}\right]$, and the spatial grid is $\mathfrak{x} = \left\{\left( i - 0.5 \right) \Delta\mathfrak{x} \right\}_{i=1}^{N_x}$, and $\Delta\mathfrak{x} = \frac{L_x}{N_x}$. Using a conservative finite differencing formulation, we semi-discretize the equation on the grid $\mathfrak{x}$ as $\vec{\phi} = \left\{\phi_1 ,\cdots,\phi_{N_x} \right\}$, and the equation on cell $i$ as:
\begin{equation}
    \label{eq:discrete_1d_advection}
    \frac{\partial \phi_i}{\partial t} + \frac{\widehat{a}_{i+1/2} \widehat{\phi}_{i+1/2} - \widehat{a}_{i-1/2}\widehat{\phi}_{i-1/2}}{\Delta\mathfrak{x}}
    = 0.
\end{equation}
Here, $\widehat{a}_{i\pm 1/2}$ is some cell reconstruction of the advection velocity on cell interfaces $i\pm1/2$, and $\widehat{\phi}_{i\pm 1/2}$ is the cell reconstruction of $\phi$, which is given using the SMART discretization \cite{gaskell_ijnmf_1988_smart} for $\widehat{a}_{i+1/2} \ge 0 $ as
\begin{equation}
    \label{eq:smart_discretization_positive_flow}
    \widehat{\phi}_{i+1/2} \equiv 
    SMART\left( \widehat{a}_{i+1/2}, \vec{\phi} \right)
    = F_{med}\left(\phi_{i}, \phi^*_4, \phi^*_1  \right) \left| \right. \widehat{a}_{i+1/2} \; \ge \; 0
\end{equation}
\begin{equation}
    \phi^*_{4} = F_{med}\left(\phi_{i},\phi^*_2,\phi^*_3\right)
\end{equation}
\begin{equation}
    \phi^*_3 = \beta_2 \phi_i + \left(1 - \beta_2 \right)\phi_{i+1}
\end{equation}
\begin{equation}
    \phi^*_2 = \beta_1 \phi_i + \left(1 -\beta_1\right) \phi_{i-1}
\end{equation}
\begin{equation}
    \phi^*_1 = \frac{\phi_{i+1} + \phi_{i}}{2} - \frac{\Delta_{\phi,i+1/2} \Delta {\mathfrak x}^2}{8}
\end{equation}
\begin{equation}
    \Delta_{\phi,i+1/2} =  \frac{\phi_{i+1} - 2\phi_{i} + \phi_{i-1}}{{\Delta \mathfrak{x}^2}}
\end{equation}
with $F_{med}\left(a,b,c\right) = a + b + c - \max\left(a,b,c\right) - \min\left(a,b,c\right)$ the median function and for $\widehat{a}_{i+1/2} < 0$ as
\begin{equation}
    \label{eq:smart_discretization_negative_flow}
    \widehat{\phi}_{i+1/2} \equiv 
    SMART\left( \widehat{a}_{i+1/2}, \vec{\phi} \right)
    = F_{med}\left(\phi_{i+1}, \phi^*_4, \phi^*_1  \right) \left| \right. \widehat{a}_{i+1/2} \; < \; 0
\end{equation}
\begin{equation}
    \phi^*_{4} = F_{med}\left(\phi_{i+1},\phi^*_2,\phi^*_3\right)
\end{equation}
\begin{equation}
    \phi^*_3 = \beta_2 \phi_{i+1} + \left(1 - \beta_2 \right)\phi_{i}
\end{equation}
\begin{equation}
    \phi^*_2 = \beta_1 \phi_{i+1} + \left(1 -\beta_1\right) \phi_{i+2}
\end{equation}
\begin{equation}
    \phi^*_1 = \frac{\phi_{i+1} + \phi_{i}}{2} - \frac{\Delta_{\phi,i+1/2} \Delta {\mathfrak x}^2}{8}
\end{equation}
\begin{equation}
    \Delta_{\phi,i+1/2} =  \frac{\phi_{i+2} - 2\phi_{i+1} + \phi_{i}}{{\Delta \mathfrak{x}^2}}
\end{equation}
Here, $\beta_1 = \frac{3}{2}$ and $\beta_2 = \frac{1}{2}$ are coefficients for the SMART reconstruction that weight the up-wind and down-wind discretization schemes.

%
%
\end{document}